\documentclass[11pt,letterpaper]{article}
\usepackage{xcolor,ifpdf,latexsym,graphicx,url,amsmath,amsfonts,subcaption,float,enumitem,wrapfig}

\usepackage[font=small,labelfont=bf]{caption}

\usepackage{chngcntr}
\counterwithin{figure}{section}

\usepackage{thmtools,thm-restate}

\usepackage{hyperref}

\advance \topmargin by -\headheight
\advance \topmargin by -\headsep
\evensidemargin \oddsidemargin
{\makeatletter \gdef\fps@figure{!htbp}}

\let\realbfseries=\bfseries
\def\bfseries{\realbfseries\boldmath}

\newtheorem{theorem}{Theorem}[section]
\newtheorem{lemma}[theorem]{Lemma}
\newtheorem{corollary}[theorem]{Corollary}
\newtheorem{problem}[theorem]{Problem}
\newtheorem{definition}[theorem]{Definition}

{\makeatletter
 \gdef\xxxmark{%
   \expandafter\ifx\csname @mpargs\endcsname\relax 
     \expandafter\ifx\csname @captype\endcsname\relax 
       \marginpar{xxx}
     \else
       xxx 
     \fi
   \else
     xxx 
   \fi}
 \gdef\xxx{\@ifnextchar[\xxx@lab\xxx@nolab}
 \long\gdef\xxx@lab[#1]#2{\textbf{[\xxxmark #2 ---{\sc #1}]}}
 \long\gdef\xxx@nolab#1{\textbf{[\xxxmark #1]}}
 \long\gdef\xxx@lab[#1]#2{}\long\gdef\xxx@nolab#1{}%
}

\makeatletter
\def\GrabProofArgument[#1]{ #1: \egroup\ignorespaces}
\def\proof{\noindent\textbf\bgroup Proof%
           \@ifnextchar[{\GrabProofArgument}{: \egroup\ignorespaces}}

\makeatother

\newenvironment{proofsketch}{\begin{proof}[sketch]}{\end{proof}}

\newif\ifabstract
\abstracttrue
\abstractfalse
\newif\iffull
\ifabstract \fullfalse \else \fulltrue \fi

\newcounter{section-preserve}
\newcounter{theorem-preserve}
\newcommand{\blank}[1]{}
\newtoks\magicAppendix
\magicAppendix={}
\newtoks\magictoks
\newif\iflater
\laterfalse
\long\def\later#1{\magictoks={#1}%
  \edef\magictodo{\noexpand\magicAppendix={\the\magicAppendix \par
    \the\magictoks%
  }}
  \magictodo}
\long\def\both#1{\magictoks={#1}%
  \edef\magictodo{\noexpand\magicAppendix={\the\magicAppendix \par
    \noexpand\setcounter{theorem-preserve}{\noexpand\arabic{theorem}}%
    \noexpand\setcounter{theorem}{\arabic{theorem}}%
    \noexpand\setcounter{section-preserve}{\noexpand\arabic{section}}%
    \noexpand\setcounter{section}{\arabic{section}}%
    \noexpand\let\noexpand\oldsection=\noexpand\thesection
    \noexpand\def\noexpand\thesection{\thesection}
    \noexpand\let\noexpand\oldlabel=\noexpand\label
    \noexpand\let\noexpand\label=\noexpand\blank
    \the\magictoks%
    \noexpand\setcounter{theorem}{\noexpand\arabic{theorem-preserve}}%
    \noexpand\setcounter{section}{\noexpand\arabic{section-preserve}}%
    \noexpand\let\noexpand\thesection=\noexpand\oldsection
    \noexpand\let\noexpand\label=\noexpand\oldlabel
  }}
  \magictodo
  \the\magictoks}
\def\magicappendix{\latertrue \the\magicAppendix}

\def\abstractlater#1{\ifabstract\later{#1}\fi}

\iffull
  \long\def\both#1{#1}
  \let\later=\both
  \def\magicappendix{}
\fi

\newcommand{\ccNP}{\textrm{\textsc{NP}}}
\newcommand{\ccPSPACE}{\textrm{\textsc{PSPACE}}}

\newcommand{\pushfight}{Push~Fight}

\definecolor{jeffrey-background}{HTML}{3C3F41}
\definecolor{jeffrey-foreground}{HTML}{CCCCCC}

\title{Computational Complexity of Generalized \pushfight{}}
\author{
Jeffrey Bosboom\thanks{MIT Computer Science and Artificial Intelligence
  Laboratory, 32 Vassar Street, Cambridge, MA 02139, USA,
  \protect\url{{jbosboom,edemaine}@mit.edu}, \protect\url{mrudoy@gmail.com}}\and
Erik D. Demaine\footnotemark[1]\and
Mikhail Rudoy\footnotemark[1]
  \thanks{Now at Google Inc.}
} 
\date{}

\begin{document}

\maketitle

\begin{abstract}
  We analyze the computational complexity of optimally playing the
  two-player board game Push Fight, generalized to an arbitrary board
  and number of pieces.  We prove that the game is PSPACE-hard to decide
  who will win from a given position, even for simple (almost rectangular)
  hole-free boards.
  We also analyze the \emph{mate-in-1} problem: can the player
  win in a single turn?  One turn in Push Fight consists of up to two
  ``moves'' followed by a mandatory ``push''.  With these rules, or
  generalizing the number of allowed moves to any constant, we show mate-in-1
  can be solved in polynomial time.  If, however, the number of moves per turn
  is part of the input, the problem becomes NP-complete.  On the other hand,
  without any limit on the number of moves per turn, the problem becomes
  polynomially solvable again.
\end{abstract}

\section{Introduction}


\begin{wrapfigure}{r}{2.5in}
    \centering
    \vspace*{-8ex}
    \includegraphics[width=\linewidth]{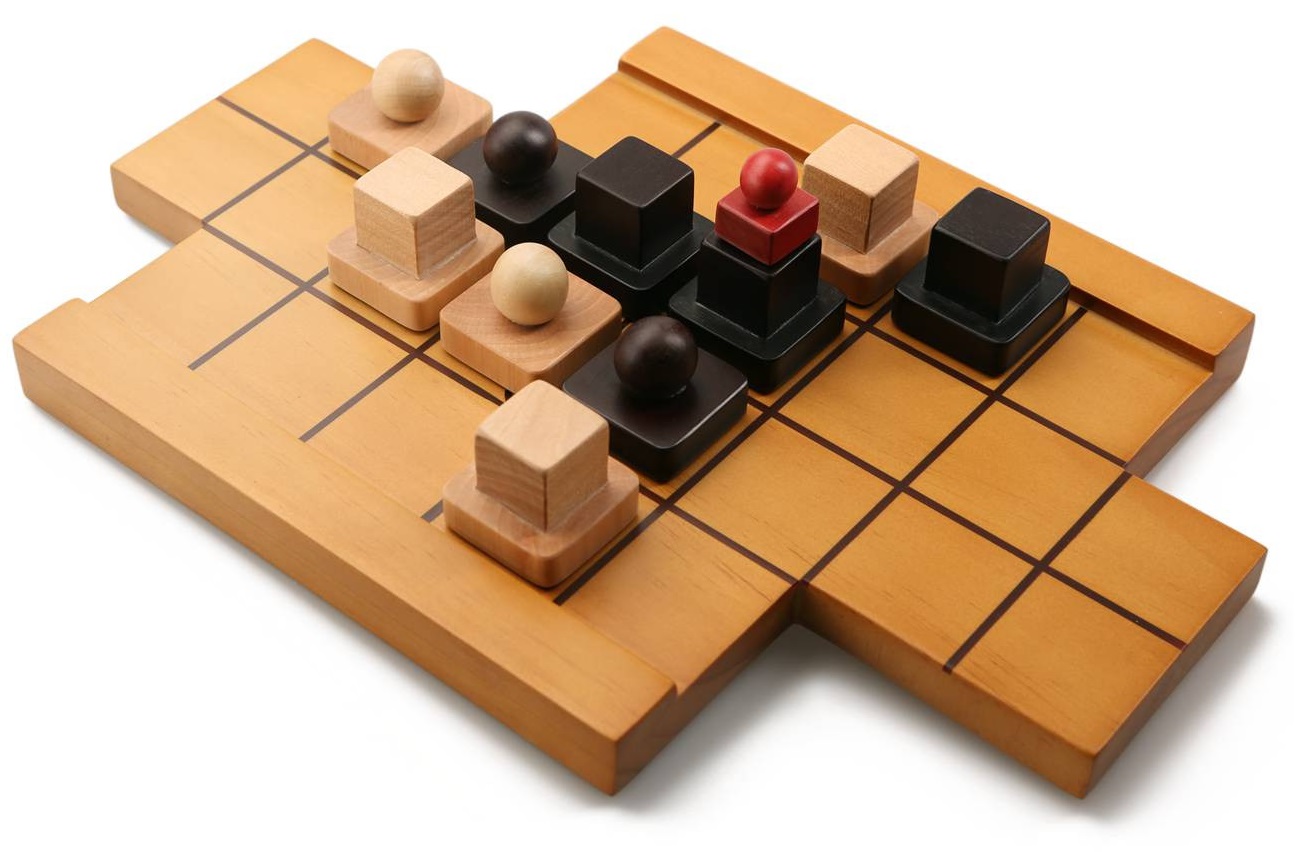}    

    \vspace*{-2ex}
    \caption{A \pushfight{} game in progress. Photo by Brettco, Inc., used with permission.}
    \vspace*{-2ex}
    \label{figure:game_board_image}
\end{wrapfigure}

\pushfight{}~\cite{pushfight} is a two-player board game,
invented by Brett Picotte around 1990, popularized by Penny Arcade in
2012 \cite{penny-arcade-1}, and briefly published by Penny Arcade in 2015
\cite{penny-arcade-2}.
Players take turns moving and pushing pieces on a square grid until a piece
gets pushed off the board or a player is unable to push on their move.
Figure~\ref{figure:game_board_image} shows a \pushfight{} game in progress,
and Section~\ref{section:rules} details the rules.


In this paper, we study the computational complexity of optimal play in
\pushfight{}, generalized to an arbitrary board and number of pieces, from
two perspectives:
\begin{enumerate}
\item \textbf{Who wins?} The typical complexity-of-games problem is to
  determine which player wins from a given game configuration.
\item \textbf{Mate-in-1}: Can the current player win \emph{in a single turn}?
\end{enumerate}

\begin{table}
  \centering
  \begin{tabular}{|l|l|l|}
    \cline{2-3}
    \multicolumn{1}{l|}{}
                       & \multicolumn{2}{|l|}{Computational complexity of\dots}
    \\ \hline
    \bf Moves per turn & \bf Mate-in-1& \bf Who wins?
    \\ \hline
    $\leq 2$           & P            & PSPACE-hard, in EXPTIME
    \\
    $\leq c$ constant  & P            & open
    \\
    $\leq k$ input     & NP-complete  & open
    \\
    unlimited          & P            & open
    \\ \hline
  \end{tabular}
  \caption{Summary of our results.}
  \label{results}
\end{table}

Table~\ref{results} summarizes our results.

Generalized \pushfight{} is a two-player game played on a polynomially bounded
board for a potentially exponential number of moves, so we conjecture
the ``who wins?''\ decision problem to be EXPTIME-complete,
as with Checkers~\cite{checkers} and Chess~\cite{chess}.
(Certainly the problem is in EXPTIME, by building the game tree.)
In Section~\ref{section:pspace},
we prove that the problem is at least PSPACE-hard,
using a proof patterned after the NP-hardness proof of Push-$*$ \cite{push}.
Our proof uses a simple, nearly rectangular board, in the spirit of the
original game; in particular, the board we use is hole-free and $x$-monotone
(see Figure~\ref{figure:pspace_reduction_example}).
It remains open whether \pushfight{} is in PSPACE, EXPTIME-hard,
or somewhere in between.

Our mate-in-1 results are perhaps most intriguing, showing a wide variability
according to whether and how we generalize the ``up to two moves per turn''
rule in \pushfight{}.  If we leave the rule as is, or generalize to
``up to $c$ moves per turn'' where $c$ is a fixed constant (part of the
problem definition), then we show that the mate-in-1 problem is in P,
i.e., can be solved in polynomial time.  However, if we generalize the
rule to ``up to $k$ moves per turn'' where $k$ is part of the input,
then we show that the mate-in-1 problem becomes NP-complete.
On the other hand, if we remove the limit on the number of moves per turn,
then we show that the mate-in-1 problem is in P again.
Section~\ref{section:mate_in_one} proves these results.

The mate-in-1 problem has been studied previously for other board games.
The earliest result is that mate-in-1 Checkers is in P,
even though a single turn can involve a long sequence of jumps
\cite{checkers-matein1}.
On the other hand, Phutball is a board game also featuring a sequence
of jumps in each turn, yet its mate-in-1 problem is NP-complete
\cite{Demaine-Demaine-Eppstein-2002}.

\xxx{For a general audience (i.e., the \pushfight{} designer), we may want to explain why generalization is necessary.}

\section{Rules}
\label{section:rules}

The original \pushfight{} board is an oddly shaped square grid containing $26$ squares; see Figure~\ref{figure:original_board}. Part of the boundary of this board has \emph{side rails} which prevent pieces from being pushed off across those edges. We generalize \pushfight{} by considering arbitrary polyomino boards, with each boundary edge possibly having a side rail.

\ifabstract
\begin{figure}
  \centering
  \begin{minipage}{0.46\textwidth}
    \centering
    \includegraphics[scale=.3]{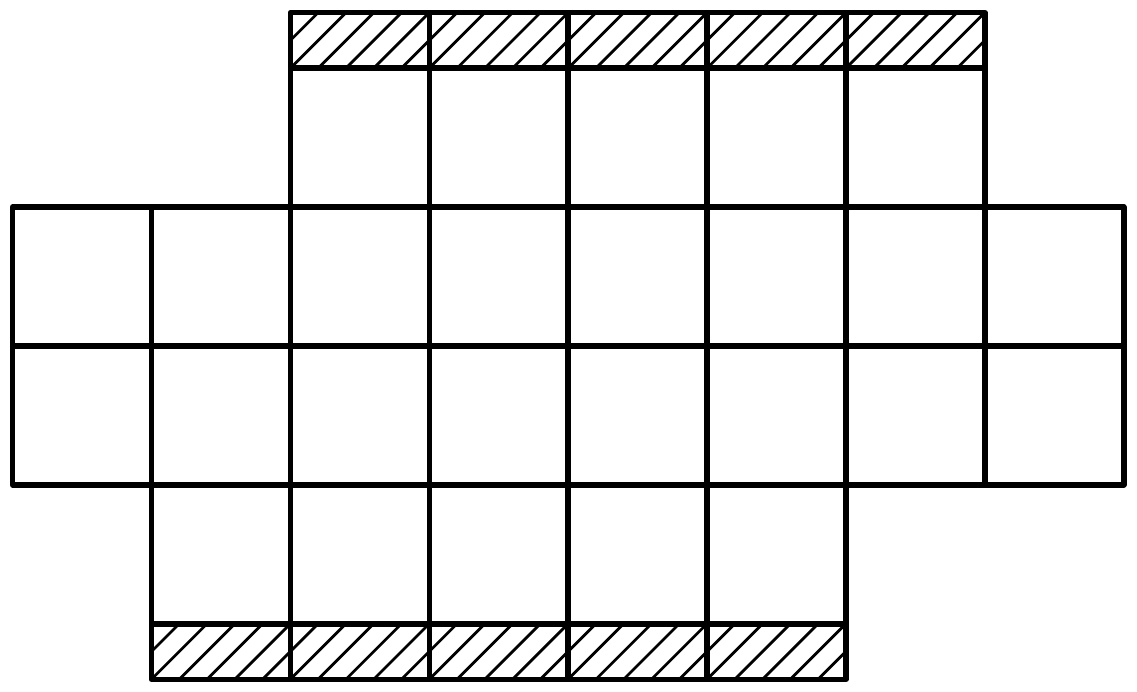}    
    \vspace*{-1ex}
    \caption{Original \pushfight{} board. Shaded regions represent side rails.}
    \label{figure:original_board}

    \medskip

    \centering
    \includegraphics[scale=.3]{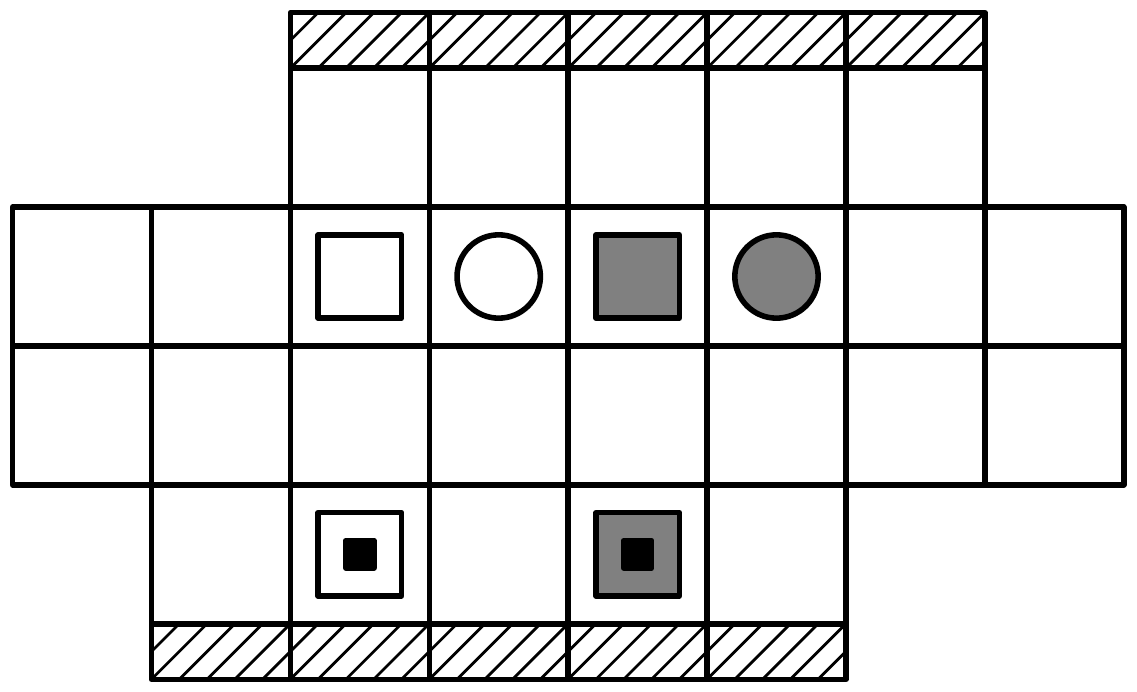}    
    \vspace*{-1ex}
    \caption{Our notation for pieces, in reading order: a white king, a white pawn, a black king, a black pawn; and white and black anchored kings (in an actual game, there is only one anchor).}
    \label{figure:piece_types}
  \end{minipage}\hfill
  \begin{minipage}{0.5\textwidth}
    \centering
    \hfill
    \raisebox{-.5\height}{\includegraphics[scale=.3]{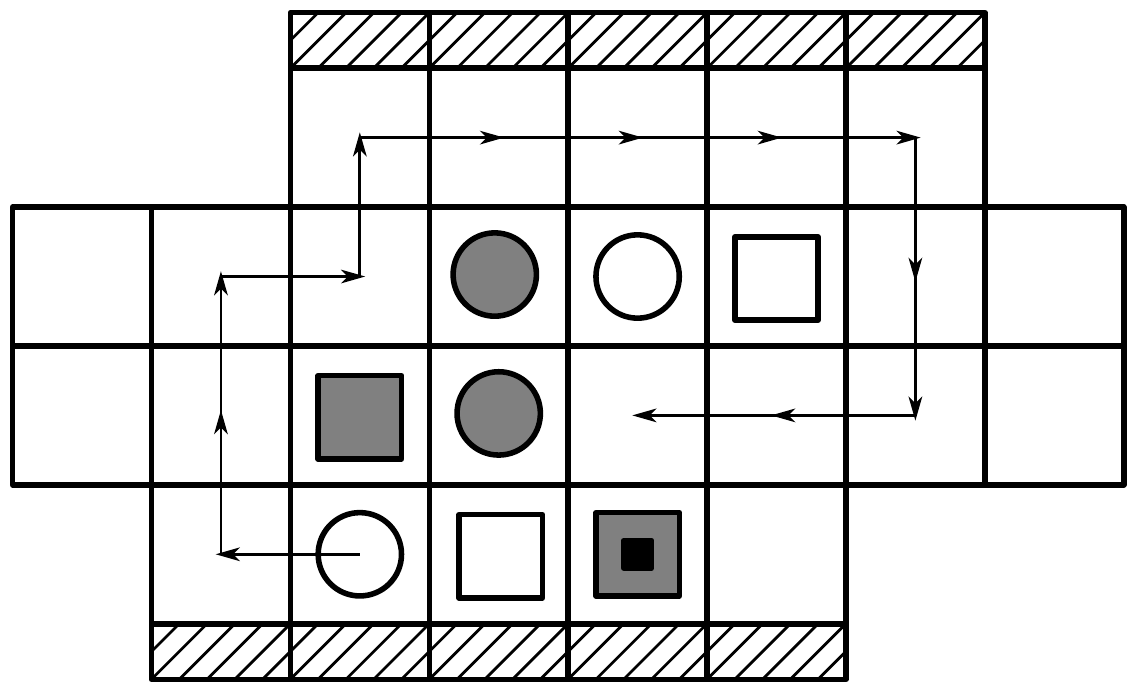}}
    \hfill\raisebox{-.5\height}{\scalebox{2}{$\to$}}\hfill
    \raisebox{-.5\height}{\includegraphics[scale=.3]{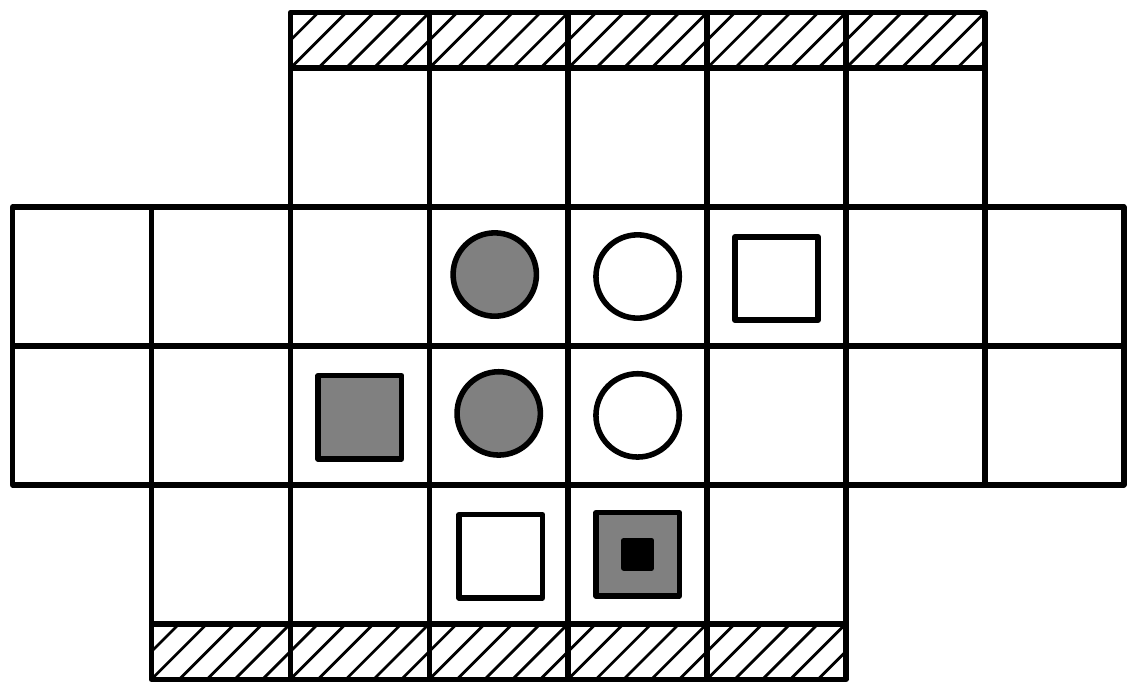}}
    \hfill\hfill
    \caption{An example move.}
    \label{figure:move_example}

    \bigskip
    \bigskip
    \bigskip

    \centering
    \hfill
    \raisebox{-.5\height}{\includegraphics[scale=.3]{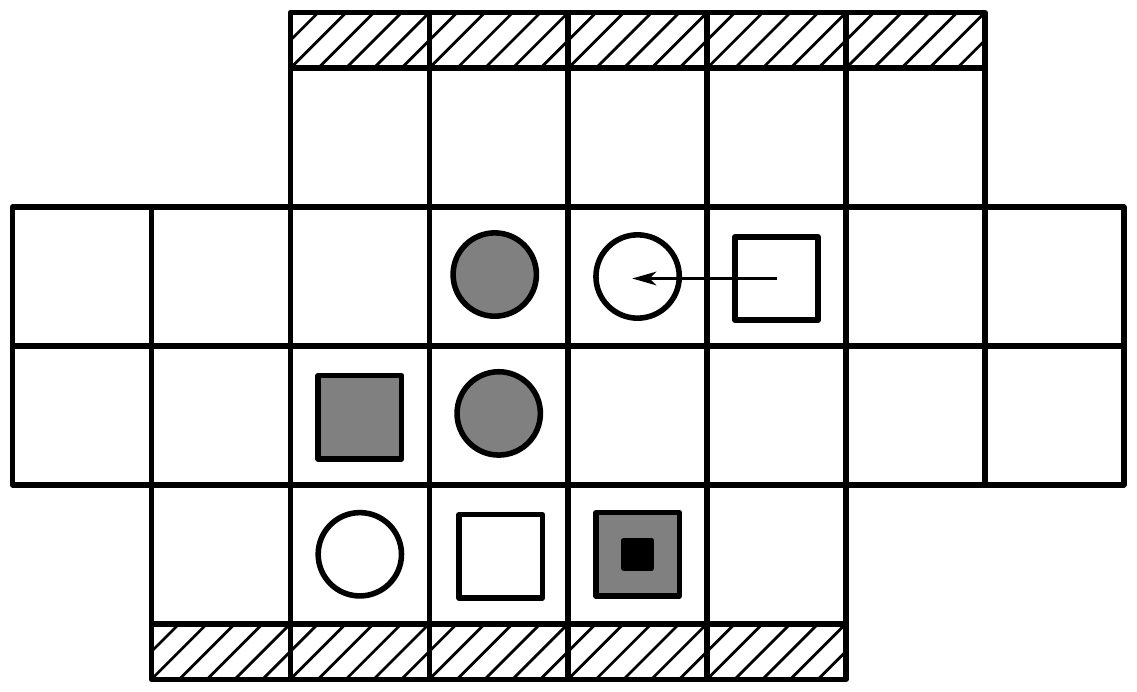}}
    \hfill\raisebox{-.5\height}{\scalebox{2}{$\to$}}\hfill
    \raisebox{-.5\height}{\includegraphics[scale=.3]{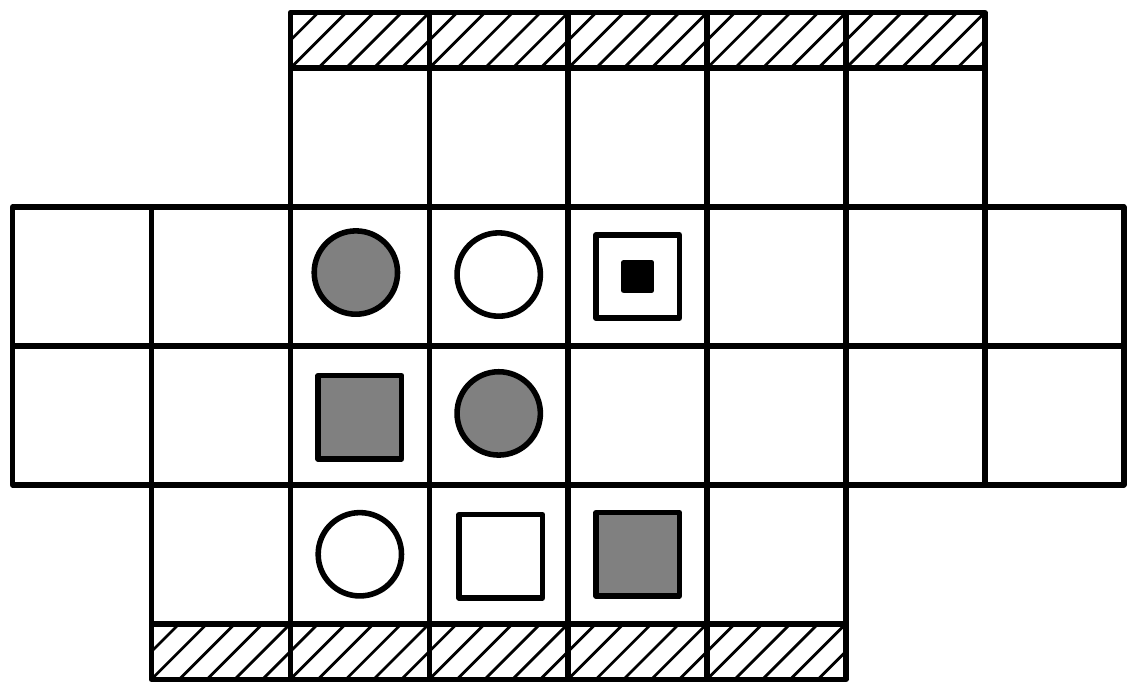}}
    \hfill\hfill
    \caption{An example push.}
    \label{figure:push_example}
  \end{minipage}
\end{figure}
\fi

\iffull

\begin{figure}
  \centering
  \begin{minipage}{0.38\textwidth}
    \centering
    \includegraphics[scale=.5]{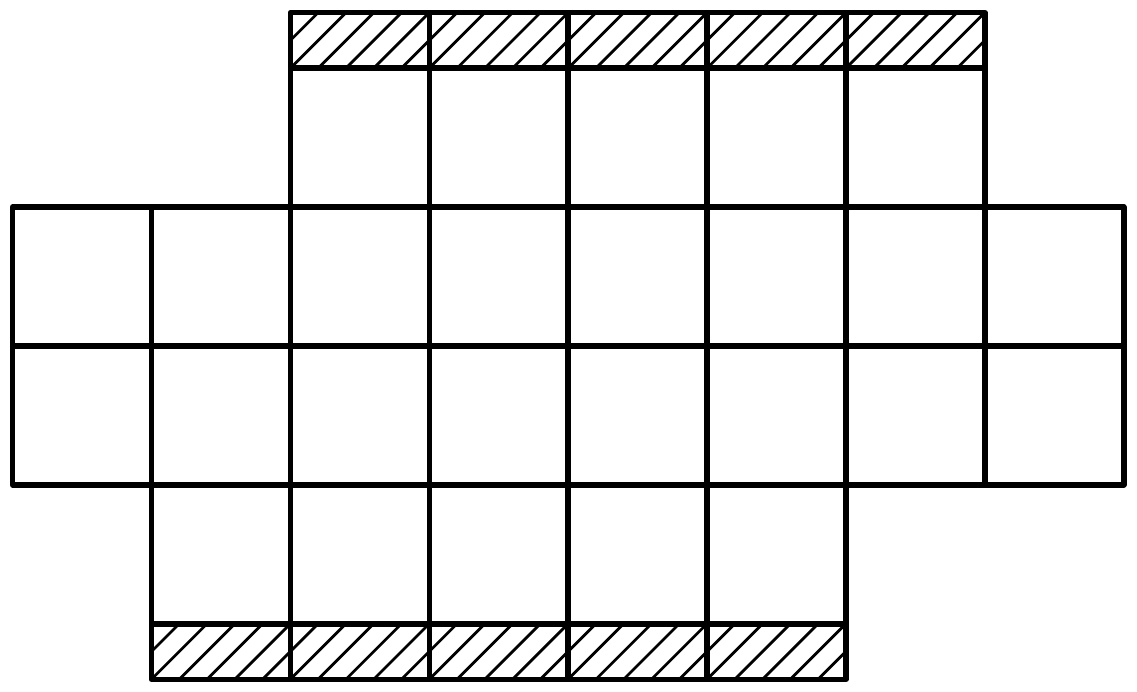}    
    \caption{The original \pushfight{} board. The shaded regions represent side rails.}
    \label{figure:original_board}
  \end{minipage}\hfill
  \begin{minipage}{0.58\textwidth}
    \centering
    \includegraphics[scale=.5]{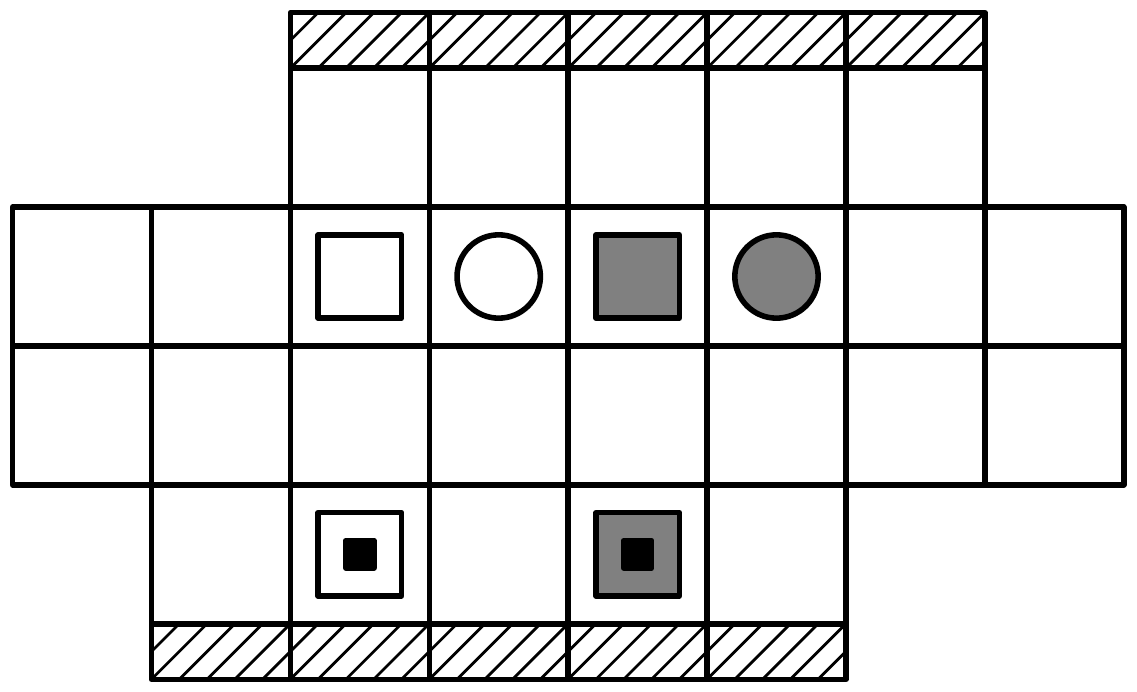}    
    \caption{Our notation for pieces. From left to right in the middle of the second row: a white king, a white pawn, a black king, and a black pawn. The bottom row shows an anchored king of each color (in an actual game, there is only one anchor).}
    \label{figure:piece_types}
  \end{minipage}
\end{figure}
\fi

\pushfight{} is played with two types of pieces, each of which takes up a square of the board: \emph{pawns} (drawn as circles) and \emph{kings} (drawn as squares). Each piece is colored either black or white, denoting which player the piece belongs to. Standard \pushfight{} is played with three kings and two pawns per player. Additionally, there is a single \emph{anchor} that is placed on top of a king after it pushes (but is never placed directly on the board). Figure~\ref{figure:piece_types} shows our notation for the pieces. 


\pushfight{} gameplay consists of the two players alternating \emph{turns}. During a player's turn, the player makes up to two optional \emph{moves} followed by a mandatory \emph{push}.

To make a move, a player moves one of their pieces along a simple path of orthogonally adjacent empty squares; see Figure~\ref{figure:move_example}.

\iffull
\begin{figure}[t]
    \centering
    \hfill
    \raisebox{-.5\height}{\includegraphics[scale=.5]{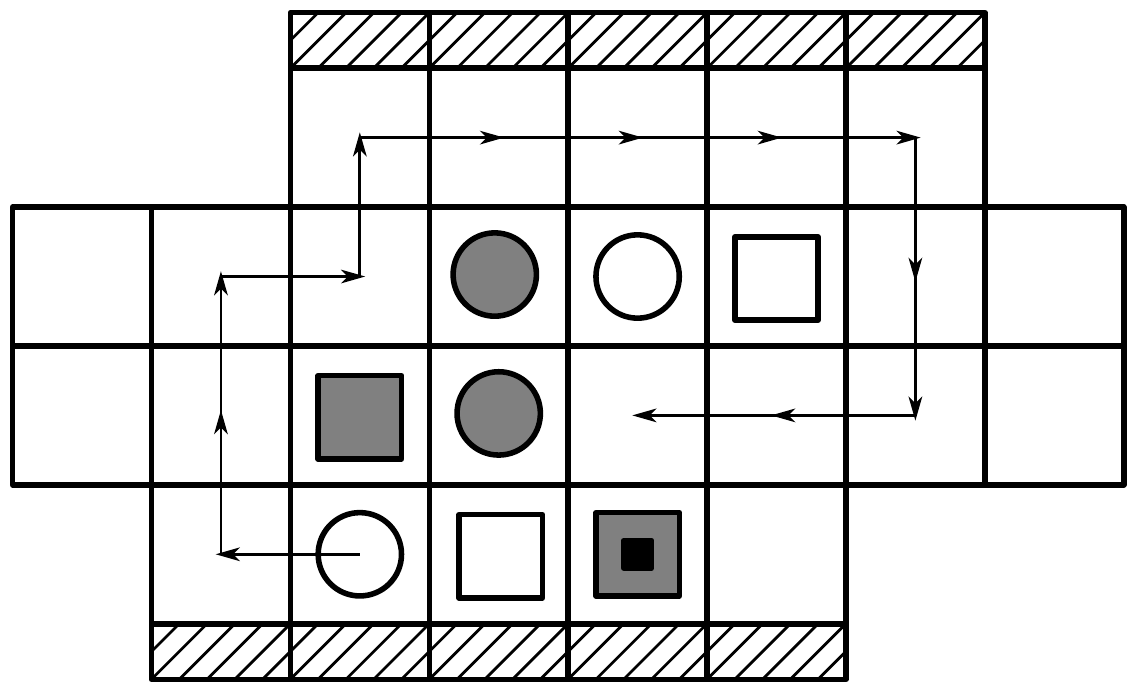}}
    \hfill\raisebox{-.5\height}{\scalebox{2}{$\to$}}\hfill
    \raisebox{-.5\height}{\includegraphics[scale=.5]{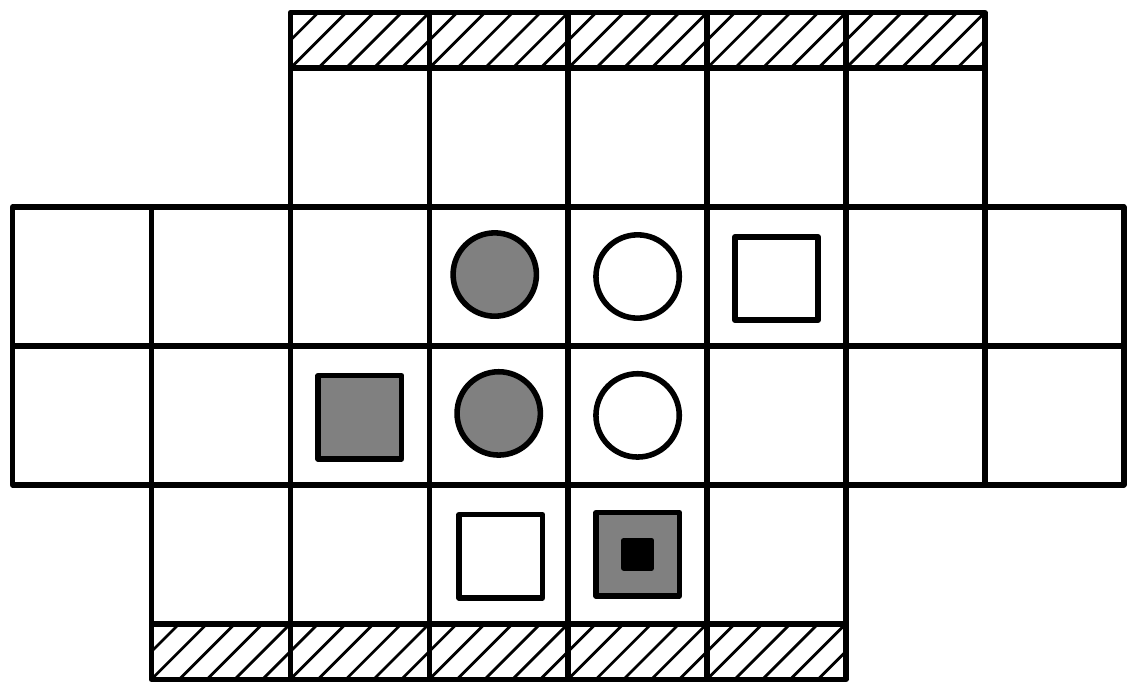}}
    \hfill\hfill
    \caption{An example move.}
    \label{figure:move_example}
\end{figure}
\fi

To push, a player moves one of their kings into an occupied adjacent square. The piece occupying that square is pushed one square in the same direction, and this continues recursively until a piece is pushed into an unoccupied square or off the board. If this process would push a piece through a side rail, or would push the anchored king, the push cannot be made. Pushes always move at least one other piece. When the push is complete, the pushing king is anchored (the anchor is placed on top of that king). Figure~\ref{figure:push_example} shows a valid push.

\iffull
\begin{figure}[t]
    \centering
    \hfill
    \raisebox{-.5\height}{\includegraphics[scale=.5]{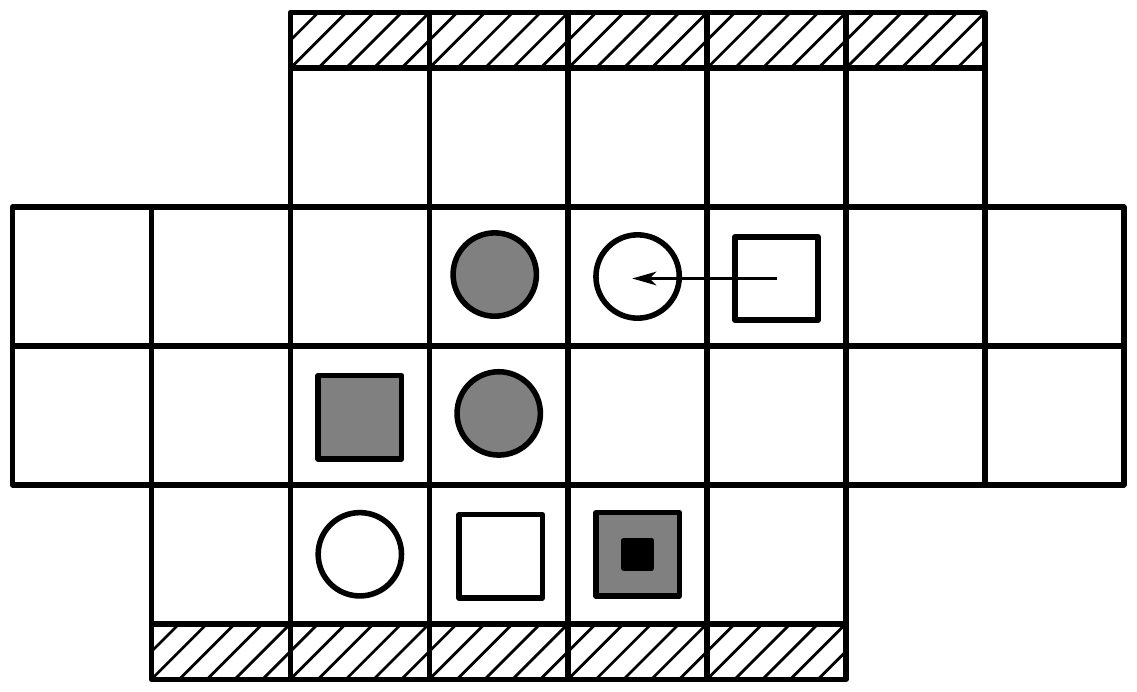}}
    \hfill\raisebox{-.5\height}{\scalebox{2}{$\to$}}\hfill
    \raisebox{-.5\height}{\includegraphics[scale=.5]{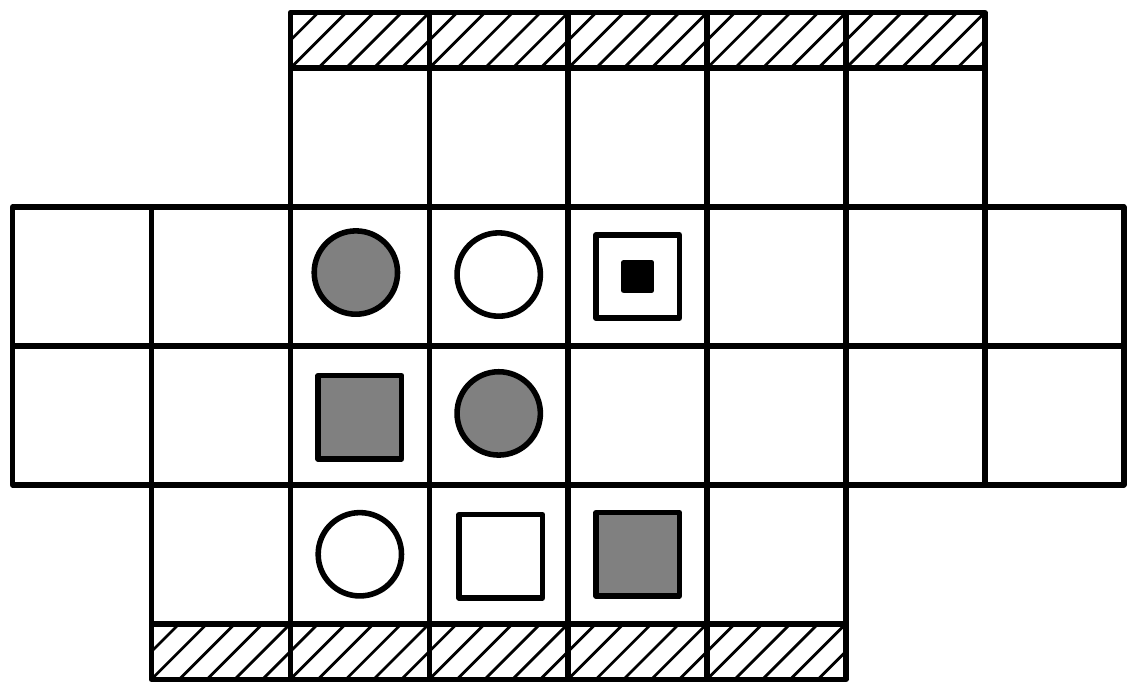}}
    \hfill\hfill
    \caption{An example push.}
    \label{figure:push_example}
\end{figure}
\fi

A player loses if any of their pieces are pushed off the board (even by their own push) or if they cannot push on their turn.

\begin{definition}
A \emph{\pushfight{} game state} is a description of the board's shape, including which board edges have side rails,
and for each board square, what type of piece or anchor occupies it (if any).
\end{definition}

Note that the position of the anchor encodes which player's turn it is: if the anchor is on a white king, it is black's turn, and vice versa. If the anchor has not been placed (no turns have been taken), it is white's turn.

\section{Mate-in-1}
\label{section:mate_in_one}
\abstractlater{
  \section{Proofs: Mate-in-1}
  \label{appendix:mate_in_one}
}

We consider three variants of mate-in-1 \pushfight, varying in how the number of moves is specified: as a constant in the problem definition, as part of the input, or without a limit.
\ifabstract
  See Appendix~\ref{appendix:mate_in_one} for omitted proofs.
\fi

\subsection{$c$-Move Mate-in-1}
\abstractlater{\subsection{$c$-Move Mate-in-1}}

\begin{problem}
\textsc{$c$-Move \pushfight{} Mate-in-1:} Given a \pushfight{} game state, can the player whose turn it is win this turn by making up to $c$ moves and one push?
\end{problem}

The standard \pushfight{} game has $c=2$.

\both{
\begin{theorem}
\textsc{$c$-Move \pushfight{} Mate-in-1} is in P.
\end{theorem}
}

\begin{proofsketch}
The number of possible turns is $\leq A^{2c+4}$ on a board of area~$A$.
\end{proofsketch}

\later{
\begin{proof}
Let $A$ denote the area of the board.
There are at most $A$ of the current player's pawns (because every piece occupies a square). On each of the $c$ moves, any of those pawns may move to any of those squares, for a maximum of $A^2$ possibilities for each move. (``Moving'' pawns to their own square represents making fewer than $c$ moves.)  Then there are at most $A$ kings of the current player, each of which can potentially push in four directions. Thus there are at most $A^{2c+4}$ possible turns.

Checking that a turn is legal and results in the current player winning requires checking that the moves are all legal and that the push is legal and leads to a win. A move can be verified in polynomial time by finding a path of unoccupied squares between the pawn's start and end positions. A push can be checked in polynomial time by scanning across the board in the direction of the push to see if one of the other player's pieces is pushed off the board, or if the push is invalid because of the anchored king or a side rail.
\end{proof}
}

\subsection{$k$-Move Mate-in-1 is in NP}
\abstractlater{\subsection{$k$-Move Mate-in-1 is in NP}}

\newcommand{\kmove}{\textsc{$k$-Move \pushfight{} Mate-in-1}}

\begin{problem}
\kmove\textsc{:} Given a \pushfight{} game state and a positive integer $k$, can the player whose turn it is win this turn by making up to $k$ moves and one push?
\end{problem}

In this section, we prove the following upper bound on the number of useful moves in a turn:

\begin{restatable}{theorem}{movebound}
\label{thm:move-bound}
Given a \pushfight{} game state on a board having $n$ squares, if the current player can win this turn, they can do so using at most $n^6$ moves followed by a push.
\end{restatable}

\begin{proofsketch}
  We divide the reachable game states into $\leq n^4$ equivalence classes,
  and show that two equivalent configurations can be reached via $\leq n^2$
  moves within that class.
\end{proofsketch}

Our bound directly implies an {\ccNP} algorithm for \textsc{$k$-Move \pushfight{} Mate-in-1}:

\both{
\begin{corollary}
\kmove{} is in \ccNP.
\end{corollary}
}

\later{
\begin{proof}
\label{alg:k-move-np}
Use the following NP algorithm.
Given a \pushfight{} game state on a board with $n$ squares, nondeterministically choose an integer $m$ between $0$ and $\min\{k, n^6\}$, then nondeterministically choose $m$ moves (a piece and a destination) and one push (a king and a direction). Accept if and only if the chosen moves and push are legal and result in a win for the current player.
\end{proof}

The remainder of this section is devoted to proving Theorem~\ref{thm:move-bound}.
}

A turn consists of making some number of moves followed by a single push. For the purpose of analyzing a single turn, kings other than the single king that pushes are indistinguishable from pawns, so we can assume the current player first chooses a king, then replaces all of their other kings with pawns before making their moves and push. The following definitions are based on this assumption.

\begin{definition}
Given a single-king game state, a \emph{board configuration} is a placement of pieces reachable by the current player
making a sequence of moves.

\xxx{Scoping to a particular game state is important (in particular, we don't need to record the pawnspace boundaries when naming a configuration in the unbounded problem later), but it's annoying to keep writing.}
\end{definition}

\begin{definition}
\label{def:pawnspace}
The \emph{pawnspace} of a board configuration is the (possibly disconnected) region of the board consisting of the empty squares and the squares containing the current player's pawns. Equivalently, the pawnspace is the region consisting of all squares not occupied by the current player's king or the other player's pieces.
\end{definition}

\begin{definition}
\label{def:signature}
The \emph{signature} of a board configuration is a list of nonnegative integers, where each integer is a count of the current player's pawns in a connected component of the configuration's pawnspace, ordered according to row-major order on the leftmost topmost square in the corresponding connected component.
\end{definition}

\begin{definition}
\label{def:equiv}
Given two board configurations $C_1$ and $C_2$ derived from the same game state, we say that $C_1 \equiv C_2$ if and only if
\begin{enumerate}
\item $C_1$ and $C_2$ have the same pawnspace (that is, the current player's only king occupies the same square in $C_1$ and $C_2$) and
\item $C_1$ and $C_2$ have the same signature (that is, each connected component of the pawnspace contains the same number of the current player's pawns in $C_1$ and $C_2$).
\end{enumerate}
\end{definition}

Relation $\equiv$ is clearly reflexive, symmetric, and transitive, so it is an equivalence relation inducing a partition of the set of board configurations derived from a given game state into equivalence classes. We need the following two lemmas about $\equiv$ for our proof of Theorem~\ref{thm:move-bound}:

\both{
\begin{lemma}
\label{thm:class-bound}
For a given game state on a board with $n$ squares, there are at most $n^4$ equivalence classes of board configurations.
\end{lemma}
}

\later{
\begin{proof}
Let $s$ be the square occupied by the current player's king. There are at most $n$ choices for $s$, so it remains to show that, for each $s$, there are at most $n^3$ equivalence classes where the current player's king occupies $s$.

The choice of $s$, together with the game state (containing the position of the other player's pieces), defines the pawnspace for all board configurations having the current player's king at $s$. For each connected component $R$ of the pawnspace, $s$ is either adjacent to $R$ or not. In the case where it is not, $R$ is surrounded by the boundary of the board and/or the other player's pieces, so no sequence of moves can change the number of the current player's pawns in $R$, and all board configurations have the same number of the current player's pawns in $R$.

Thus the only connected components that may have different numbers of pieces in different board configurations are the connected components bordering $s$, of which there are at most $4$. Each of these components comprises at most $n-1$ squares and so contains between $0$ and $n-1$ pieces. The total number of pieces in these components is invariant across all board configurations in each equivalence class, so the count of pawns in one of the components is fully determined by the others (that is, if there are $q$ components, there are $q-1$ degrees of freedom). There are $n$ possible values for each of up to $3$ free-to-vary pawn counts, so there are at most $n^3$ equivalence classes in which the current player's king occupies $s$. Together with the at most $n$ choices for $s$, there are at most $n^4$ equivalence classes of board configurations.
\end{proof}
}

\both{
\begin{lemma}
\label{thm:class-path}
If $C_1 \equiv C_2$, then $C_2$ can be reached from $C_1$ in at most $n^2 - 1$ moves without leaving the equivalence class of $C_1$.
\end{lemma}
}

\later{
\begin{proof}
\xxx{Part of this proof also wants to be titled algorithm. Maybe embed the algorithm environment in the proof environment?}
By the assumption that $C_1 \equiv C_2$, the current player's king and the other player's pieces are already in the same positions in $C_1$ and $C_2$. Notice that while moving the king may change the connected components of the pawnspace, moving pawns never does, nor can pawns move from one component to another (by the same argument as in Lemma~\ref{thm:class-bound}). Thus to ensure we remain in the equivalence class of $C_1$, it is sufficient to give an algorithm that moves only pawns.

Initialize the current configuration $C$ to be $C_1$. Call a pawn $\emph{misplaced}$ if it occupies a square in $C$ that is empty in $C_2$. While there are misplaced pawns, choose one and let $s$ denote the square it occupies. Let $R$ be the connected component in the pawnspace\footnote{Being in the same equivalence class, $C_1$, $C_2$ and all values of $C$ have the same pawnspace.} of $C$ containing $s$ and let $T$ be the set of squares in $R$ that contain pawns in $C_2$. By the assumption that $C_1 \equiv C_2$, there are the same number of pawns in $R$ in $C_1$ and $C_2$, so because there is a misplaced pawn, there is at least one square $t \in T$ that is empty (a missing pawn). Let $s_1,s_2,\dots,s_l$ be the squares containing pawns along a shortest path in $R$ from $s$ to $t$ (with $s_1 = s$). Move the pawn at $s_l$ to $t$, then move the pawn at $s_{l-1}$ to $s_l$, and so on, finishing by moving the pawn at $s_1$ to $s_2$. The net effect of this sequence of moves is that $t$ now holds a pawn and $s$ no longer holds a pawn (so $C$ now contains one fewer misplaced pawn). Continue with the next iteration of the loop.

During each iteration of the above loop, each pawn moves at most once. The number of misplaced pawns decreases by $1$ each iteration, so the number of iterations is at most the number of pawns, of which there are at most $n-1$. Thus the number of moves used to transform $C_1$ into $C_2$ is at most $(n-1)^2 \leq n^2 - 1$.
\end{proof}
}

We are now ready to prove Theorem~\ref{thm:move-bound}:

\movebound*

\begin{proof}
By our assumption that the current player can win this turn, there exists a sequence of moves for the current player after which they can immediately win with a push, corresponding to a sequence of board configurations $C_1,C_2,\dots,C_l$. Configuration $C_1$ is obtained from the initial game state by replacing all of the current player's kings, except the one that ends up pushing, with pawns. Each $C_{i+1}$ can be reached from $C_i$ in one move, and $C_l$ is a configuration from which the current player can win with a push.

We now define \emph{simplifying} a sequence of board configurations over an equivalence class $E$. If the sequence contains no configurations from $E$, then simplifying the sequence over $E$ leaves it unchanged. Otherwise, let $A_i$ be the first configuration in the sequence in $E$ and $A_j$ be the last configuration in the sequence in $E$. By Lemma~\ref{thm:class-path}, there exists a sequence of fewer than $n^2 - 1$ moves that transforms $A_i$ into $A_j$, corresponding to a sequence of board configurations $A_i = D_0, D_1, \dots, D_u = A_j$ with $u \leq n^2 - 1$. Then simplifying over $E$ consists of replacing all configurations between and including $A_i$ and $A_j$ with the replacement sequence $D_0, D_1, \dots, D_u$.

Notice that simplifying a sequence (over any class) never changes the first or last configuration in the sequence, and each configuration in the resulting sequence remains reachable in one move from the previous configuration in the resulting sequence. After simplifying over a class $E$, the only configurations in the resulting sequence in $E$ are those in the replacement sequence, so the number of configurations in the sequence in $E$ is at most $n^2$. Furthermore, all configurations in the replacement sequence are in $E$, so simplifying over $E$ never increases (but may decrease) the number of configurations falling in other classes.

Let $C'_1,C'_2,\dots,C'_l$ be the result of simplifying $C_1,C_2,\dots,C_l$ over every equivalence class. By Lemma~\ref{thm:class-bound}, there are at most $n^4$ such classes, and by the above paragraph there are at most $n^2$ configurations from each class in $C'_1,C'_2,\dots,C'_l$, so the length of $C'_1,C'_2,\dots,C'_l$ is at most $n^6$. Each configuration in $C'_1,C'_2,\dots,C'_l$ is reachable in one move from the previous configuration, and that sequence of at most $n^6$ moves leaves the current player in position to win with a push, as desired.
\end{proof}

\subsection{Unbounded-Move Mate-in-1}
\abstractlater{\subsection{Unbounded-Move Mate-in-1}}

\begin{problem}
\textsc{Unbounded-Move \pushfight{} Mate-in-1:} Given a \pushfight{} game state, can the player whose turn it is win this turn by making any number of moves and one push?
\end{problem}

\begin{restatable}{theorem}{unboundedp}
\label{thm:unbounded-p}
\textsc{Unbounded-Move \pushfight{} Mate-in-1} is in P.
\end{restatable}

We can of course solve \textsc{Unbounded-Move \pushfight{} Mate-in-1} by trying all possible sequences of moves to find a board configuration from which the current player can win with a push, but there are exponentially many board configurations, so such an algorithm takes exponential time. Instead, we can use the fact that any two configurations in the same equivalence class are reachable from each other in a polynomial number of moves (from Lemma~\ref{thm:class-path}) to search over equivalence classes of board configurations instead of searching over board configurations. There are at most $n^4$ equivalence classes (by Lemma~\ref{thm:class-bound}), so they can be searched in polynomial time.

We will make use of the following definitions:

\begin{definition}
Two equivalence classes of board configurations $C_1$ and $C_2$ are \emph{neighbors} if there exist board configurations $b_1 \in C_1$ and $b_2 \in C_2$ such that $b_1$ can be reached from $b_2$ with a king move of exactly one square. The \emph{equivalence class graph} is a graph whose vertices are equivalence classes of board configurations and whose edges connect neighboring equivalence classes.

An equivalence class of board configurations $C$ is a \emph{winning equivalence class} if there exists a board configuration $b \in C$ such that the player whose turn it is can win with a push.
\end{definition}

The key idea for our algorithm is the following:

\both{
\begin{lemma}
\label{thm:key_idea}
There exists a path in the equivalence class graph from the equivalence class of the initial board configuration to a winning equivalence class if and only if there exists a winning move sequence.
\end{lemma}
}

\later{
\begin{proof}
Let $b_0$ be the initial board configuration, and let $C_0$ be the equivalence class of $b_0$. 

First, suppose that there exists a path $C_0, C_1, \ldots, C_k$ in the equivalence class graph which ends at some winning equivalence class. Consider any $i \in \{0, 1, \ldots, k-1\}$. Because $C_i$ is adjacent to $C_{i+1}$ in the equivalence class graph, equivalence class $C_i$ neighbors $C_{i+1}$, and therefore there exist two board configurations $b_i' \in C_i$ and $b_{i+1} \in C_{i+1}$ such that $b_{i+1}$ can be reached from $b_i'$ with a king move of exactly one square. Because $C_k$ is a winning equivalence class, there exists a board configuration $b_k'$ such that the current player can win with a push. Then consider the sequence of board configurations $b_0, b_0', b_1, b_1', \ldots, b_k, b_k'$. For each $i$, $b_i$ and $b_i'$ are both in equivalence class $C_i$, so by Lemma~\ref{thm:class-path}, there exists a sequence of moves converting board configuration $b_i$ into $b_i'$. Together with the single-square king moves between $b_i'$ and $b_{i+1}$, we can use these sequences to form a winning move sequence from board configuration $b_0$ to board configuration $b_k'$.

Next, suppose there exists a winning move sequence. Break each king move along a path of more than one square in that sequence into a sequence of king moves of exactly one square. This yields a new winning move sequence whose moves are all king moves of one square or pawn moves. Pawn moves never change the equivalence class of the current board configuration, and single-square king moves always change the equivalence class of the current board configuration to a neighboring equivalence class. Thus, this sequence of moves corresponds to a path in the equivalence class graph. Because this is a winning move sequence, the last board configuration has the property that the player whose turn it is can win with a push, and so the last equivalence class in the path is a winning equivalence class. Thus there exists a path in the equivalence class graph from the equivalence class of the initial board configuration to a winning equivalence class.
\end{proof}
}

The size of the equivalence class graph is polynomial in $n$ (by Lemma~\ref{thm:class-bound}), so provided the graph can be constructed and the winning equivalence classes identified, this type of path in the equivalence class graph, if it exists, can be found in polynomial time.

Recall from Definition~\ref{def:equiv} that equivalence classes of board configurations are defined by the pawnspace and signature, and that, for configurations derived from the same game state (i.e., having the other player's pieces in the same positions), the pawnspace is defined by the position of the current player's king. Thus we can uniquely name a class using the king position and signature.

\begin{definition}
\label{def:class-descriptor}
The \emph{class descriptor} of an equivalence class of board configurations for a given game state is the ordered pair of the position of the current player's king and the signature defining that class.
\end{definition}

To prove Theorem~\ref{thm:unbounded-p}, we need to give polynomial-time algorithms to compute the neighbors of an equivalence class and to decide whether a class is a winning equivalence class.

\both{
\begin{lemma}
\label{thm:find-edges}
Given an initial game state and a class descriptor for some class $C$, we can compute in polynomial time the equivalence classes (as class descriptors) neighboring $C$.
\end{lemma}
}

\later{
\begin{proof}
Moving the king to an adjacent square changes the pawnspace by adding the previously occupied square to a connected component and removing the newly occupied square from the pawnspace. Moving the king to an adjacent square may also join up to three components adjacent to the king's previous square or split a component containing the king's new square into up to three components. Other components are unaffected, and their corresponding signature elements are not modified.

In all cases, the component in the current configuration containing the king's new position must have area at least one greater than the number of pawns in that component. That is, there must be an empty square in that component. Otherwise, moving the king in that direction is not possible. All neighboring class descriptors specify the new king position, but how the signature is updated varies.

\xxx{The following is a long-winded way of saying ``Update the signature according to the changes in the components."}

If the king move does not change the number of connected components of the pawnspace, the signature is left unchanged in the neighboring class descriptors.

If the king move joins components, then the neighboring class descriptor is updated by inserting a signature element for the newly joined component with value equal to the sum of the components it was created from, and deleting the signature elements corresponding to the joined components.

If the king move splits components, then there are potentially multiple neighboring class descriptors, one for each possible resulting signature. Suppose the signature value of the component being split is $S$. There is one neighboring class descriptor for each of the weak compositions\footnote{A weak composition of an integer is a way of writing that integer as a sum of other positive integers, where terms may be $0$ and order is significant.} of $S$ of length equal to the number of newly split components such that each term in the composition is at most the area of the corresponding newly split component. Each neighboring class descriptor's signature is updated by removing the element corresponding to the component having been split and inserting the terms of the composition in the positions corresponding to the newly split components.

There are up to four directions in which the king can move. Each of the above update procedures takes time polynomial in the number of class descriptors produced. By Lemma~\ref{thm:class-bound}, there are only polynomially many classes, so only polynomially many descriptors can be produced, and thus we can compute the neighboring descriptors in polynomial time.
\end{proof}
}

\both{
\begin{lemma}
\label{thm:did-we-win}
Given an initial game state and a class descriptor for some class $C$, we can decide in polynomial time whether $C$ is a winning equivalence class.
\end{lemma}
}

\later{
\begin{proof}
We wish to decide whether $C$ contains a board configuration such that the current player can win with one push.

Consider each possible push direction. The king's position is specified in the class descriptor of $C$, so let $L$ be the line of squares starting adjacent to the king in the chosen direction and continuing to the boundary of the board. Pushing in this direction results in a win exactly when

\begin{enumerate}
\item The square in $L$ furthest from the king contains a piece belonging to the other player.
\item There is no side rail at the boundary edge at the end of $L$ furthest from the king.
\item No square in $L$ contains the anchored king.
\item Every square in $L$ contains a piece.
\end{enumerate}

The first three conditions depend only on the current player's king's position and the positions of the other player's pieces, information specified in the given class descriptor and game state, and so can be checked in polynomial time independent of any specific board configuration in $C$.  Because the current player's pawns can move freely within the connected components of the pawnspace without leaving $C$, the fourth condition is equivalent to the condition that the intersection of $L$ with each connected component of the pawnspace has area less than or equal to the number of the current player's pawns in that component.

Thus, we can check in polynomial time for each direction whether there exists a board configuration in $C$ such that pushing in the chosen direction results in the current player winning.  By repeating this check for all four possible push directions, we can decide whether $C$ is a winning equivalence class in polynomial time.
%
\end{proof}
}

We are now ready to prove Theorem~\ref{thm:unbounded-p}:

\unboundedp*

\begin{proof}
First, compute the class descriptor for the equivalence class of the initial board configuration. Then perform a breadth- or depth-first search of the equivalence class graph, using the algorithm given in the proof of Lemma~\ref{thm:find-edges} to compute the neighboring class descriptors and the algorithm given in the proof of Lemma~\ref{thm:did-we-win} to decide if the search has found a winning equivalence class. Each of these procedures takes polynomial time. By Lemma~\ref{thm:class-bound}, there are only polynomially many equivalence classes, so the search terminates in polynomial time. By Lemma~\ref{thm:key_idea}, there exists a winning move sequence if and only if this search finds a path to a winning equivalence class. 
\end{proof}

The key idea of the above proof is that, if we do not care how many moves we make inside an equivalence class, then it is sufficient to search the graph of equivalence classes. Thus the above proof does not apply to \kmove{}, and in the next section, we prove \kmove{} is NP-hard.

\subsection{$k$-Move Mate-in-1 is NP-hard}
\abstractlater{\subsection{$k$-Move Mate-in-1 is NP-hard}}
\label{section:np_hardness}

To prove \kmove{} hard, we reduce from the following problem, proved strongly NP-hard in \cite{steiner}:

\newcommand{\irst}{\textsc{Integer Rectilinear Steiner Tree}}
\begin{problem}
\irst\textsc{:} Given a set of points in $\mathbb{R}^2$ having integer coordinates and a length $\ell$, is there a tree of horizontal and vertical line segments of total length at most $\ell$ containing all of the points?
\end{problem}

The basic idea of our reduction is to create a game state mostly full of the current player's pawns, but with a few empty squares (\emph{holes}). The player must ``move'' the holes (by moving pawns into them, creating a new hole at the pawn's former square) to free a king that can push one of the other player's pieces off the board. Initially each pawn can only travel one square (into an adjacent hole) per move, but once two holes have been brought together, a pawn can travel two squares per move, and so on. Bringing the holes together optimally amounts to finding a Steiner tree covering the holes' initial positions.

\both{
\begin{theorem}
\label{thm:kmove-np-hard}
\kmove{} is strongly \ccNP-hard.
\end{theorem}
}


\paragraph{Reduction:} Suppose we are given an instance of \irst{} consisting of points $p_i = (x_i, y_i)$ with $i = 1, \ldots, n$ and length $\ell$. For convenience, and without affecting the answer, we first translate the points so that $\min x_i = 2$ and $\min y_i = 4$ and reorder the points such that $y_1 = 4$.

We then build a \pushfight{} game state with a rectangular board with a height of $\max y_i$ and a width of $n + \max x_i$, indexed using $1$-based coordinates with the origin in the bottom-left square; refer to Figure~\ref{figure:np_reduction_example}. The entire boundary of the board has side rails except the edge adjacent to square $(x_1, 1)$. There is a white king in square $(x_1 + n, 2)$ and a black king with the anchor in square $(x_1 - 1, 2)$. There is a black pawn in square $(x, y)$ if any of the following are true:
\begin{enumerate}
\item $y = 3$ and $x \ne x_1$,
\item $y = 2$ and either $x < x_1 - 1$ or $x > x_1 + n$, or
\item $y = 1$.
\end{enumerate} 
The squares $(x_i, y_i)$ with $1 \leq i \leq n$ (corresponding to the points in the \irst{} instance) are empty. All remaining squares are filled with white pawns. The output of the reduction is this \pushfight{} board together with $k = \ell+3$. 

\begin{figure}
    \centering
    \hfill
    \raisebox{-.5\height}{\includegraphics[scale=.5]{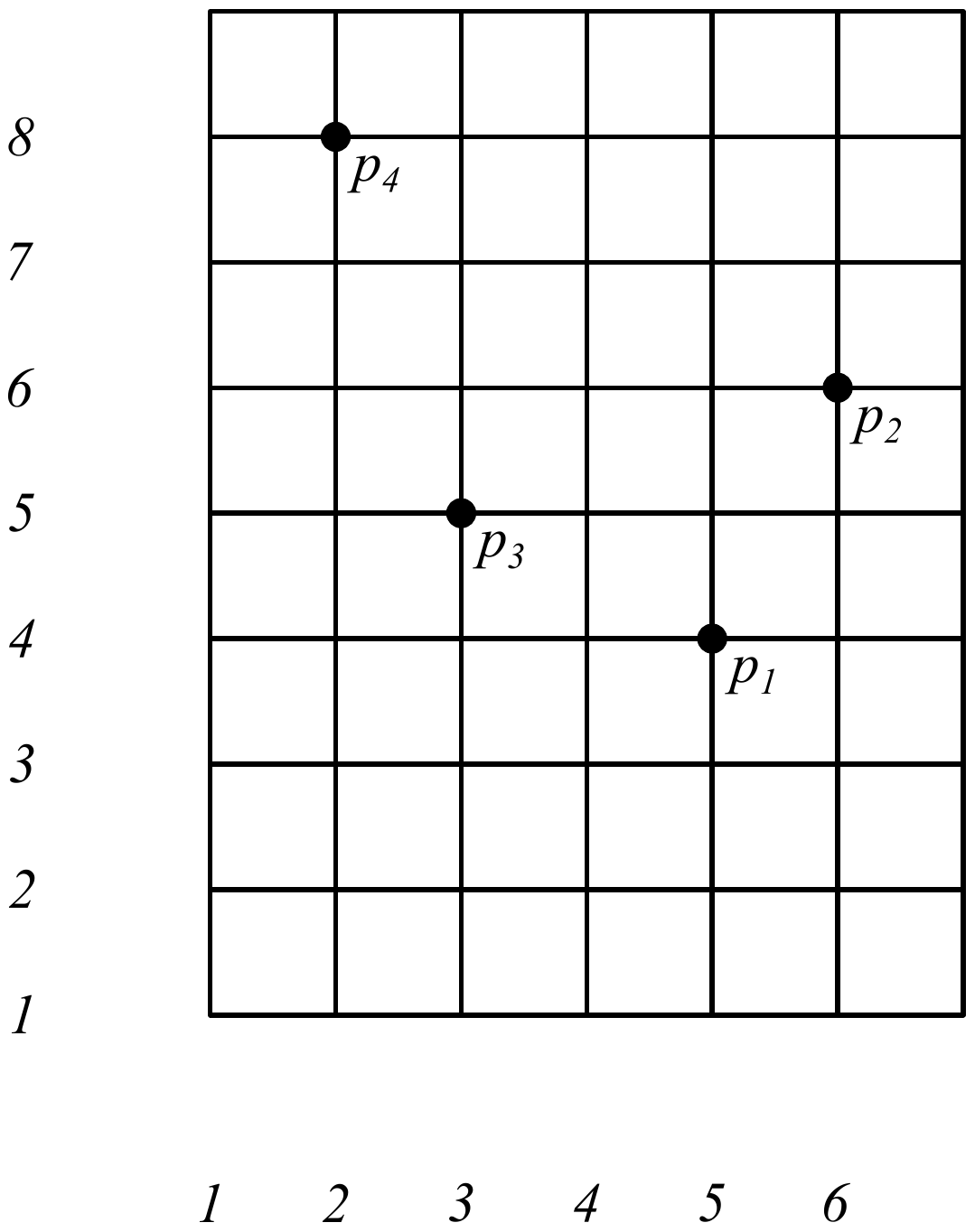}}
    \hfill\raisebox{-.5\height}{\scalebox{2}{$\to$}}\hfill
    \raisebox{-.5\height}{\includegraphics[scale=.5]{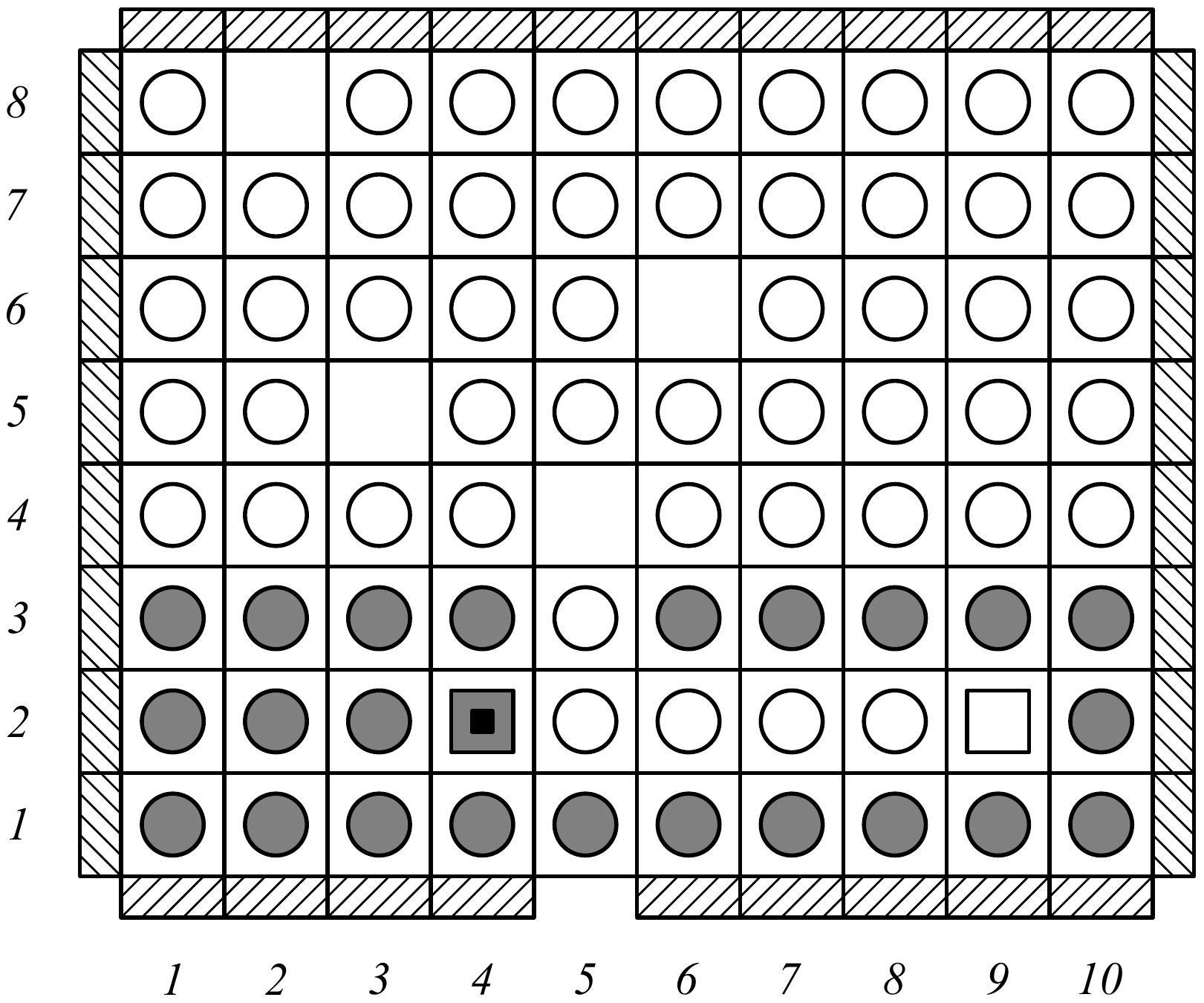}}
    \hfill\hfill
    \caption{A \pushfight{} board (right) produced during the reduction from the points in an example rectilinear Steiner tree instance (left).}
    \label{figure:np_reduction_example}
\end{figure}

\later{
\ifabstract The reduction from Section~\ref{section:np_hardness}
\else This reduction
\fi
can clearly be computed in polynomial time. It remains to show that there exists a solution to the \irst{} instance if and only if White can win in the corresponding \kmove{} instance.



\paragraph{Move sequence $\implies$ Steiner tree:} To win, White must push a black piece off the board. The only boundary edge without a side rail in the game state produced by the reduction is the south edge of square $(x_1, 1)$, so White must move their king to $(x_1, 2)$ and push downward. Call the $n$ squares directly to the left of the white king the \emph{alley} and all squares with $y > 3$ the \emph{plaza}. Before moving the white king into position to push, White must move the $n$ white pawns in the alley into the plaza (which has exactly $n$ empty squares). For the example reduction output in Figure~\ref{figure:np_reduction_example}, Figure~\ref{figure:np_about_to_win_example} shows the state of the board after moving all pawns from the alley into the plaza, then moving the white king through the now-empty alley.

\begin{figure}
    \centering
    \includegraphics[scale=.5]{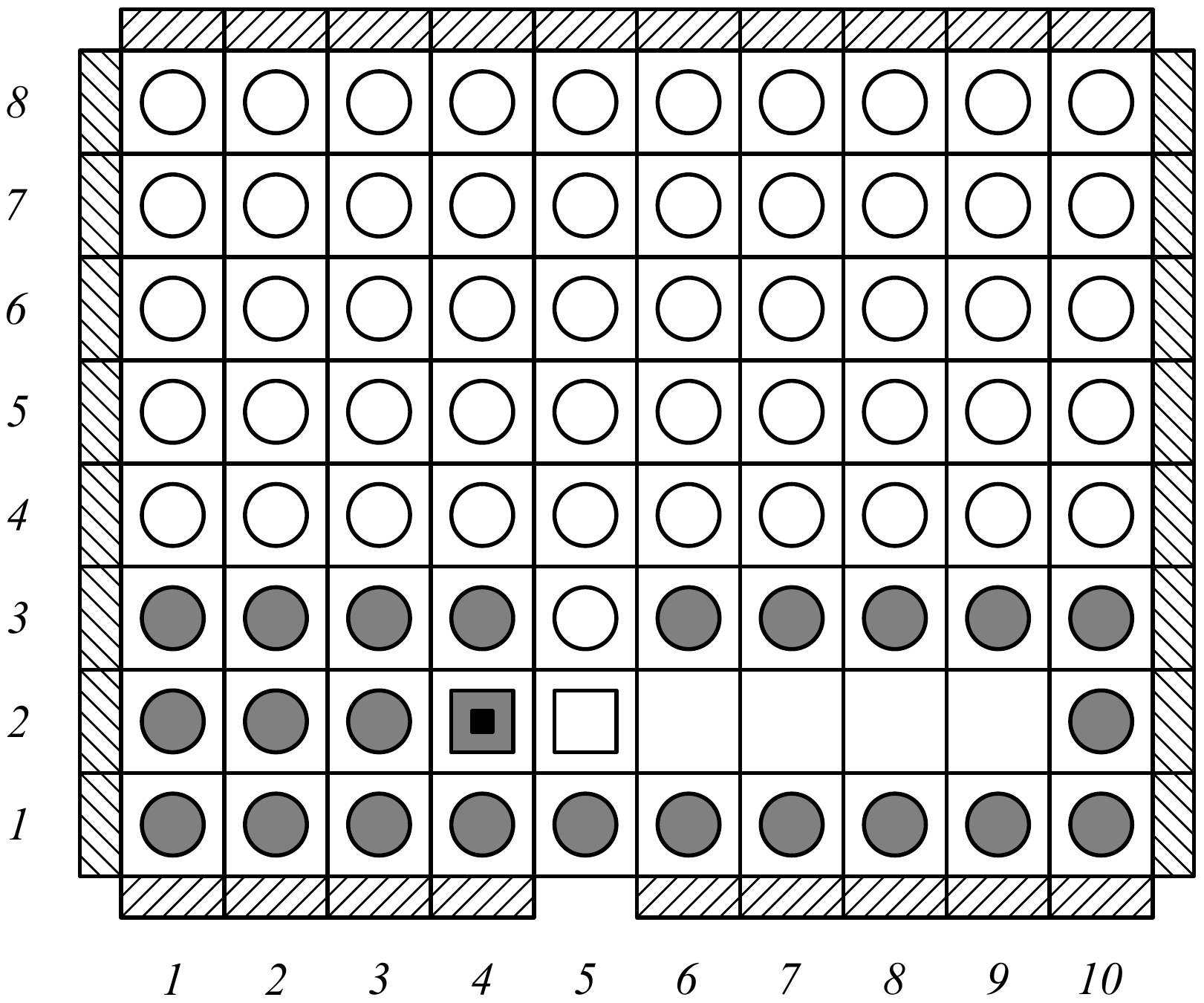}
    \caption{The result of moving all alley pawns in Figure~\ref{figure:np_reduction_example} into the plaza, then moving the white king through the alley. White wins by pushing down.}
    \label{figure:np_about_to_win_example}
\end{figure}

Suppose White can win with at most $k$ moves and a push. Let $S$ be the set of squares that are ever empty during the winning move sequence and let $S'$ be the subset of the plaza induced by $S$. 
We prove two useful facts about $S'$.

\begin{lemma}
\label{lemma:np_first_direction_useful_1}
$S'$ is connected.
\end{lemma}

\begin{proof}
We can use the move sequence (except for the king move(s)) to build a set of paths $P_1, P_2, \ldots, P_n$ from the empty squares in the plaza to the squares in the alley. Each path essentially traces the path of one hole over the course of the move sequence. As we go through the move sequence building the paths we maintain the invariant that the currently empty squares are the last elements in the paths. Initialize the $n$ different paths to start in the $n$ empty plaza squares. This satisfies the invariant at the start of the move sequence. During each move, a pawn moves from some square $s$, through some sequence of empty squares $s_1, s_2, \ldots, s_l$, and into a currently empty square $t$. By our invariant, there exists a path currently ending at $t$, so we extend that path with the sequence $s_l, \ldots, s_2, s_1, s$. The new endpoint of that path is the square that was just emptied, so the invariant is maintained.

All squares added to paths during this procedure are empty when they are added, so $\bigcup P_i \subseteq S$. In the other direction, all of the initially empty squares are endpoints of paths, the square emptied by each move is always appended to a path, and squares are never removed from paths, so $S \subseteq \bigcup P_i$. Thus $\bigcup P_i = S$.

We know that the alley must be empty for the king to move in position to push, so the final endpoints of the paths are exactly the alley squares. Consider any two squares $s,t \in S$. $s$ occurs in at least one path $P_i$, and the section of that path starting at $s$ (call it $P_i^s$) is a path in $S$ from $s$ to an alley square. Similarly, there is a section $P_j^t$ starting at $t$ of some path $P_j$ that also ends in an alley square. Then $P_i^s$ and the reverse of $P_j^t$, joined by a horizontal path entirely within the alley, is a path in $S$ from $s$ to $t$. Thus we can construct a path in $S$ between any two elements of $S$, so $S$ is connected.

Because the plaza is separated from the alley by black pawns except at $(x_1, 3)$, all paths in $S$ containing squares both inside and outside the plaza must contain $(x_1, 3)$. Then we can convert any path in $S$ starting and ending in $S'$ to a path entirely in $S'$ by deleting all squares between and including the first and last instances of $(x_1, 3)$ (if any), so $S'$ is also connected.
\end{proof}

\begin{lemma}
\label{lemma:np_first_direction_useful_2}
$|S'| \le \ell+1$.
\end{lemma}

\begin{proof}
The reduction leaves $n$ squares empty. Of the $k$ moves in the move sequence, at least one moves the king; the remaining $k-1$ moves empty at most one square each, so $|S| \leq n+k-1$. There are $n+1$ squares in $S$ that are not in the plaza (the $n$ alley squares and $(x_1, 3)$), so \xxx{do we carry the $\leq$ to subsequent lines, or only where we introduced it? could also use one-line display math}


\begin{align*}
|S'| &= |S| - (n+1) \\
     &\leq n + k - 1 - n - 1 \\
     &= k - 2 \\
     &= \ell+1
\end{align*}
\end{proof}

Consider the grid graph $G$ induced by vertex set $S'$. $S'$ is connected (Lemma~\ref{lemma:np_first_direction_useful_1}), so $G$ is also connected. Let $T$ be any spanning tree of $G$. Each edge in $G$ (and therefore in $T$) is a unit-length vertical or horizontal segment. $S'$ contains every $p_i$, so $T$ does also. By Lemma~\ref{lemma:np_first_direction_useful_2}, $G$ has $|S'| \le \ell+1$ vertices, and $T$ has one fewer edge than it has vertices, so the total length of $T$ is at most $\ell$. Thus $T$ is a solution to the \irst{} instance, and such a solution exists if White can win in the \kmove{} instance.


\paragraph{Steiner tree $\implies$ move sequence:}
We are given a Steiner tree with total length at most $\ell = k-3$. Hanan's Lemma~\cite[Theorem 4]{Hanan-1966} states that there exists an optimal tree whose edges are entirely contained on vertical and horizontal lines through each $p_i$, so we assume without loss of generality that the given tree has this form. Each $p_i$ has integer coordinates, so we can subdivide the edges of the given tree into unit-length edges, resulting in a tree with at most $\ell+1$ vertices, all having integer coordinates. By our assumption that the tree is optimal, each leaf of the tree is one of the $p_i$ (though not all $p_i$ are necessarily leaves).

Let $T$ be the result of augmenting the subdivided given tree with a path through the vertices corresponding to square $(x_1, 3)$ \xxx{find a name for the entrance to an alley} and the squares in the alley, and consider $T$ to be rooted at $(x_1 + n - 1, 2)$ (the alley square adjacent to the white king). We added $n+1$ vertices, so $T$ has at most $\ell+1+n+1 = (k - 3) + 1 + n + 1 = k + n - 1$ vertices. Observe that each vertex in $T$ corresponds to a square that is either occupied by a white pawn or empty. \xxx{We say 'corresponding to' a lot throughout. Can we abuse notation a bit and have $T$ contain squares?}

We build the move sequence from $T$ by repeatedly choosing a move as follows. Choose any leaf $(x, y)$ of $T$ and walk along the tree towards the root until a vertex $(x', y')$ corresponding to an occupied square is reached. The reverse of that walk is a valid move of the white pawn at square $(x', y')$ through some number of holes to the hole at $(x, y)$. Append that move to the sequence (updating the board state accordingly), then remove the leaf $(x, y)$ from $T$. Continue this loop until the root of $T$ corresponds to an empty square (such that the preceding procedure would fail to find an occupied square when moving towards the root).

$T$ initially contains vertices corresponding to $n$ holes (the $p_i$). Each move creates a hole at $(x', y')$, but fills in a hole at $(x, y)$, so the number of vertices corresponding to holes in $T$ always remains at $n$. The loop ends when there is a path from a leaf to the root containing only vertices corresponding to holes. All paths from the root of $T$ begin with a specific path of length $n+1$ (the path we augmented the given tree with), so when the loop terminates, the first $n$ squares corresponding to that path (the alley) must be holes. Each iteration of the loop removes a vertex from $T$ and appends a move to the move sequence. $T$ initially contains $k+n-1$ vertices and ends containing $n$ vertices, so the generated move sequence contains $k-1$ moves, to which we append a move of the white king through the now-empty alley into position to win by pushing down. Thus the move sequence is a solution to the \kmove{} instance, and such a solution exists if there is a solution to the \irst{} instance.

Having proved both directions, we have proved Theorem~\ref{thm:kmove-np-hard}. \xxx{split this from the above paragraph-heading, or restate the theorem, or put the whole thing in a proof environment}

}

\section{\pushfight{} is \ccPSPACE-hard}
\label{section:pspace}
\abstractlater{
  \section{Proofs: \pushfight{} is \ccPSPACE-hard}
  \label{appendix:pspace}
}

In this section, we analyze the problem of deciding the winner of a \pushfight{} game in progress.
\ifabstract
  See Appendix~\ref{appendix:pspace} for omitted proofs.
\fi

\newcommand{\pfprob}{\textsc{\pushfight{}}}
\begin{problem}
\pfprob{}\textsc{:} Given a \pushfight{} game state, does the current player have a winning strategy (where players make up to two moves per turn)?
\end{problem}

\begin{theorem}
  \pfprob{} is \ccPSPACE-hard.
\end{theorem}

\newcommand{\qbf}{\textsc{Q3SAT}}
To prove \ccPSPACE-hardness, we reduce from \qbf{}, proved \ccPSPACE-complete in \cite{Stockmeyer-Meyer-1973,garey}:

\begin{problem}
\qbf{}\textsc{:} Given a fully quantified boolean formula in conjunctive normal form with at most three literals per clause, is the formula true?
\end{problem}

\newcommand{\pushstar}{\textsc{Push-$*$}}
Our proof parallels the \ccNP-hardness proof of \pushstar{} in~\cite{push}. \pushstar{} is a motion-planning problem in which a robot (agent) traverses a rectangular grid, some squares of which contain blocks. The robot can push any number of consecutive blocks when moving into a square containing a block, provided no blocks would be pushed over the boundary of the board. The \pushstar{} decision problem asks, given a initial placement of blocks and a target location, can the robot reach the target location by some sequence of moves? In our proof, the white king takes the place of the \pushstar{} robot\footnote{The \pushstar{} robot can move without pushing blocks, so the correspondence is not exact.} and white pawns function as blocks. Our proof has the additional complication that Black sets the universally quantified variables, and that White's moves and Black's push must be forced at all times to keep the other gadgets intact.

\begin{figure}
    \centering
    \includegraphics[scale=.5]{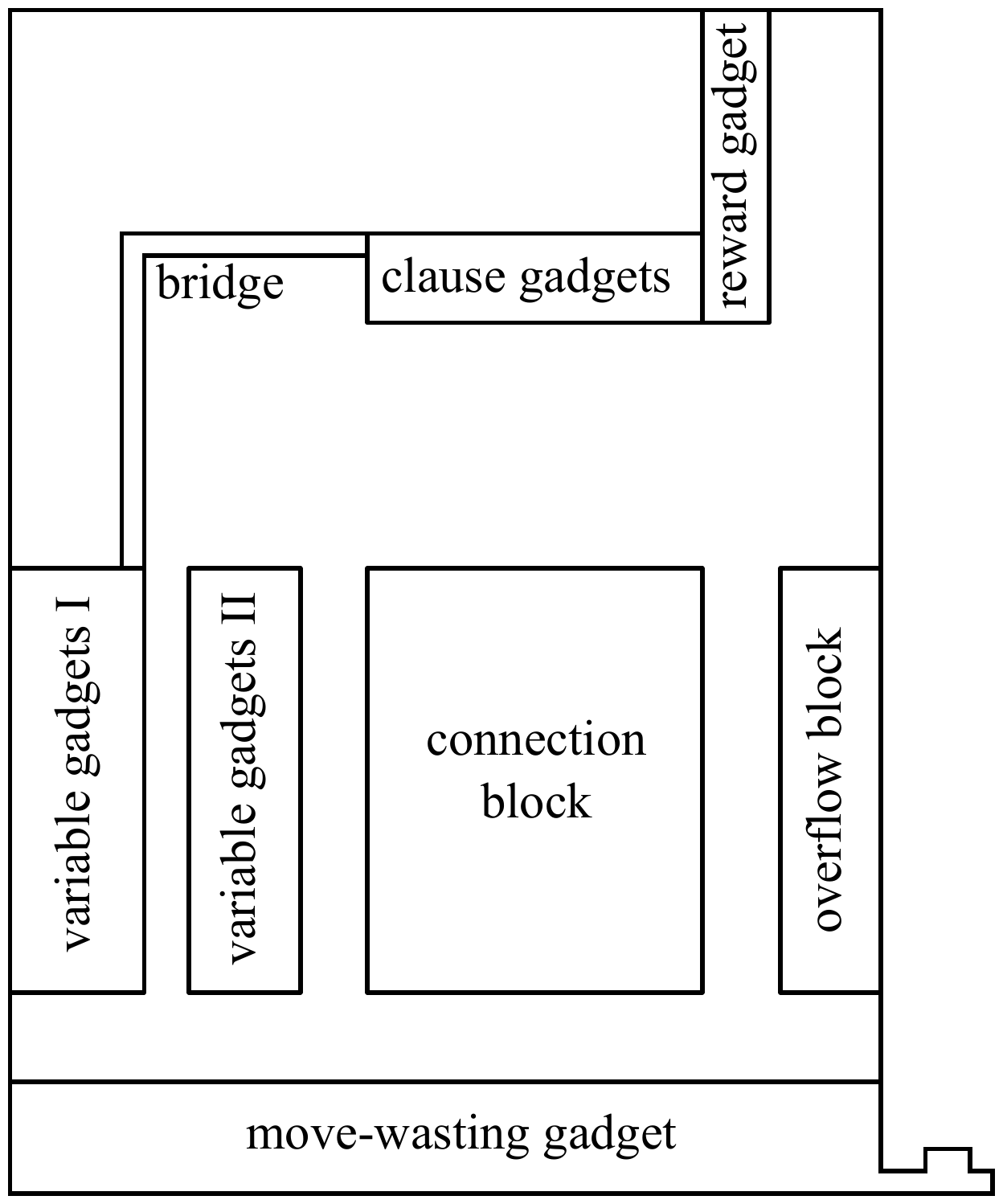}
    \caption{An overview of the \pushfight{} board produced by our reduction.}
    \label{figure:reduction_plan}
\end{figure}

Figure~\ref{figure:reduction_plan} shows an overview of the reduction. The sole white king begins at the bottom-left of the \emph{variable gadget I} block, setting existentially quantified variables as it pushes up and right. The \emph{variable gadget II} block contains black pawns and holes that allow Black to set the universally quantified variables. After all the variables have been set, the white king traverses the \emph{bridge} to the \emph{clause gadget} block. The variable and clause gadgets interact via a pattern of holes in the \emph{connection block} encoding the literals in each clause. The white king can traverse the clause gadgets only if the variable gadgets were traversed in a way corresponding to a satisfying assignment of the variables. The \emph{reward gadget} contains a boundary square without a side rail, such that the white king can push a black pawn off the board if the white king reaches the reward gadget. The \emph{overflow block} contains empty squares needed by the variable gadgets that were not used in the connection block (for variables appearing in few clauses). The \emph{move-wasting gadget} forces White's moves and Black's push, ensuring the integrity of the other gadgets. Finally, all other squares on the board are filled with white pawns, and the boundary has side rails except at specific locations in the reward and move-wasting gadgets. Figure~\ref{figure:pspace_reduction_example} shows an example output of the reduction.

\begin{figure}
    \centering
    \includegraphics[scale=.25]{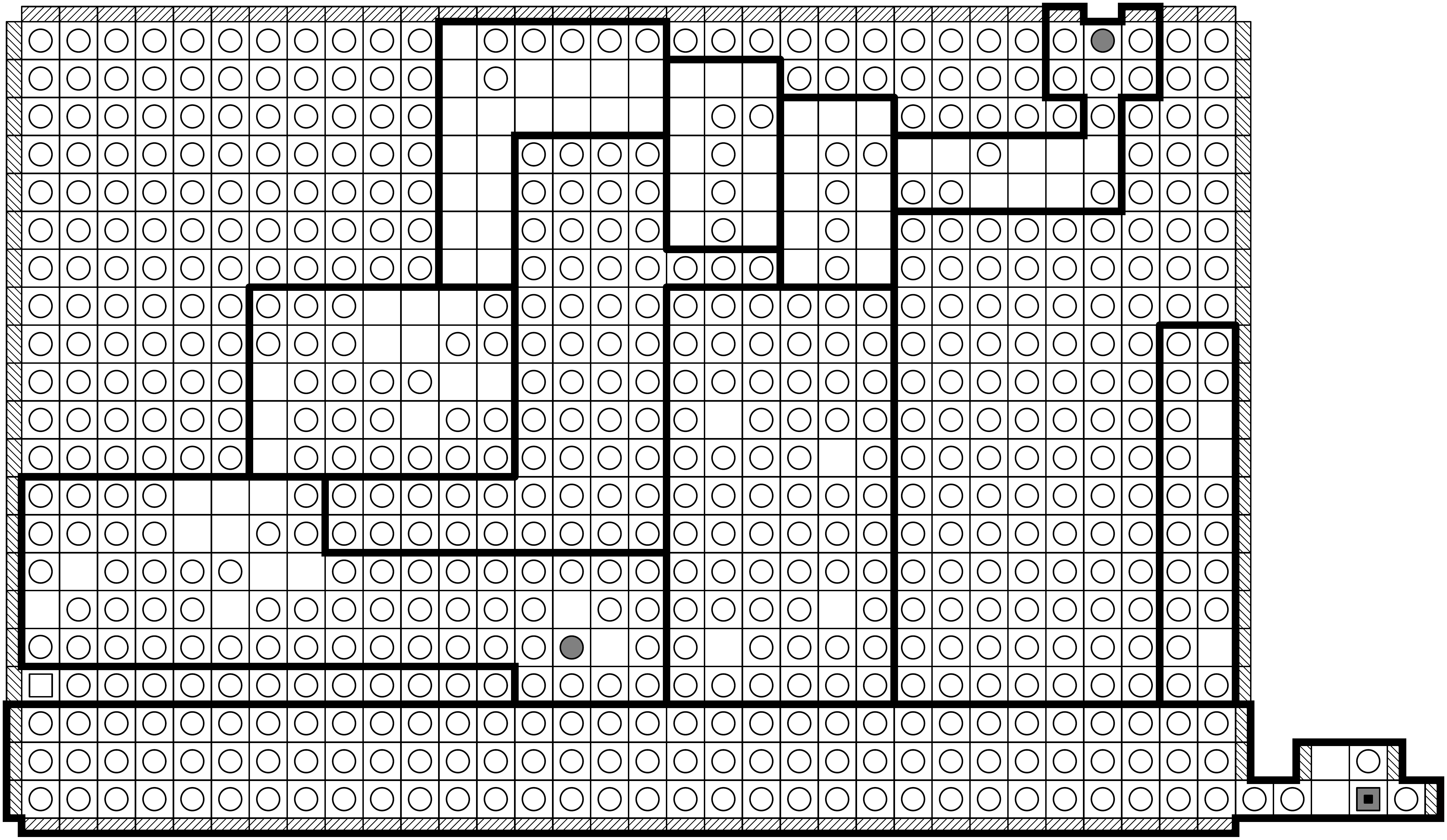}
    \caption{The result of performing the reduction on the formula $\forall x \exists y~ (x \lor \neg y) \land (\neg x \lor y)$. Gadgets and blocks are outlined.}
    \label{figure:pspace_reduction_example}
\end{figure}

\later{
We first prove the behavior of each of the gadgets, then describe how the gadgets are assembled. \xxx{we could explain/justify the out-push constraints here}
}



\subsection{Move-wasting gadget}
\abstractlater{\subsection{Move-wasting gadget}}

The move-wasting gadget requires White to use both moves to prevent Black from winning on the next turn (unless White can win in the current turn). The move-wasting gadget contains the only black king, thus consuming (and allowing) Black's push each turn. When analyzing the other gadgets, we can thus assume White can only push and Black can only move. The move-wasting gadget comprises the entire bottom three rows of the board, but pieces only move in the far-right portion. Figure~\ref{figure:move_wasting_gadget_1} shows the initial state of the gadget.
\later{
Throughout this analysis, we assume White cannot win in one turn; Section~\ref{sec:reward-gadget}, which analyzes the reward gadget, describes the position in which White can immediately win in one turn, and can therefore disregard the threat from Black in the move-wasting gadget.
}

\begin{figure}
  \subcaptionbox{\label{figure:move_wasting_gadget_1} Initial state}
    {\includegraphics[scale=.5]{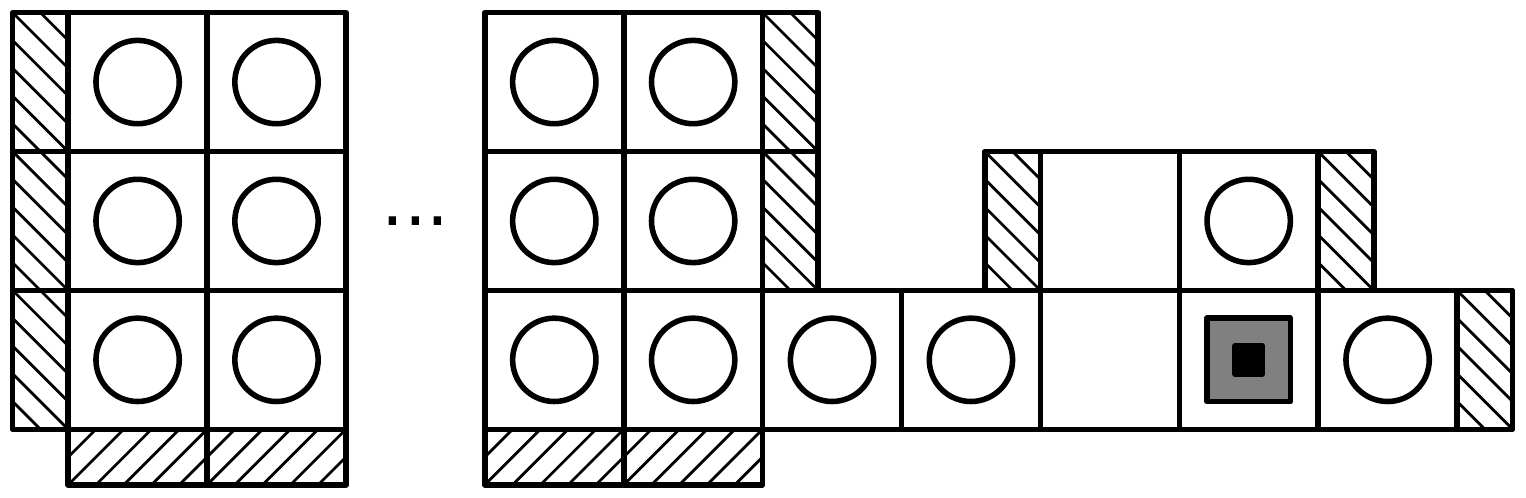}}
  \subcaptionbox{\label{figure:move_wasting_gadget_2} One white turn after (\subref{figure:move_wasting_gadget_1})}
    {\includegraphics[scale=.5]{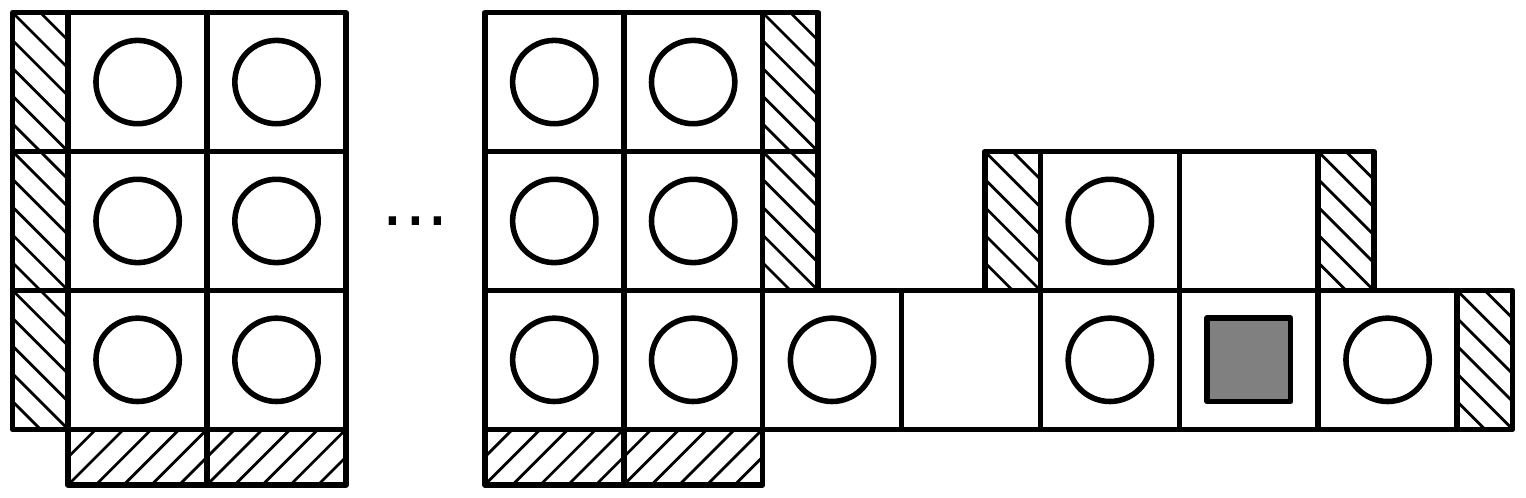}}

  \medskip

  \subcaptionbox{\label{figure:move_wasting_gadget_3} One black turn after (\subref{figure:move_wasting_gadget_2})}
    {\includegraphics[scale=.5]{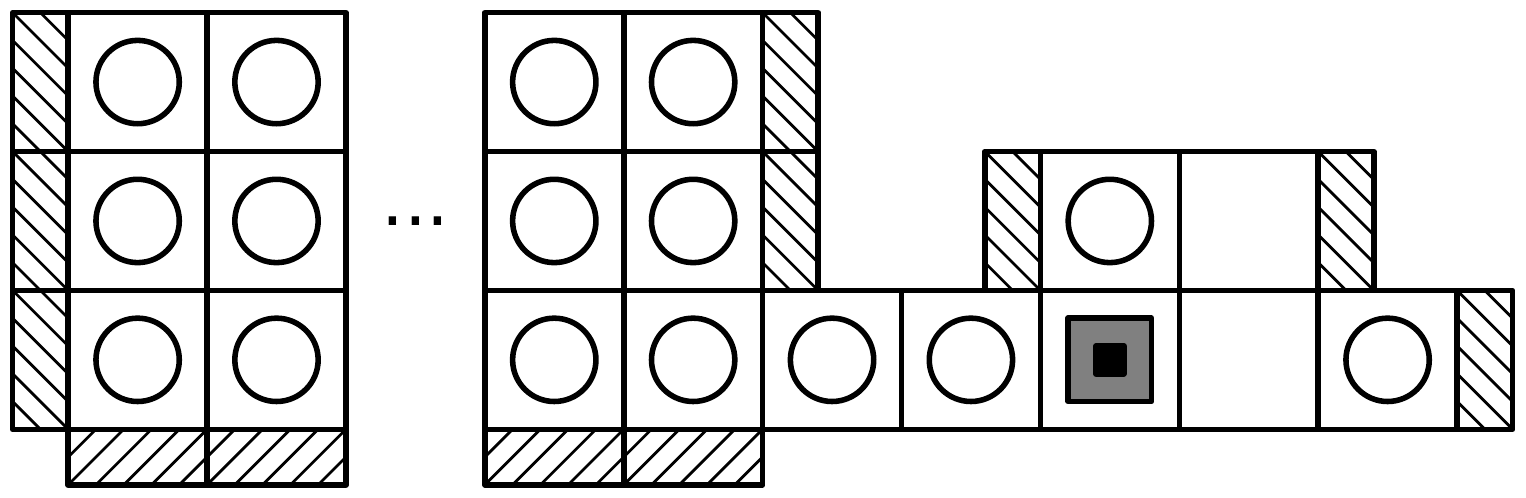}}
  \subcaptionbox{\label{figure:move_wasting_gadget_4} One white turn after (\subref{figure:move_wasting_gadget_3})}
    {\includegraphics[scale=.5]{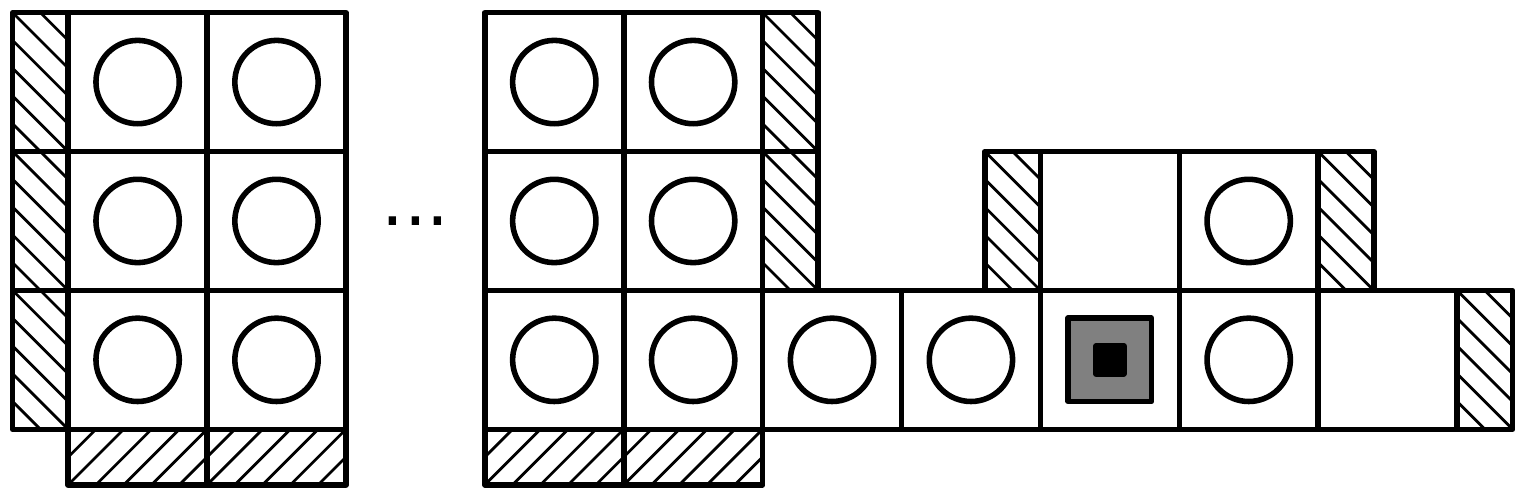}}
  \caption{The move-wasting gadget.}
  \label{figure:move_wasting_gadget}
\end{figure}

\later{

In the initial state, the anchor is on the black king, so it is White's turn. White must move the pawn above the black king to avoid losing next turn. There are only two reachable empty squares, both in the column left of the black king. If the other square in that column remains empty, Black can move the black king into it and push the white pawn in that column off the board. Thus White must fill the other square in that column, and the only way to do so is to move the pawn two columns left of the white king one square right. Figure~\ref{figure:move_wasting_gadget_2} shows the resulting position (after White pushes elsewhere in the board).


Black's only legal push is to the left, resulting in the position shown in Figure~\ref{figure:move_wasting_gadget_3}.


The rightmost four columns in Figure~\ref{figure:move_wasting_gadget_3} are simply the reflection of those columns in Figure~\ref{figure:move_wasting_gadget_1}, so by the same argument White must fill the column to the right of the black king, resulting in Figure~\ref{figure:move_wasting_gadget_4}.


Again, the rightmost four columns of Figures~\ref{figure:move_wasting_gadget_4} and~\ref{figure:move_wasting_gadget_2} are reflections of each other. Black's only legal push is to the right, restoring the gadget to the initial state shown in Figure~\ref{figure:move_wasting_gadget_1}. Thus until White can win in one turn, White must use both moves in the move-wasting gadget, and at all times Black must (and can) push in the move-wasting gadget. In the analysis of the remaining gadgets, if the white king reaches a position from which it cannot push, we conclude that White immediately loses, because if White moves a pawn or the king into position to push, Black can win on the next turn as explained above.

}

\subsection{Variable gadgets}
\abstractlater{\subsection{Variable gadgets}}

\ifabstract
\begin{figure}
  \centering
  \begin{minipage}{0.38\textwidth}
    \centering
    \includegraphics[scale=.4]{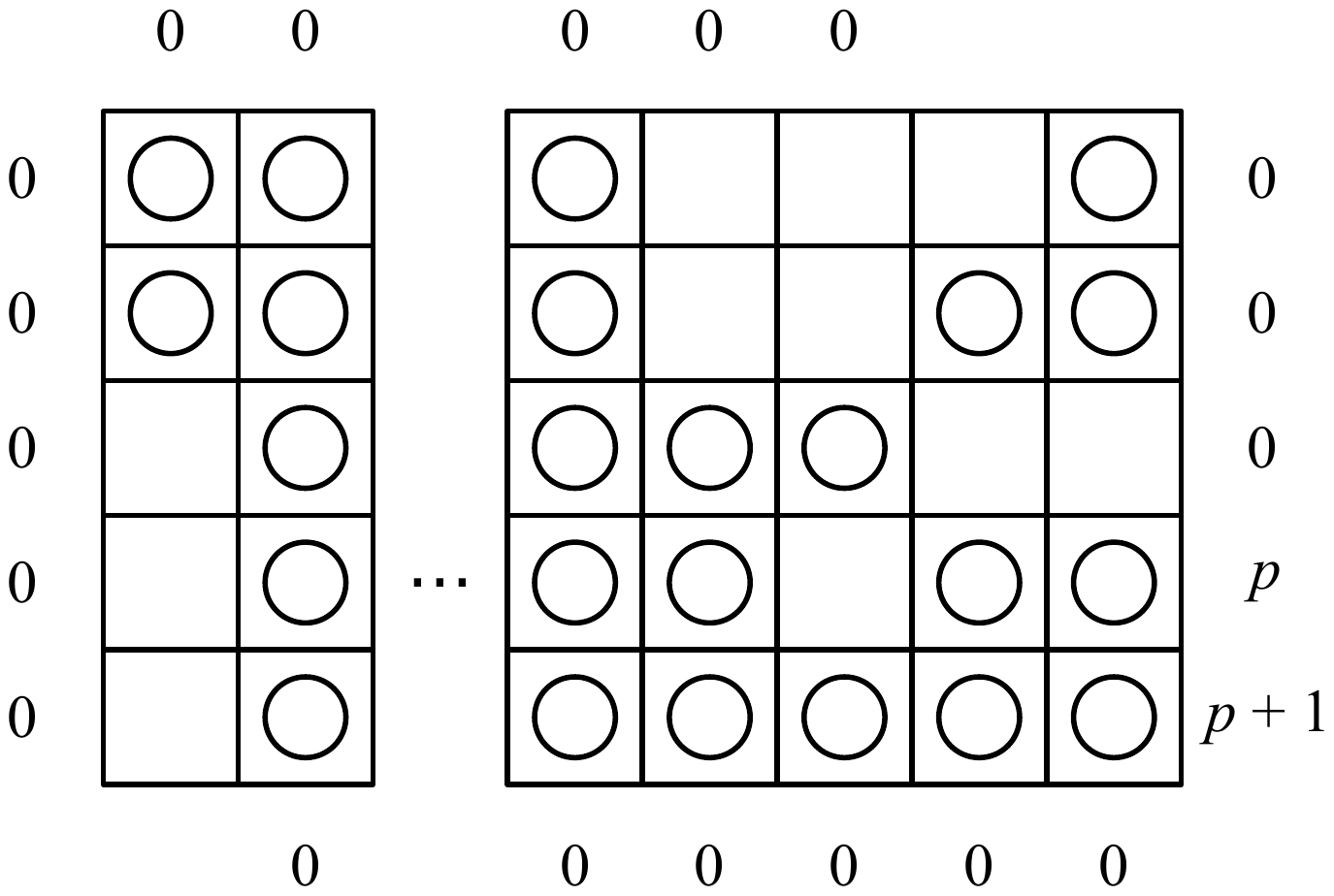}
    \caption{Existential variable gadget.}
    \label{figure:existential_variable_gadget}
  \end{minipage}\hfill
  \begin{minipage}{0.6\textwidth}
    \centering
    \includegraphics[scale=.4]{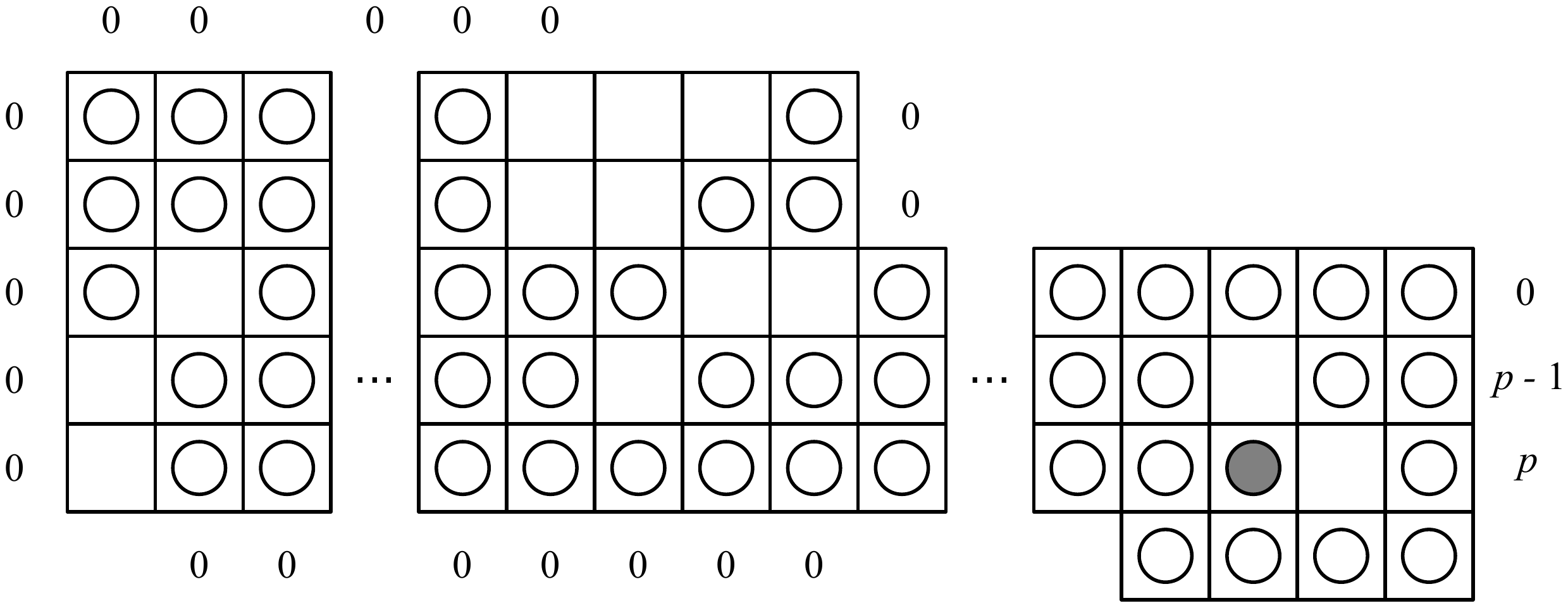}
    \caption{Universal variable gadget.}
    \label{figure:universal_variable_gadget}
  \end{minipage}
\end{figure}
\fi

The existential variable gadget forces White to fill all empty squares in one row of the connection block, corresponding to setting the value of that variable. The universal variable gadget allows Black to choose the value of the corresponding variable, then forces White to similarly fill a row of empty squares.
\later{
We first analyze a core gadget; the existential variable gadget is a minor variant of the core gadget and its correctness follows directly, while the universal variable gadget has an additional component to allow Black to choose the variable's value. Throughout our analysis, we take advantage of the board being filled with white pawns to limit the number of pieces that can leave the gadget.

The core gadget occupies a rectangle of width $p+5$ and height $5$. When instantiated in the reduction, the gadget lies entirely within the \emph{variable gadget I} block. Integer $p$ is one more than the maximum number of occurrences of a literal in the input formula. The initial state of the core gadget is shown in Figure~\ref{figure:variable_gadget_0}. Each number along the boundary of the figure gives the number of empty squares outside the gadget in that direction, and thus an upper bound on the number of pieces that can leave the gadget via that edge.

\begin{figure}
    \centering
    \includegraphics[scale=.5]{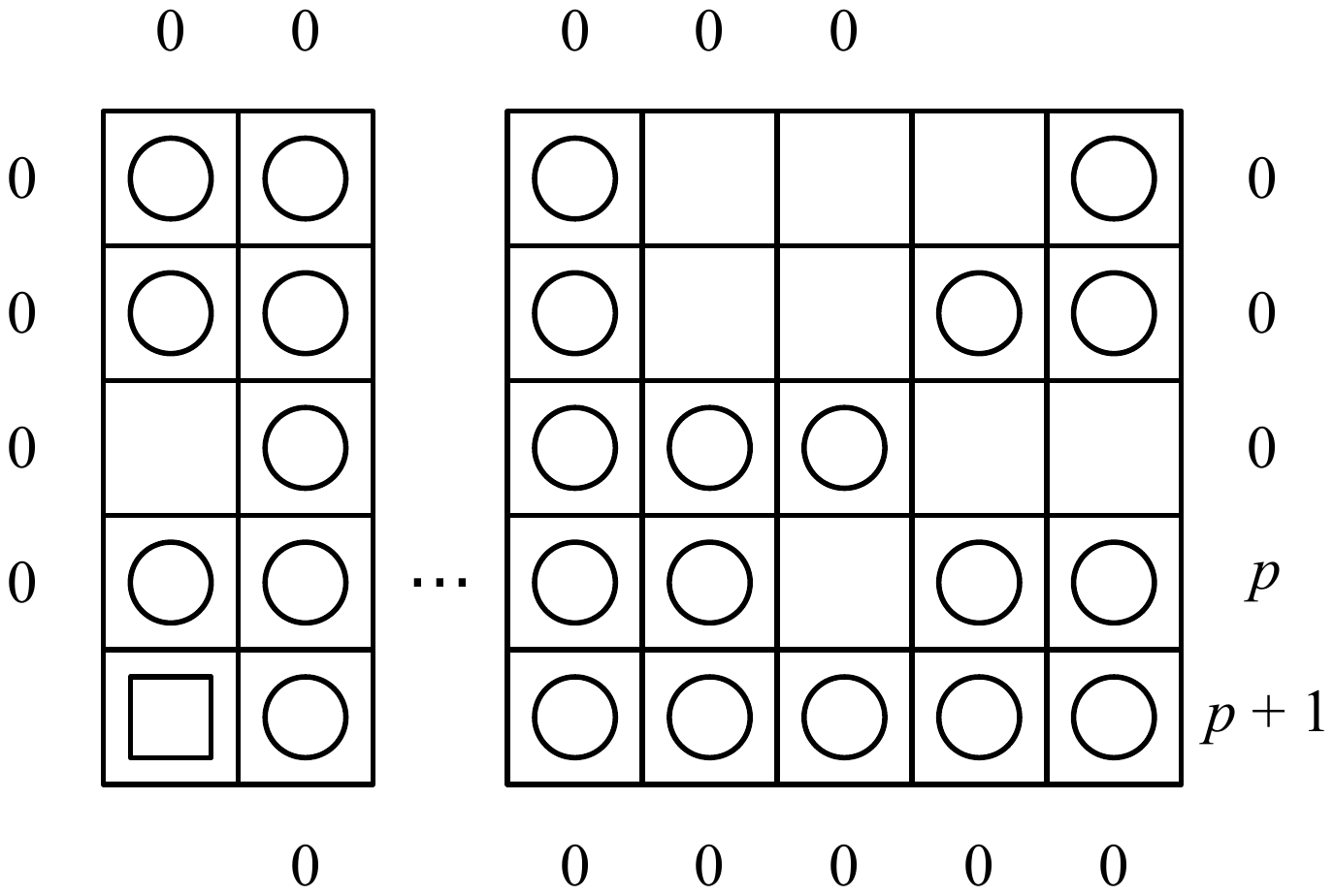}
    \caption{The initial configuration of the core gadget together with upper bounds on the number of pushes out of the gadget at each boundary edge. Omitted columns do not have a given upper bound.}
    \label{figure:variable_gadget_0}
\end{figure}

The following lemma summarizes the constraints we prove about the core gadget.

\begin{lemma}
\label{thm:core}
Starting from the position in Figure~\ref{figure:variable_gadget_0}, and assuming the white king does not push down or left from this position,
\begin{enumerate}[label=(\roman*),ref=\ref{thm:core}(\roman*)]
\item\label{thm:core-chaining} the white king leaves in the second-rightmost column, and
\item when the white king leaves either
\begin{enumerate}[ref=\ref{thm:core}(\roman{enumi}\alph*)]
\item\label{thm:core-true} the gadget is as shown in Figure~\ref{figure:variable_gadget_O1} and $p+1$ white pawns have been pushed out along the bottom row of the gadget, or
\item\label{thm:core-false} the gadget is as shown in Figure~\ref{figure:variable_gadget_O2} and $p$ white pawns have been pushed out along the second-to-bottom row of the gadget,
\end{enumerate}
\item\label{thm:core-integrity} and no other pieces have left the gadget. \xxx{do we use this?}\xxx[Mikhail]{I don't know, but it seems worth including just because if we do this theorem statement fully describes the situation instead of partially describing it.}
\end{enumerate}
\end{lemma}

We will construct the existential and universal variable gadgets from the core gadget such that the assumption holds. Lemma~\ref{thm:core-chaining} ensures we can chain variable gadgets together in sequence without the white king escaping. The outcomes implied by Lemma~\ref{thm:core-true} and \ref{thm:core-false} correspond to setting the variable to true or false (respectively) by filling in the empty squares in the connection block that could be used to satisfy a clause gadget for a clause containing the opposite literal; that is, pushing pawns out along the bottom row of a gadget prevents all negative literals from being used to satisfy a clause, and similarly for the second-to-bottom row and positive literals.

\begin{figure}
  \centering
  \begin{minipage}{0.45\linewidth}
    \centering
    \includegraphics[scale=.5]{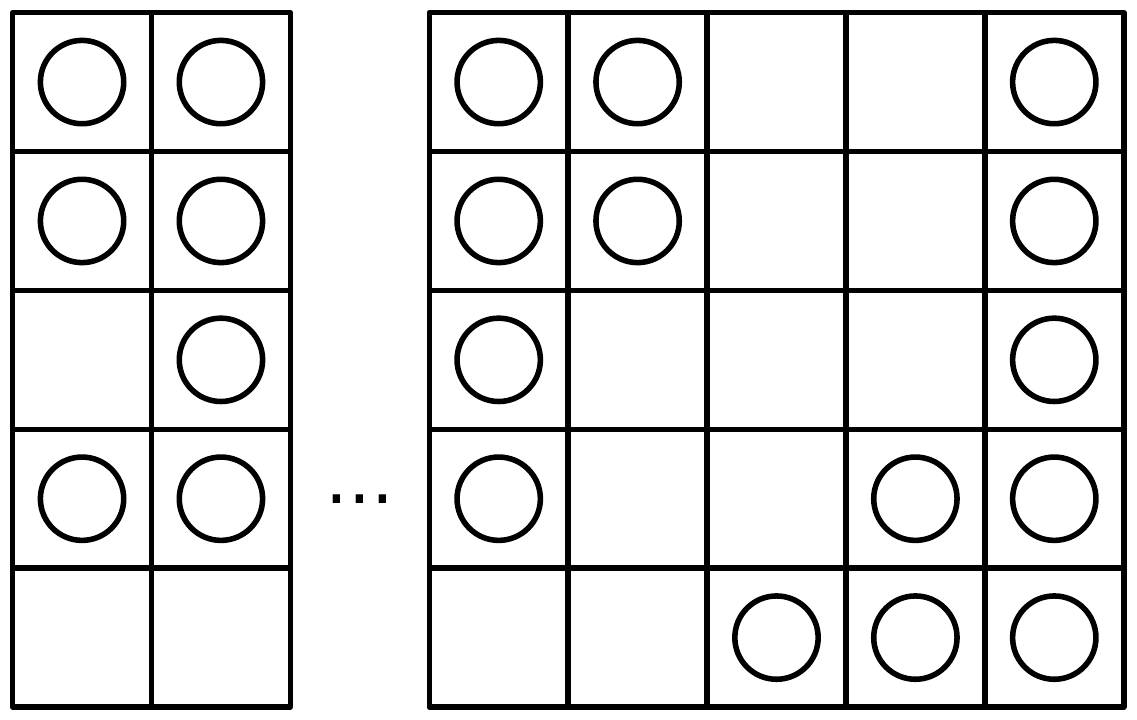}
    \caption{The final configuration of the core gadget after setting the variable to true.}
    \label{figure:variable_gadget_O1}
  \end{minipage}\hfill
  \begin{minipage}{0.45\linewidth}
    \centering
    \includegraphics[scale=.5]{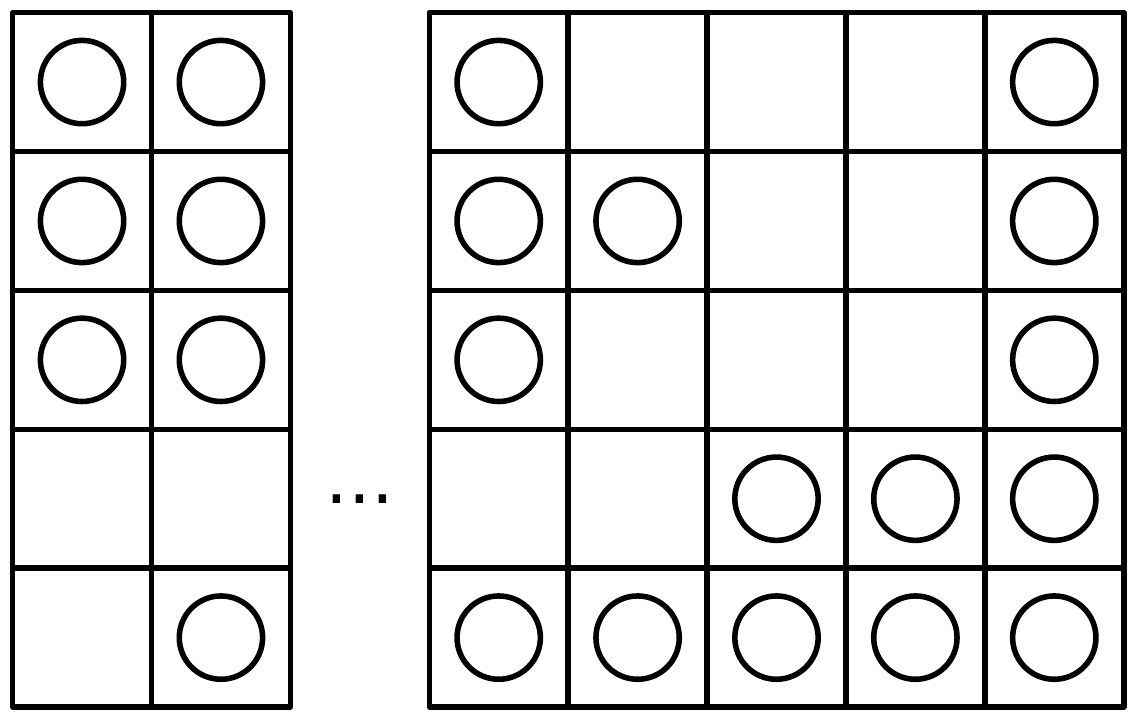}
    \caption{The final configuration of the core gadget after setting the variable to false.}
    \label{figure:variable_gadget_O2}
  \end{minipage}
\end{figure}

\begin{proof}
We proceed by case analysis starting from Figure~\ref{figure:variable_gadget_0}. The move-wasting gadget consumes White's moves, and there are no black pieces in the core gadget, so we need only analyze the sequence of White's pushes.

Suppose the white king first pushes right. Because of the upper bounds along the top and bottom edges of the gadget, the only legal push in the resulting configuration is to the right, and this remains the case until the white king reaches the fourth column from the right of the gadget. At this point $p+1$ pawns have been pushed off the right edge along the bottom row of the gadget, so there are no empty squares remaining in that row, so pushing right is no longer possible and the only legal push is up. Then the only legal push is again up because of the constraints on the left edge of the gadget. Figure~\ref{figure:variable_gadget_1a} shows the result of this sequence of pushes.

\begin{figure}
    \centering
    \raisebox{-.5\height}{\includegraphics[scale=.25]{images/variable_gadget_0}}
    ~~\raisebox{-.5\height}{\scalebox{2}{$\to$}}~~
    \raisebox{-.5\height}{\includegraphics[scale=.25]{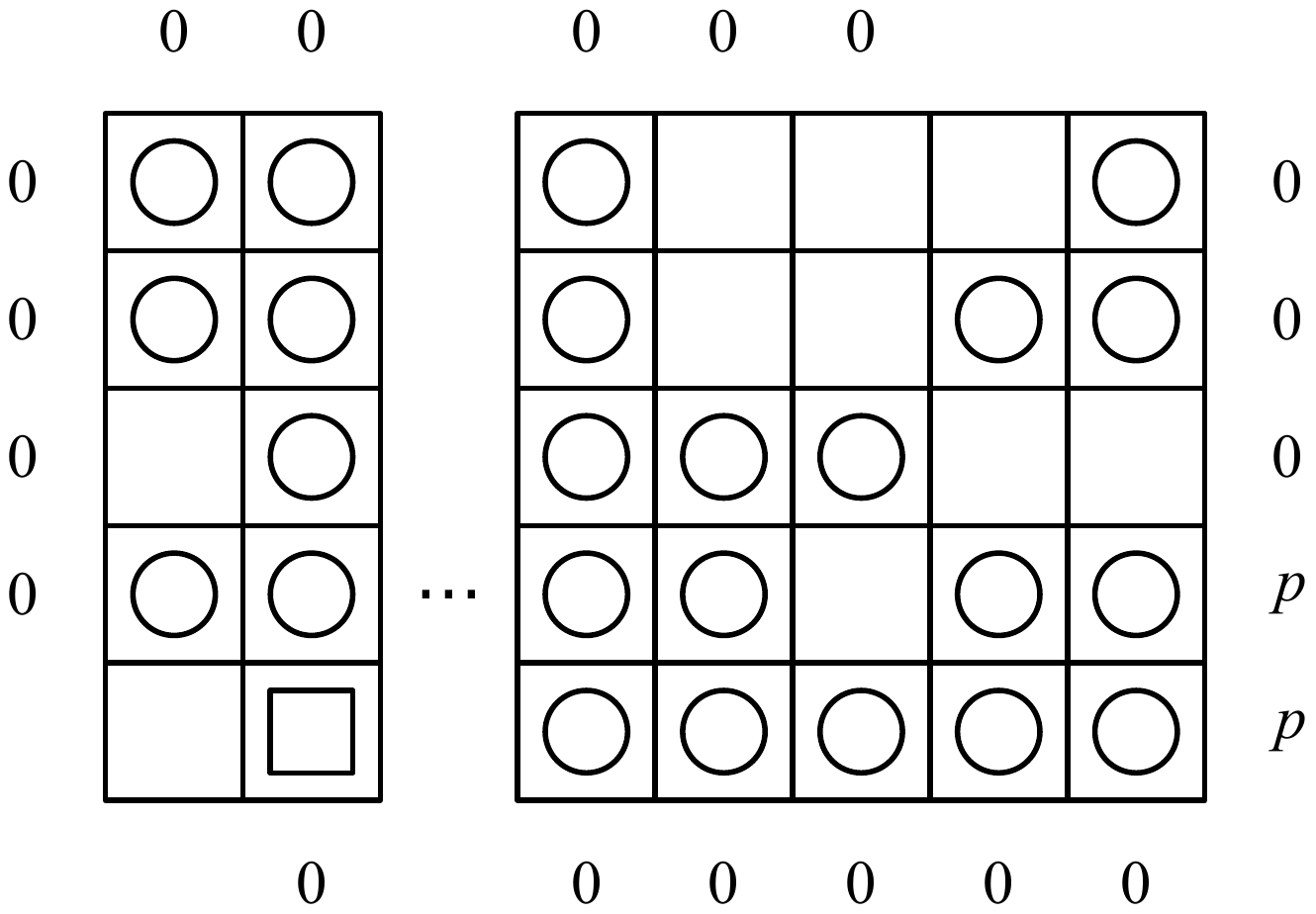}}
    ~~\raisebox{-.5\height}{\scalebox{2}{$\to$}}
    \raisebox{-.5\height}{$\cdots$}
    \raisebox{-.5\height}{\scalebox{2}{$\to$}}~~
    \raisebox{-.5\height}{\includegraphics[scale=.25]{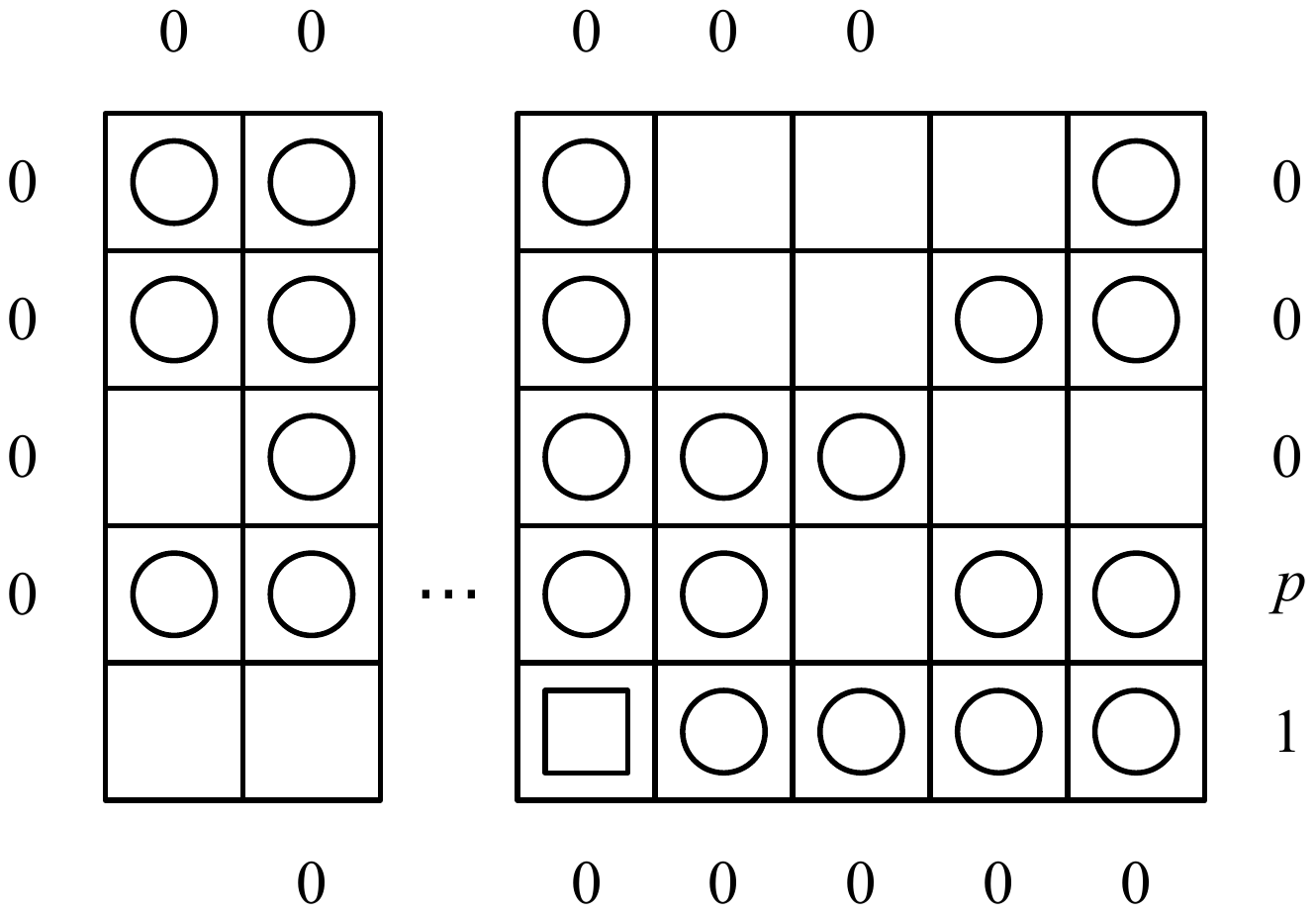}} 
    ~~\raisebox{-.5\height}{\scalebox{2}{$\to$}}~~
    \raisebox{-.5\height}{\includegraphics[scale=.25]{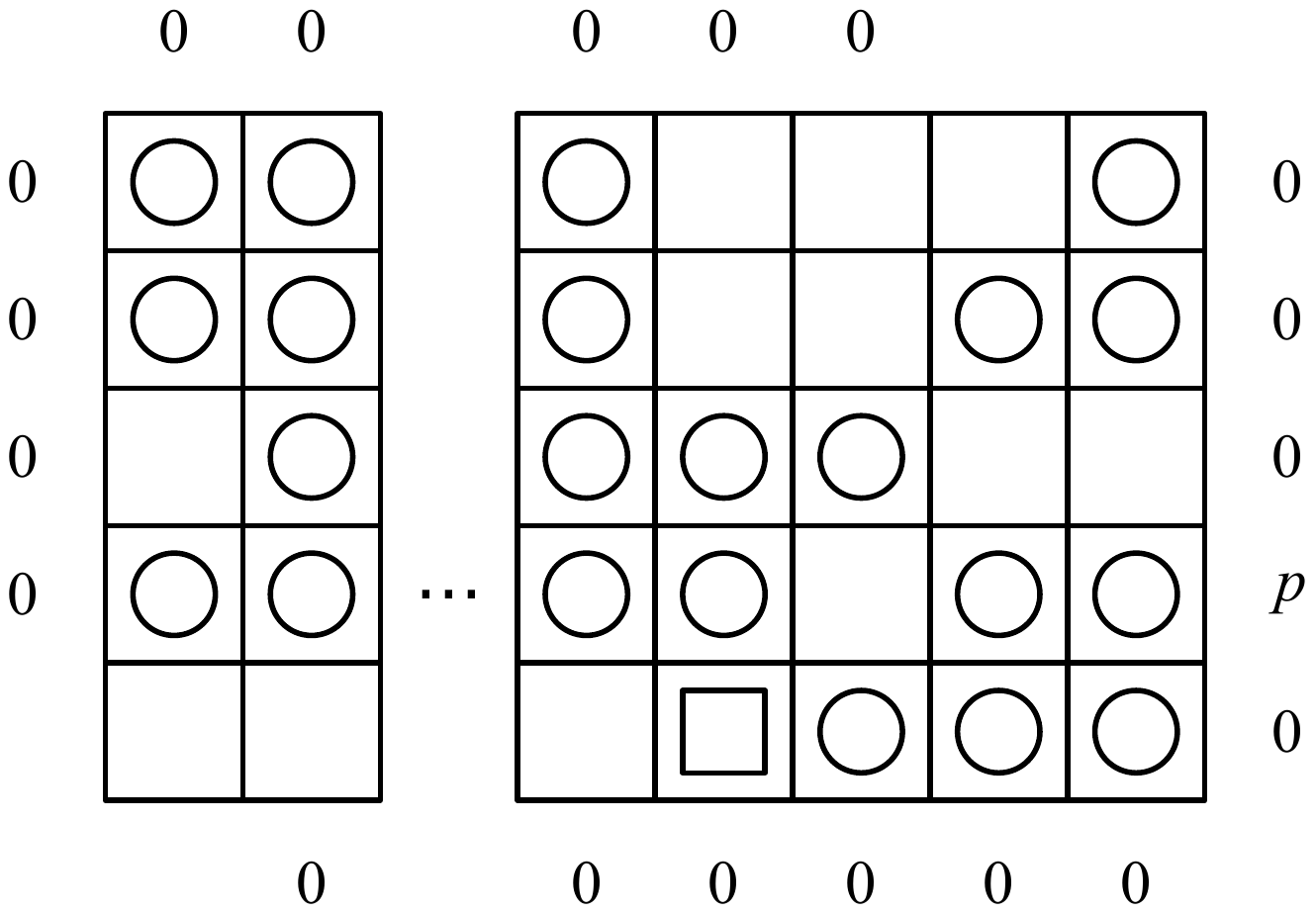}}
    ~~\raisebox{-.5\height}{\scalebox{2}{$\to$}}~~
    \raisebox{-.5\height}{\includegraphics[scale=.25]{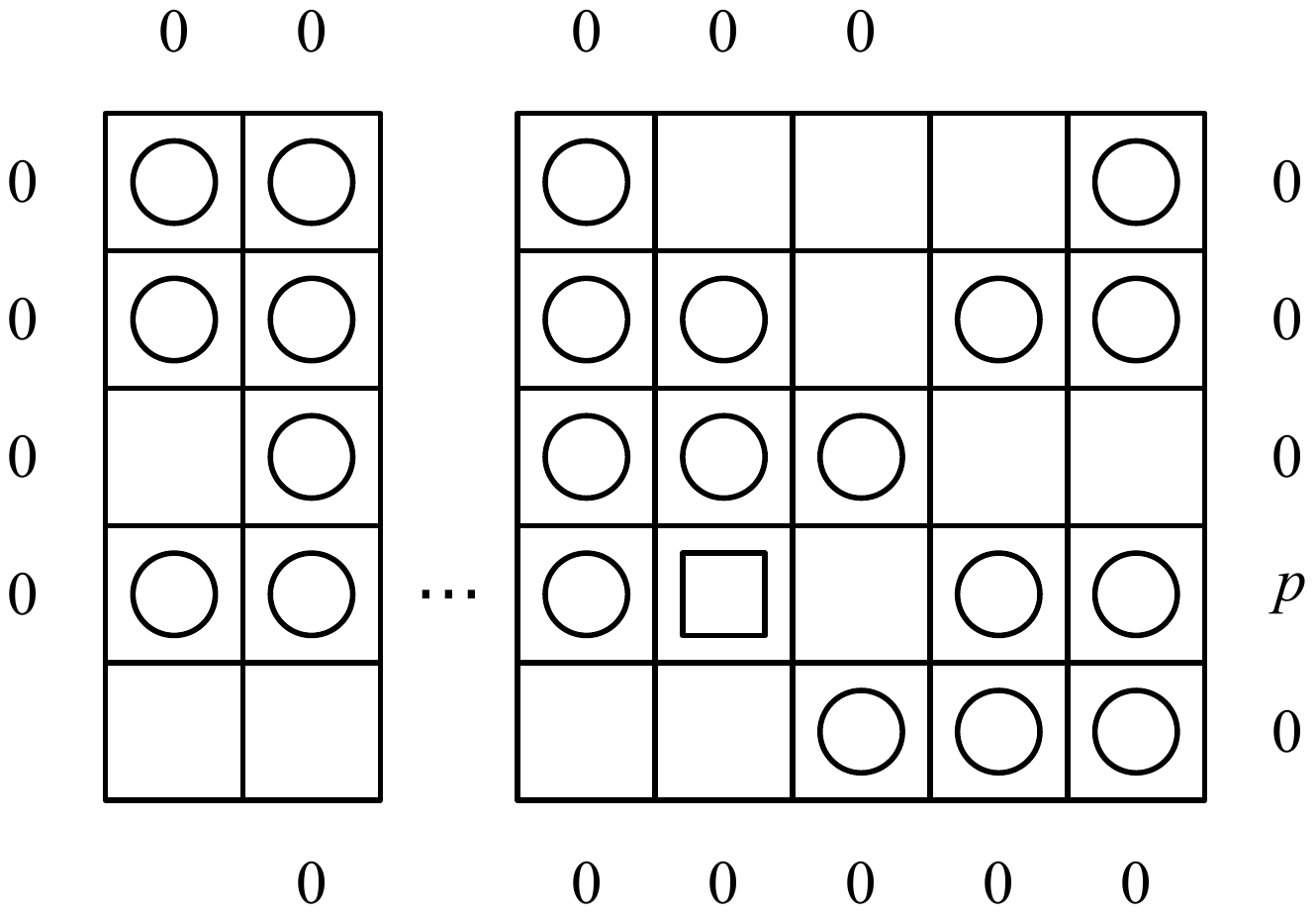}}
    ~~\raisebox{-.5\height}{\scalebox{2}{$\to$}}~~
    \raisebox{-.5\height}{\includegraphics[scale=.25]{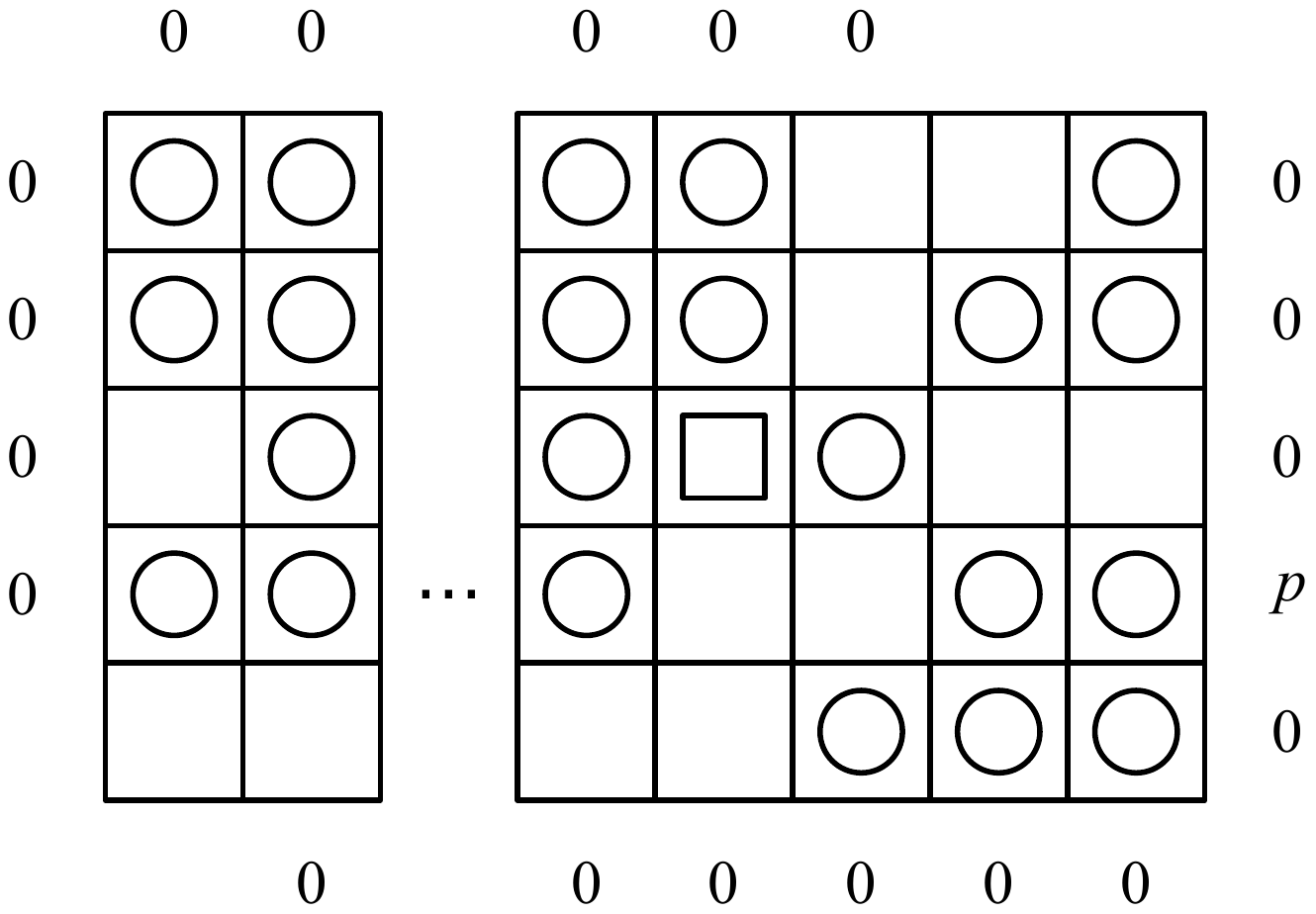}}
    \caption{One possible push sequence starting from the initial state of the core gadget.}
    \label{figure:variable_gadget_1a}
\end{figure}

If the white king pushes left from this position, the only possible next push is down, after which there are no legal pushes, resulting in a loss for White. Figure~\ref{figure:variable_gadget_2aa} shows this sequence of pushes.

\begin{figure}
    \centering
    \raisebox{-.5\height}{\includegraphics[scale=.25]{images/variable_gadget_5a}}
    ~~\raisebox{-.5\height}{\scalebox{2}{$\to$}}~~
    \raisebox{-.5\height}{\includegraphics[scale=.25]{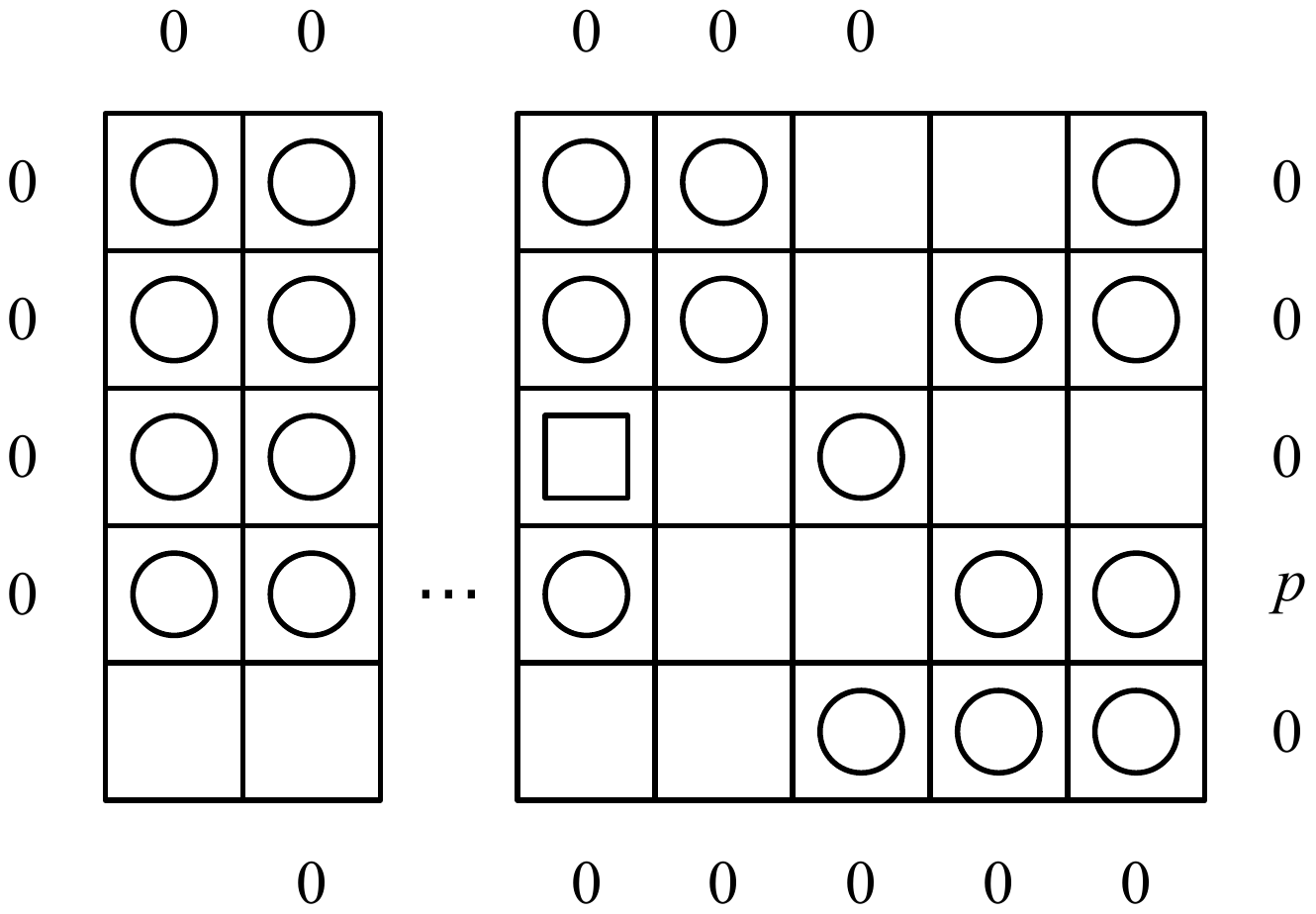}}
    ~~\raisebox{-.5\height}{\scalebox{2}{$\to$}}~~
    \raisebox{-.5\height}{\includegraphics[scale=.25]{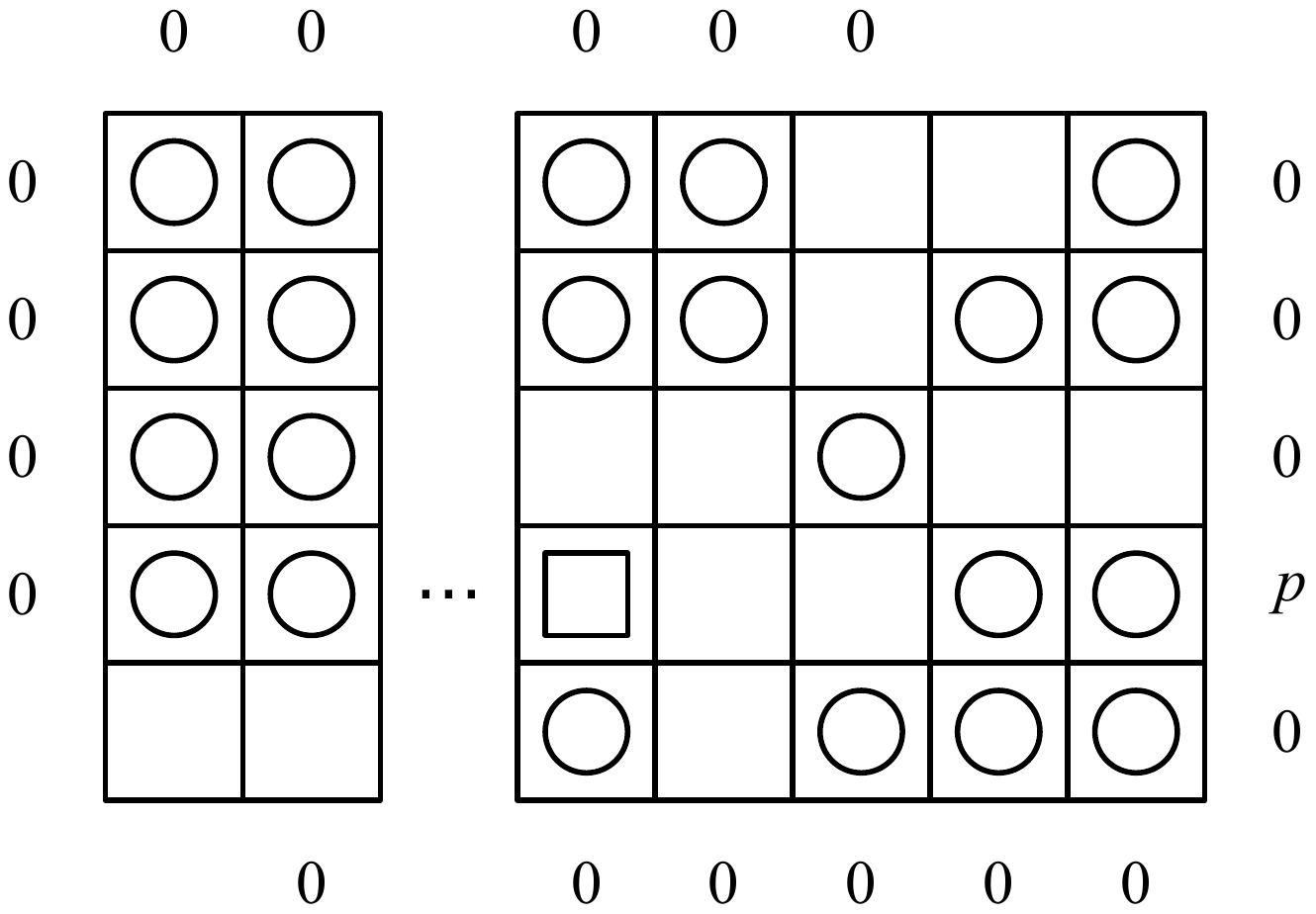}}
    \caption{The result of pushing left and down from the last position in Figure~\ref{figure:variable_gadget_2aa}. White has no legal pushes in the final position.}
    \label{figure:variable_gadget_2aa}
\end{figure}

The only other legal push from the last position in Figure~\ref{figure:variable_gadget_1a} is to the right, after which pushes right, up, up and up again are the only legal pushes. This sequence results in the white king, preceded by a white pawn, exiting the top of the gadget in the second-rightmost column, as desired by Lemma~\ref{thm:core-chaining}. Figure~\ref{figure:variable_gadget_2ab} shows the positions resulting from this sequence. The final position reached is the position in Figure~\ref{figure:variable_gadget_O1}, $p+1$ pawns were pushed out of the gadget to the right along the bottom row, as desired by Lemma~\ref{thm:core-true}, and and no other pieces were pushed out of the gadget, as desired by Lemma~\ref{thm:core-integrity}.

\begin{figure}
    \centering
    \raisebox{-.5\height}{\includegraphics[scale=.25]{images/variable_gadget_5a}}
    ~~\raisebox{-.5\height}{\scalebox{2}{$\to$}}~~
    \raisebox{-.5\height}{\includegraphics[scale=.25]{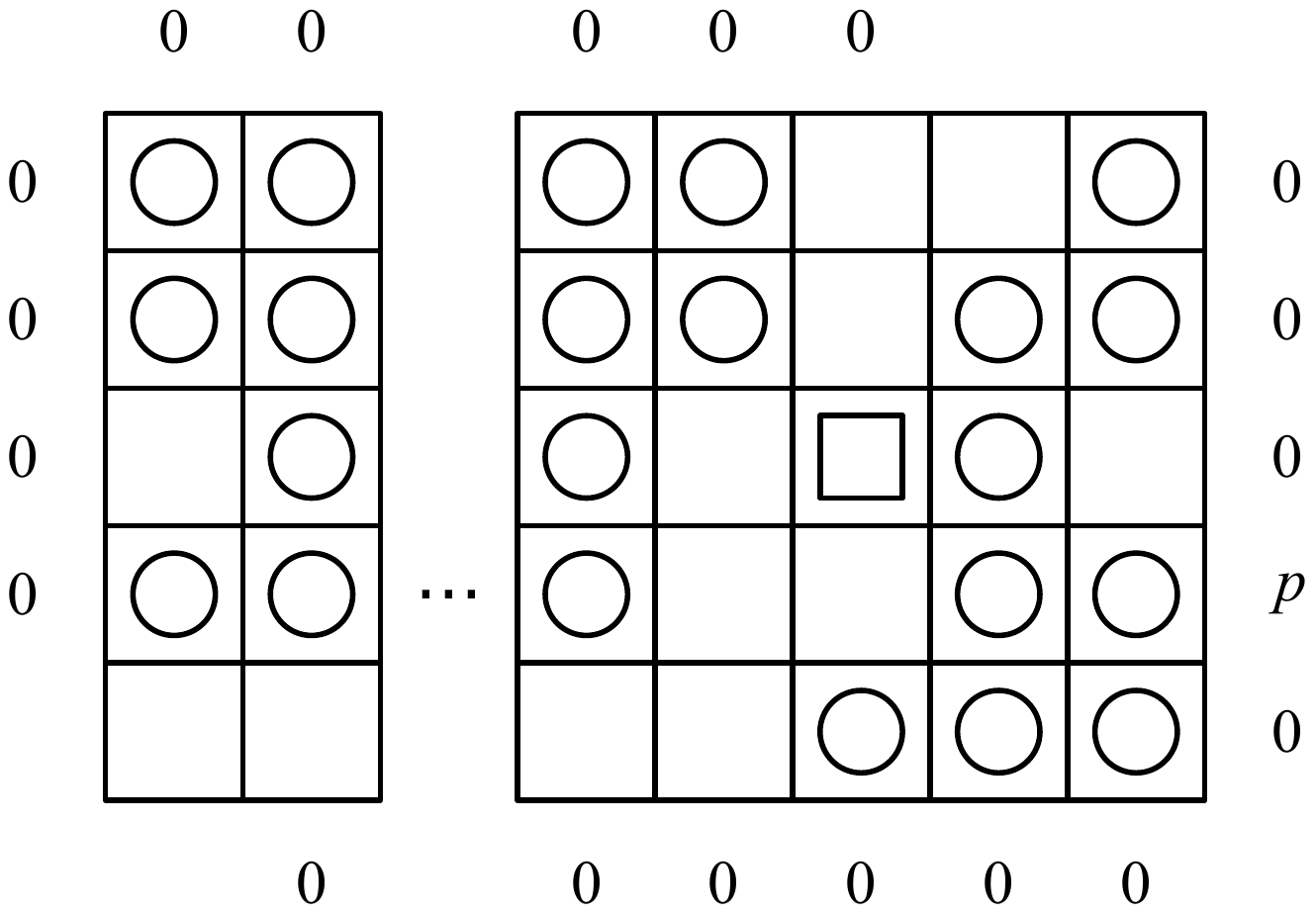}}
    ~~\raisebox{-.5\height}{\scalebox{2}{$\to$}}~~
    \raisebox{-.5\height}{\includegraphics[scale=.25]{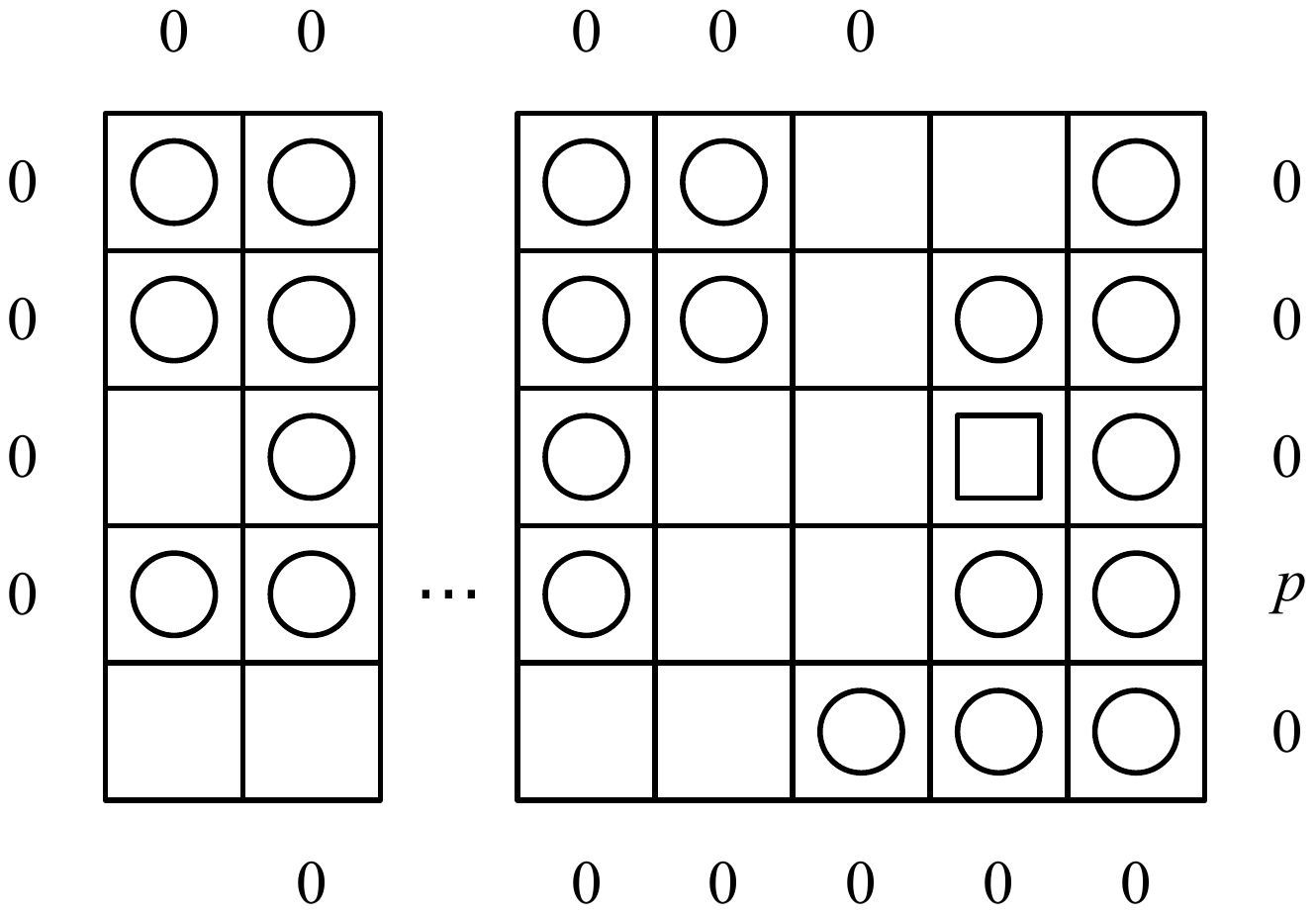}}
    ~~\raisebox{-.5\height}{\scalebox{2}{$\to$}}~~
    \raisebox{-.5\height}{\includegraphics[scale=.25]{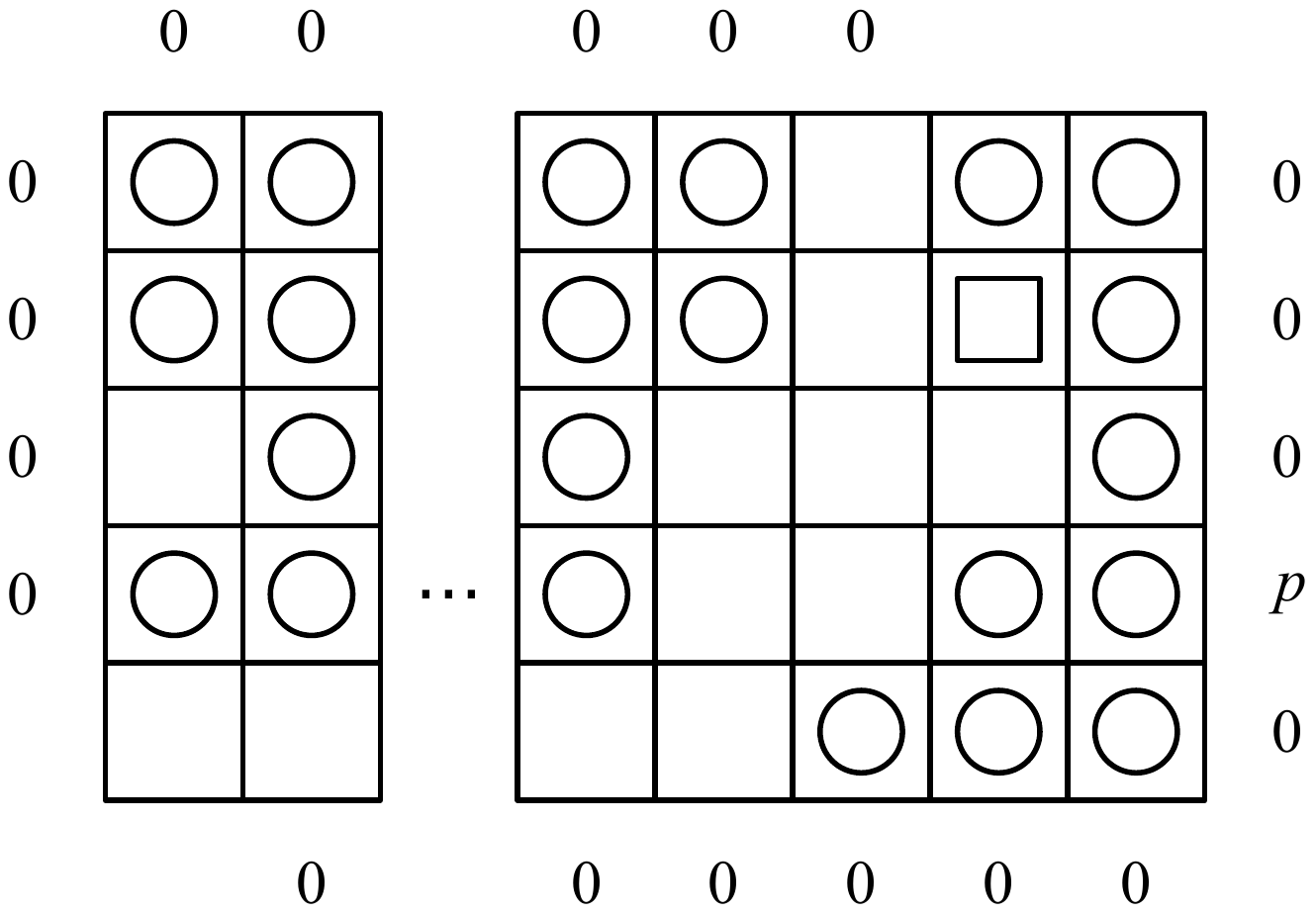}}
    ~~\raisebox{-.5\height}{\scalebox{2}{$\to$}}~~
    \raisebox{-.5\height}{\includegraphics[scale=.25]{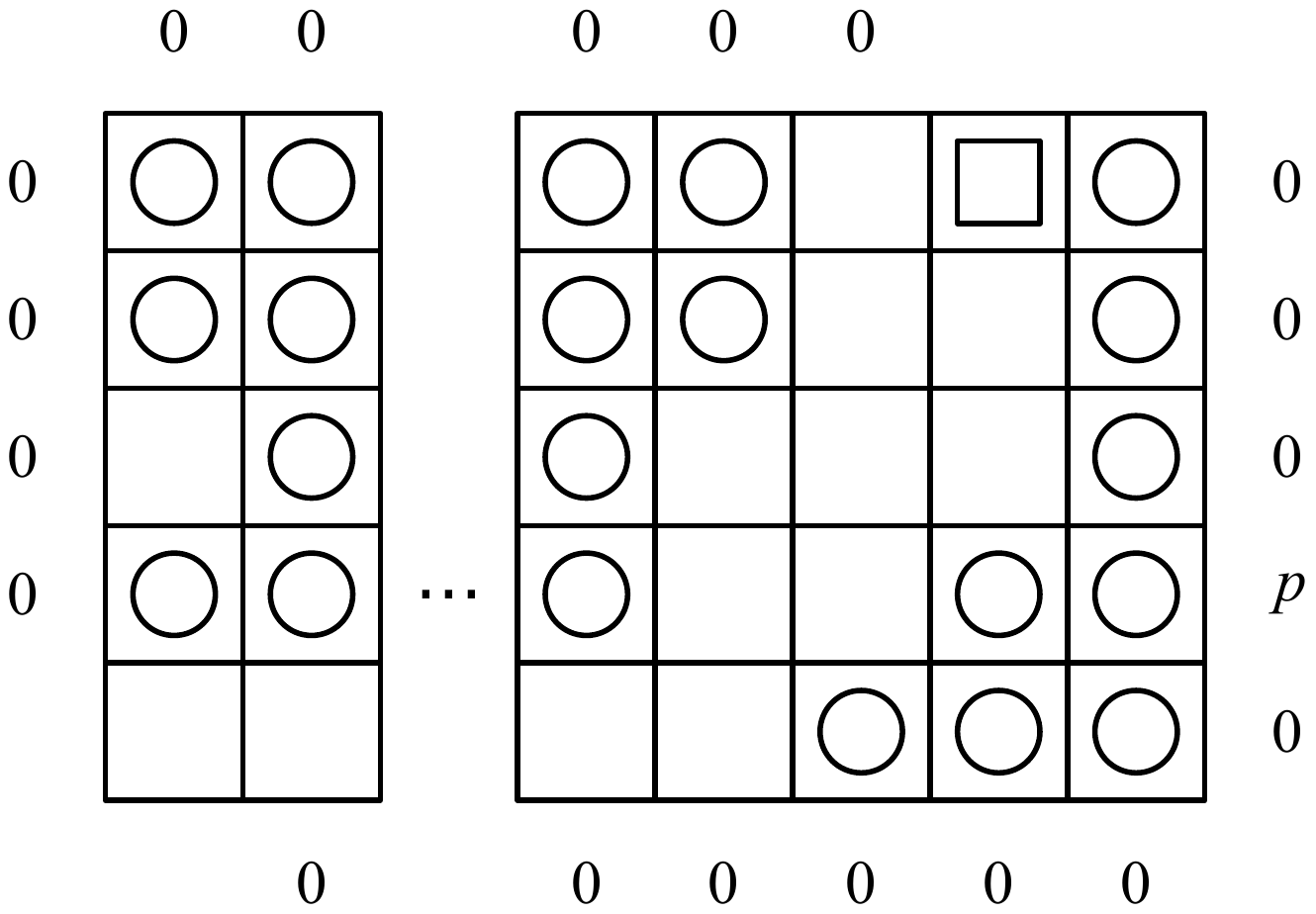}}
    ~~\raisebox{-.5\height}{\scalebox{2}{$\to$}}~~
    \raisebox{-.5\height}{\includegraphics[scale=.25]{images/variable_gadget_O1}}
    \caption{The result of pushing right from the last position in Figure~\ref{figure:variable_gadget_1a}, reaching the position in Figure~\ref{figure:variable_gadget_O1}.}
    \label{figure:variable_gadget_2ab}
\end{figure}

Now suppose that the white king pushes up from the initial configuration. Because of the constraints on the gadget boundary, the only legal push is to the right until the white king reaches the fourth column from the right of the gadget. At this point $p$ pawns have been pushed off the right edge along the second-to-bottom row of the gadget, so there are no empty squares remaining in that row, so pushing right is no longer possible and the only legal push is up. Then the only legal push is again up because of the constraints on the left edge of the gadget. Figure~\ref{figure:variable_gadget_1b} shows the result of this sequence of pushes.

\begin{figure}
    \centering
    \raisebox{-.5\height}{\includegraphics[scale=.25]{images/variable_gadget_0}}
    ~~\raisebox{-.5\height}{\scalebox{2}{$\to$}}~~
    \raisebox{-.5\height}{\includegraphics[scale=.25]{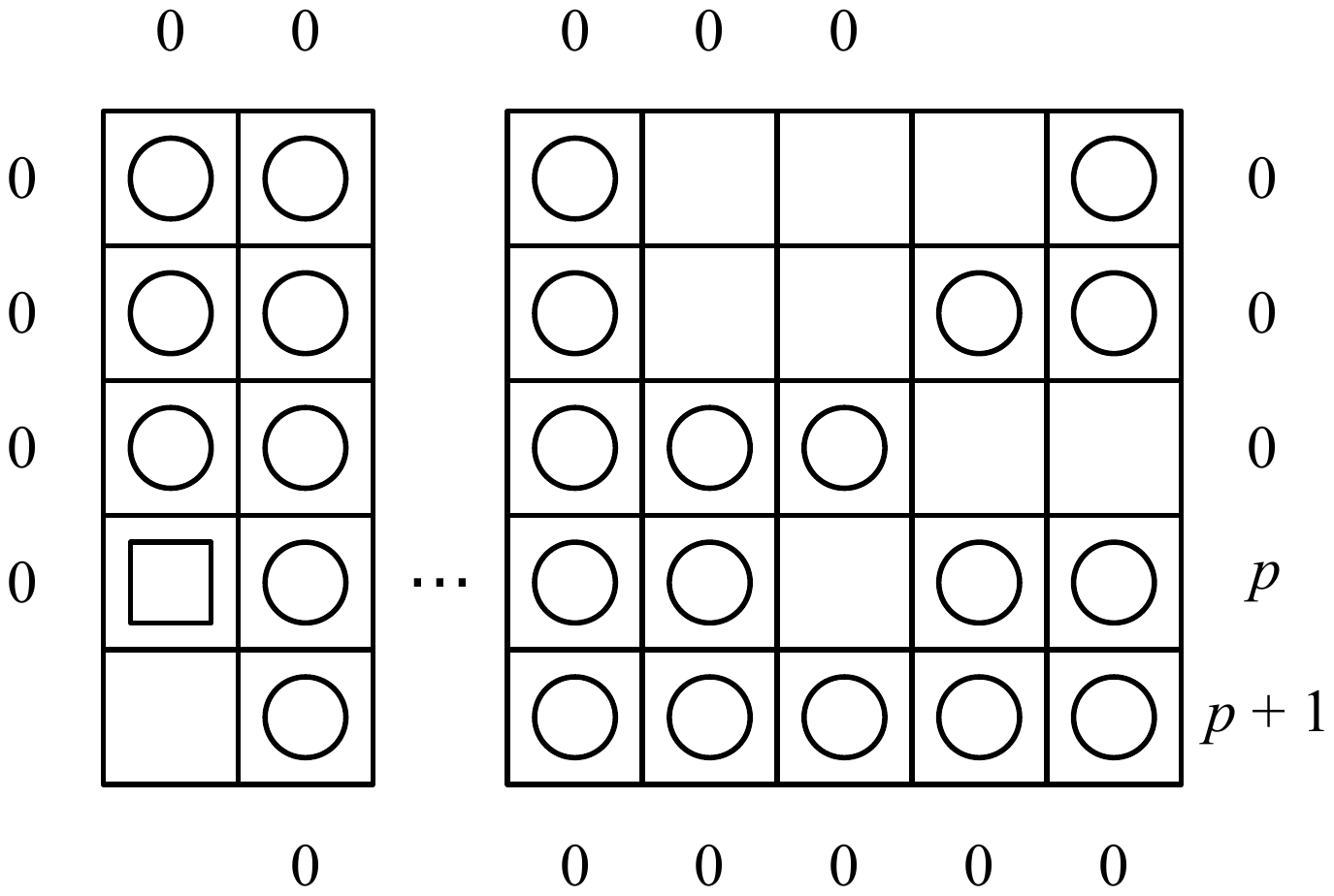}}
    ~~\raisebox{-.5\height}{\scalebox{2}{$\to$}}~~
    \raisebox{-.5\height}{\includegraphics[scale=.25]{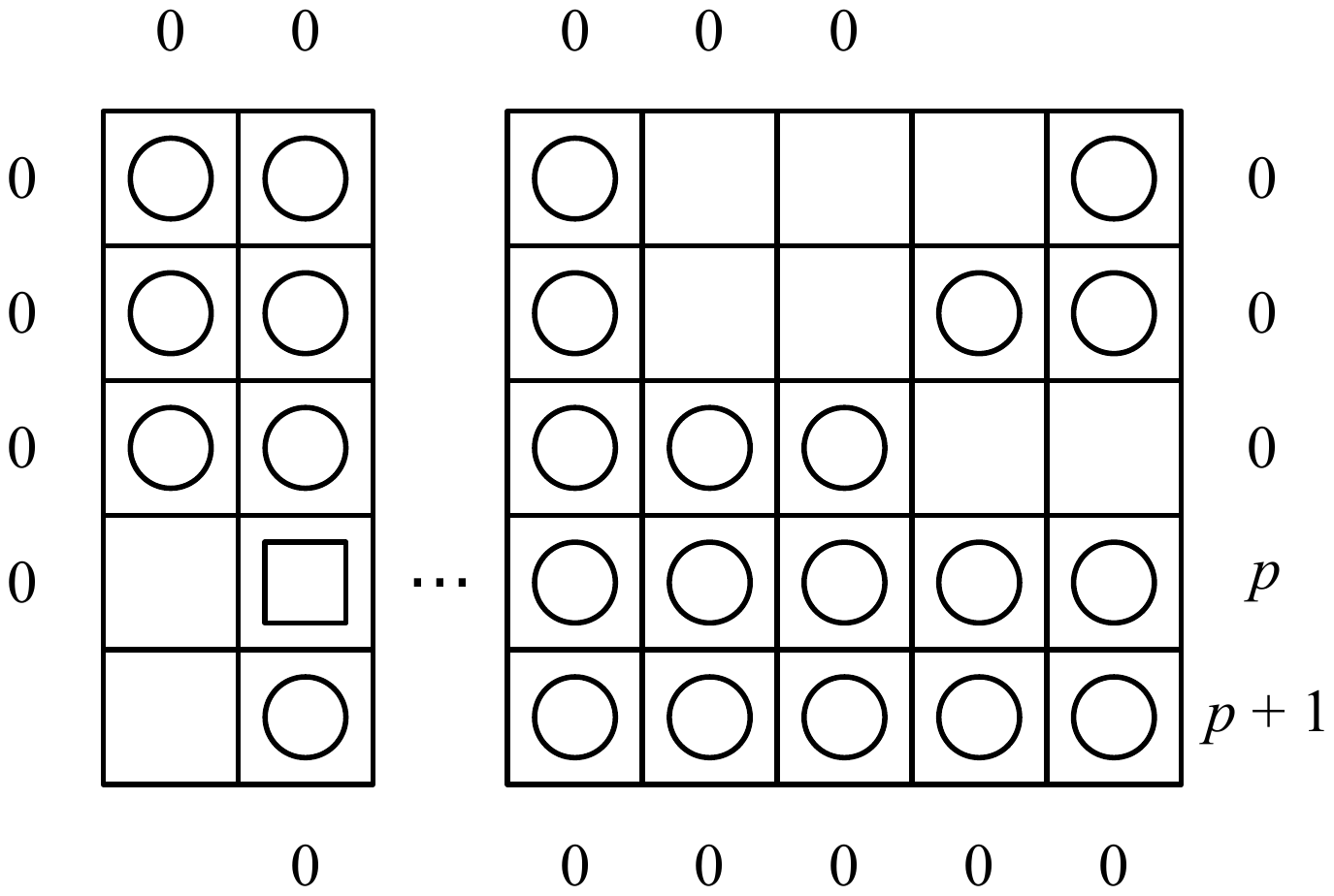}}
    ~~\raisebox{-.5\height}{\scalebox{2}{$\to$}}
    \raisebox{-.5\height}{$\cdots$}
    \raisebox{-.5\height}{\scalebox{2}{$\to$}}~~
    \raisebox{-.5\height}{\includegraphics[scale=.25]{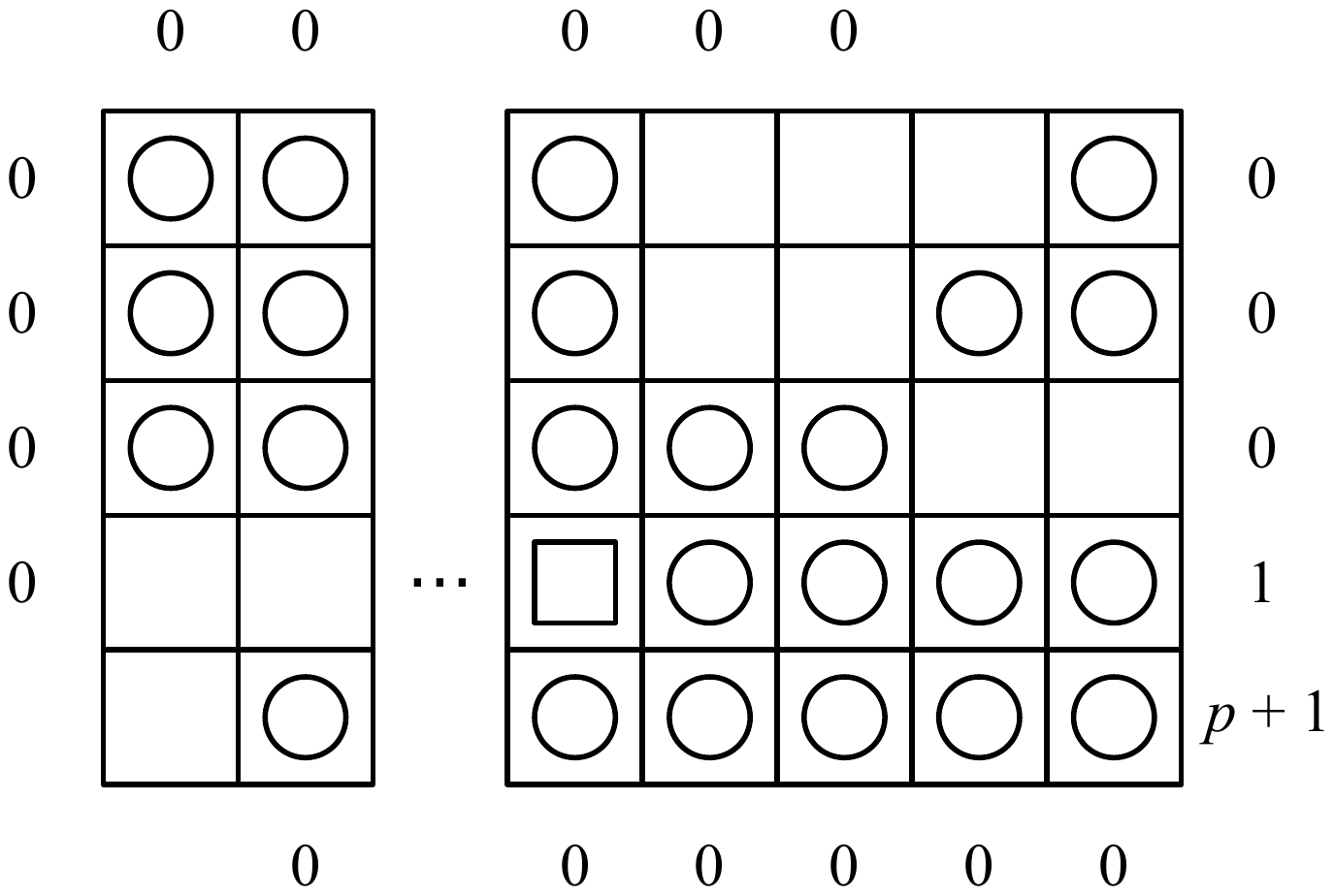}}
    ~~\raisebox{-.5\height}{\scalebox{2}{$\to$}}~~
    \raisebox{-.5\height}{\includegraphics[scale=.25]{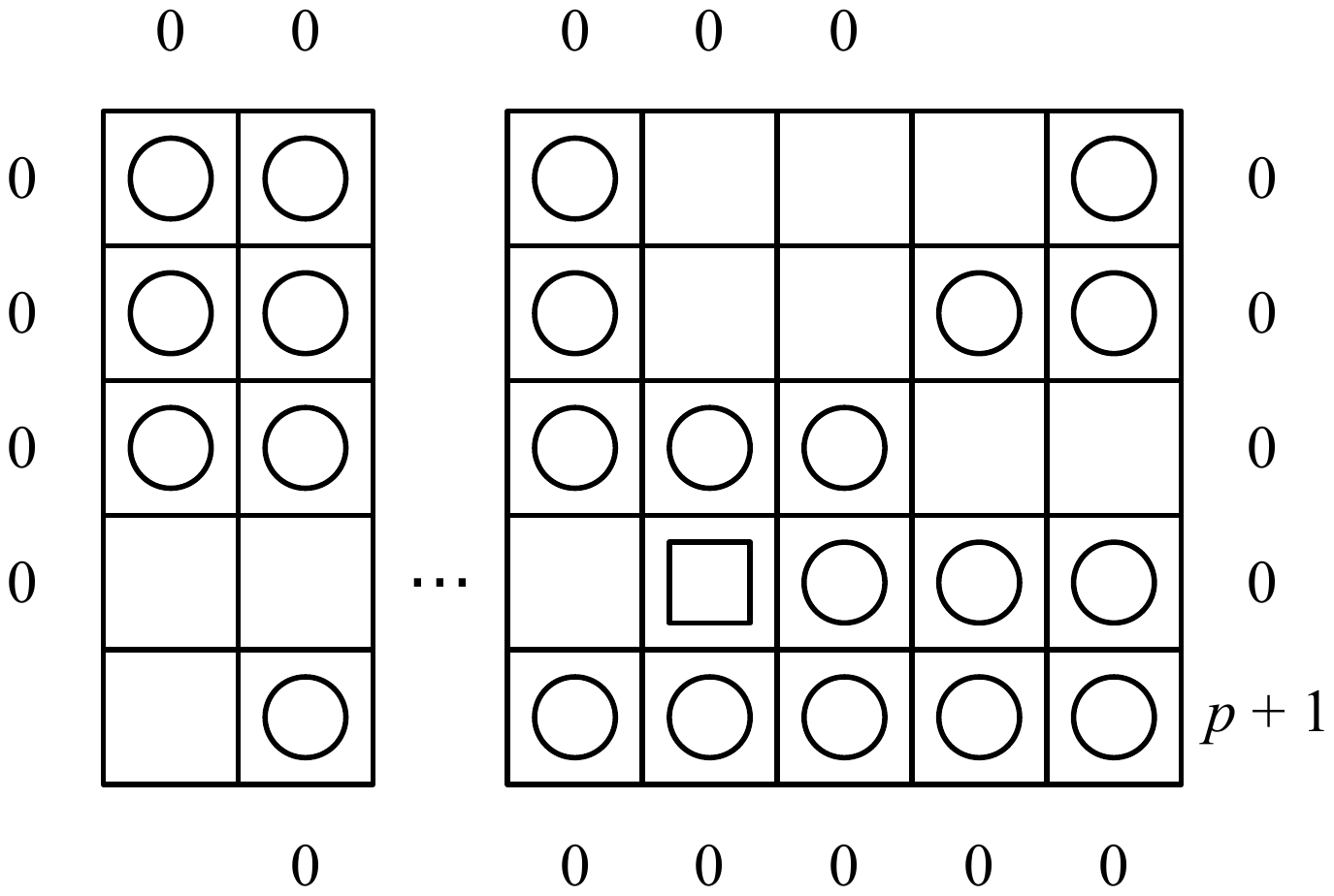}}
    ~~\raisebox{-.5\height}{\scalebox{2}{$\to$}}~~
    \raisebox{-.5\height}{\includegraphics[scale=.25]{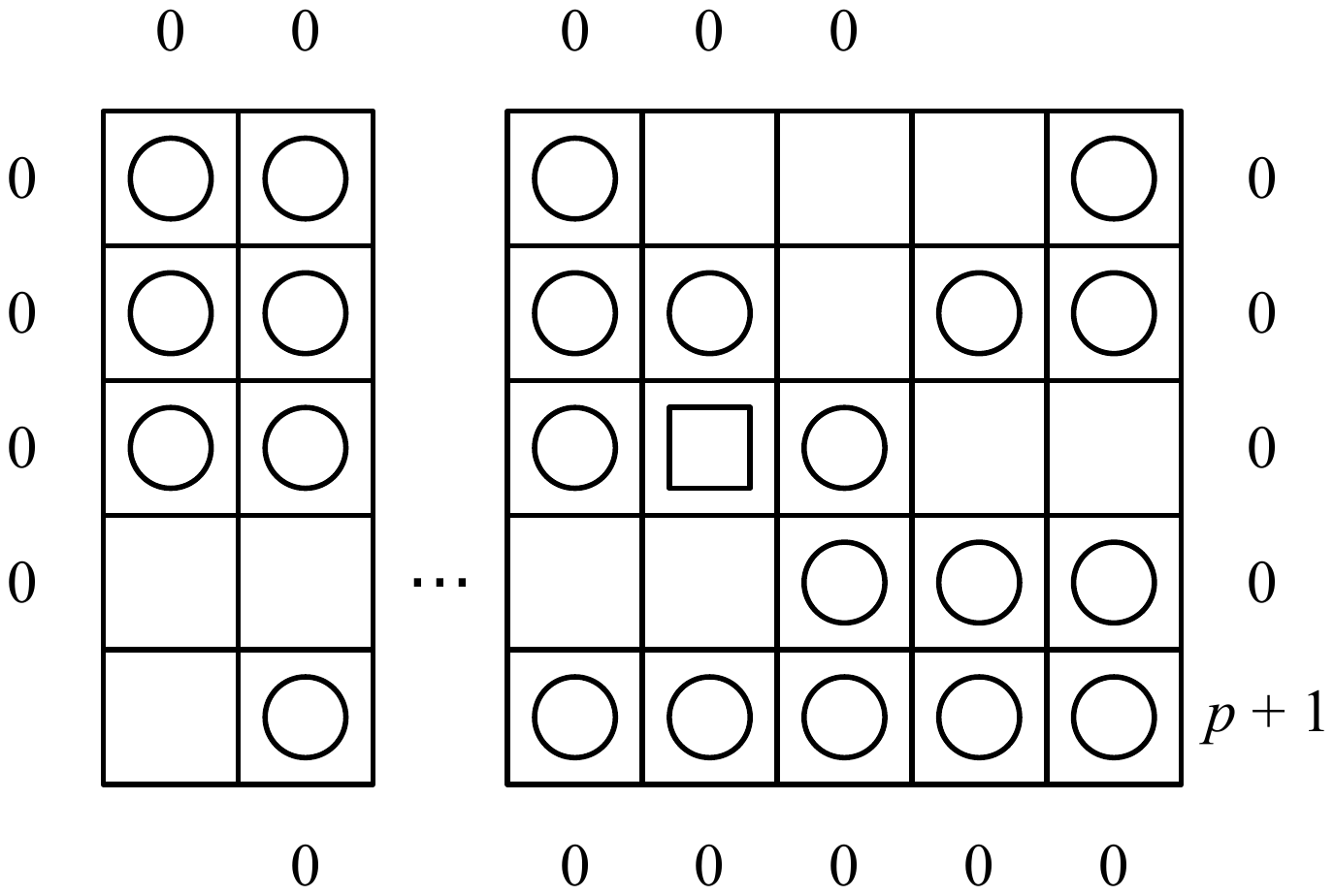}}
    \caption{The other possible push sequence starting from the initial state of the core gadget.}
    \label{figure:variable_gadget_1b}
\end{figure}

If the white king pushes up from this position, there are no legal pushes in the resulting position, resulting in a loss for White. Figure~\ref{figure:variable_gadget_2ba} shows this push and the resulting losing position.

\begin{figure}
    \centering
    \raisebox{-.5\height}{\includegraphics[scale=.25]{images/variable_gadget_5b}}
    ~~\raisebox{-.5\height}{\scalebox{2}{$\to$}}~~
    \raisebox{-.5\height}{\includegraphics[scale=.25]{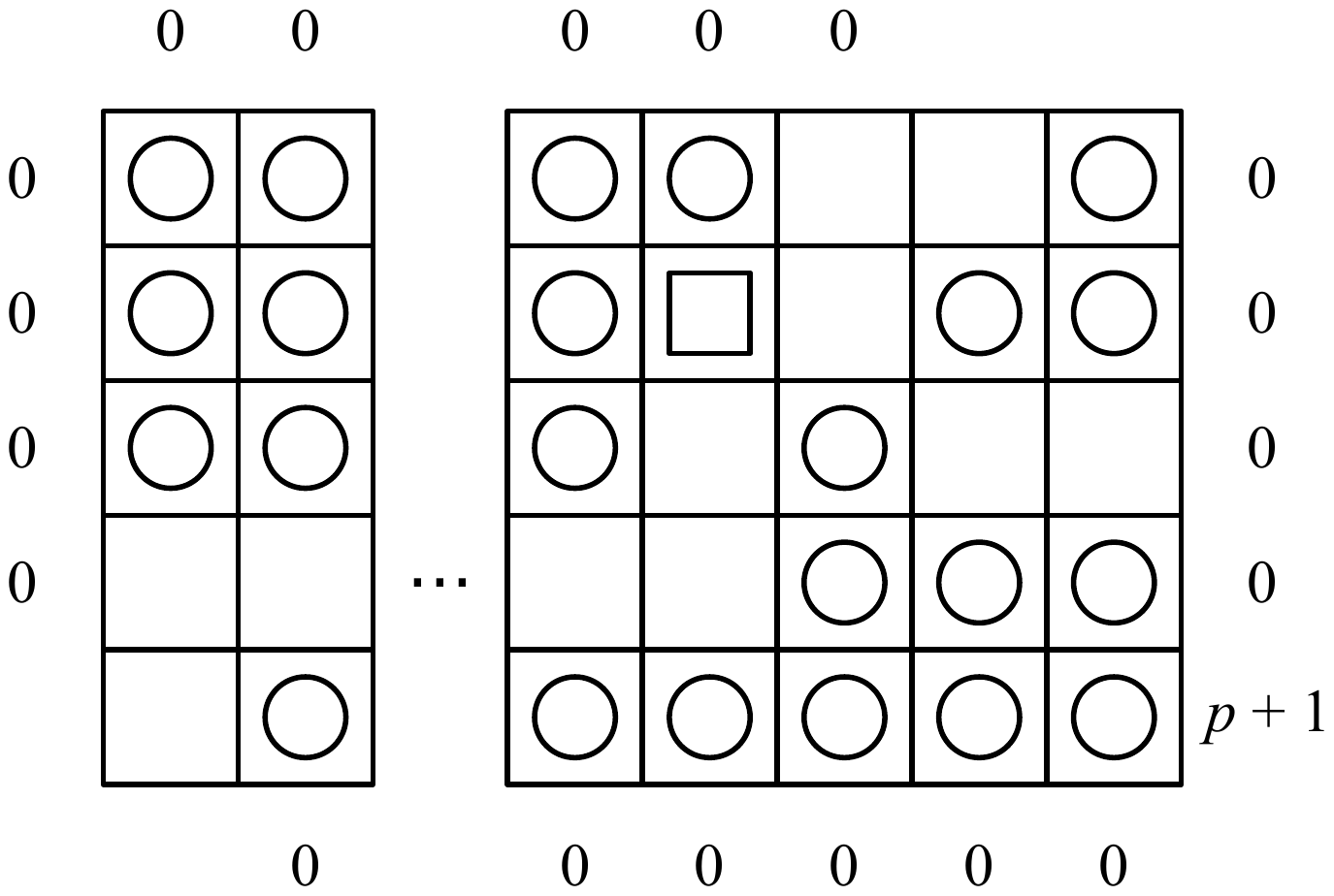}}
    \caption{The result of pushing up from the last position in Figure~\ref{figure:variable_gadget_1b}. White has no legal pushes in the final position.}
    \label{figure:variable_gadget_2ba}
\end{figure}

The only other legal push from the last position in Figure~\ref{figure:variable_gadget_1b} is to the right, after which pushes right, up, up and up again are the only legal pushes. This sequence results in the white king, preceded by a white pawn, exiting the top of the gadget in the second-rightmost column, as desired by Lemma~\ref{thm:core-chaining}. Figure~\ref{figure:variable_gadget_2bb} shows the positions resulting from this sequence. The final position reached is the position in Figure~\ref{figure:variable_gadget_O2}, and $p$ pawns were pushed out of the gadget to the right along the second-to-bottom row, as desired by Lemma~\ref{thm:core-false}. No other pieces were pushed out of the gadget, as desired by Lemma~\ref{thm:core-integrity}.

\begin{figure}
    \centering
    \raisebox{-.5\height}{\includegraphics[scale=.25]{images/variable_gadget_5b}}
    ~~\raisebox{-.5\height}{\scalebox{2}{$\to$}}~~
    \raisebox{-.5\height}{\includegraphics[scale=.25]{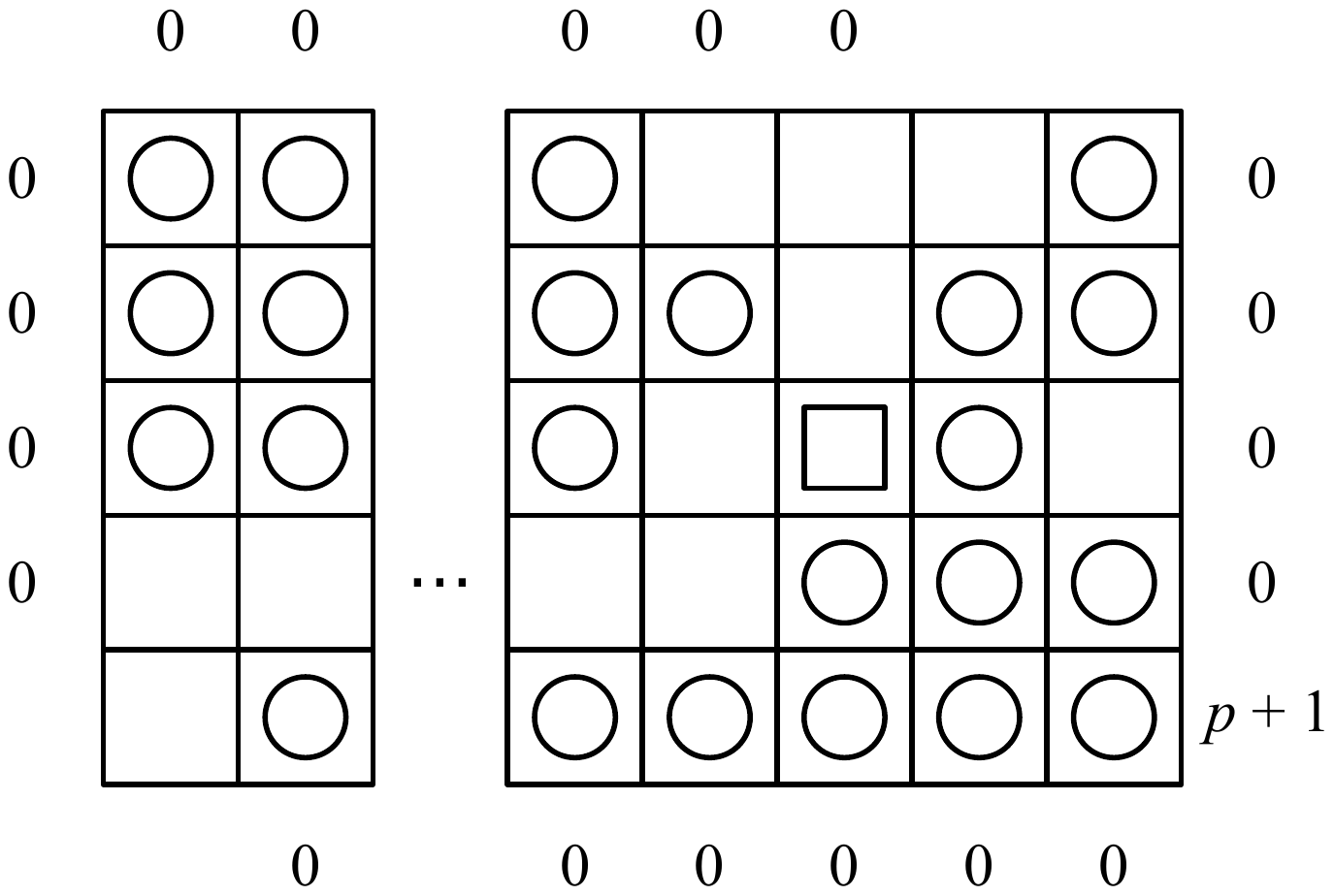}}
    ~~\raisebox{-.5\height}{\scalebox{2}{$\to$}}~~
    \raisebox{-.5\height}{\includegraphics[scale=.25]{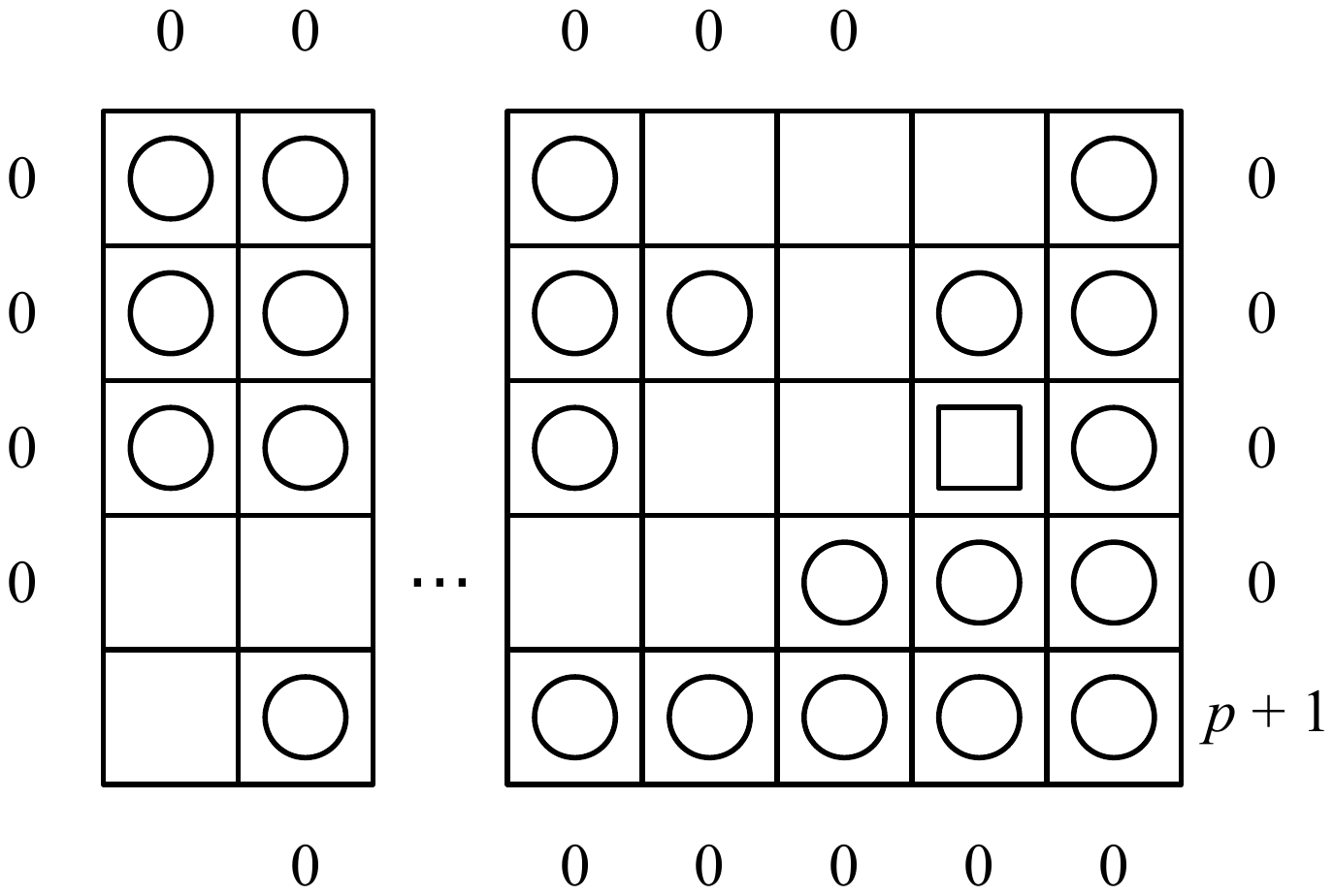}}
    ~~\raisebox{-.5\height}{\scalebox{2}{$\to$}}~~
    \raisebox{-.5\height}{\includegraphics[scale=.25]{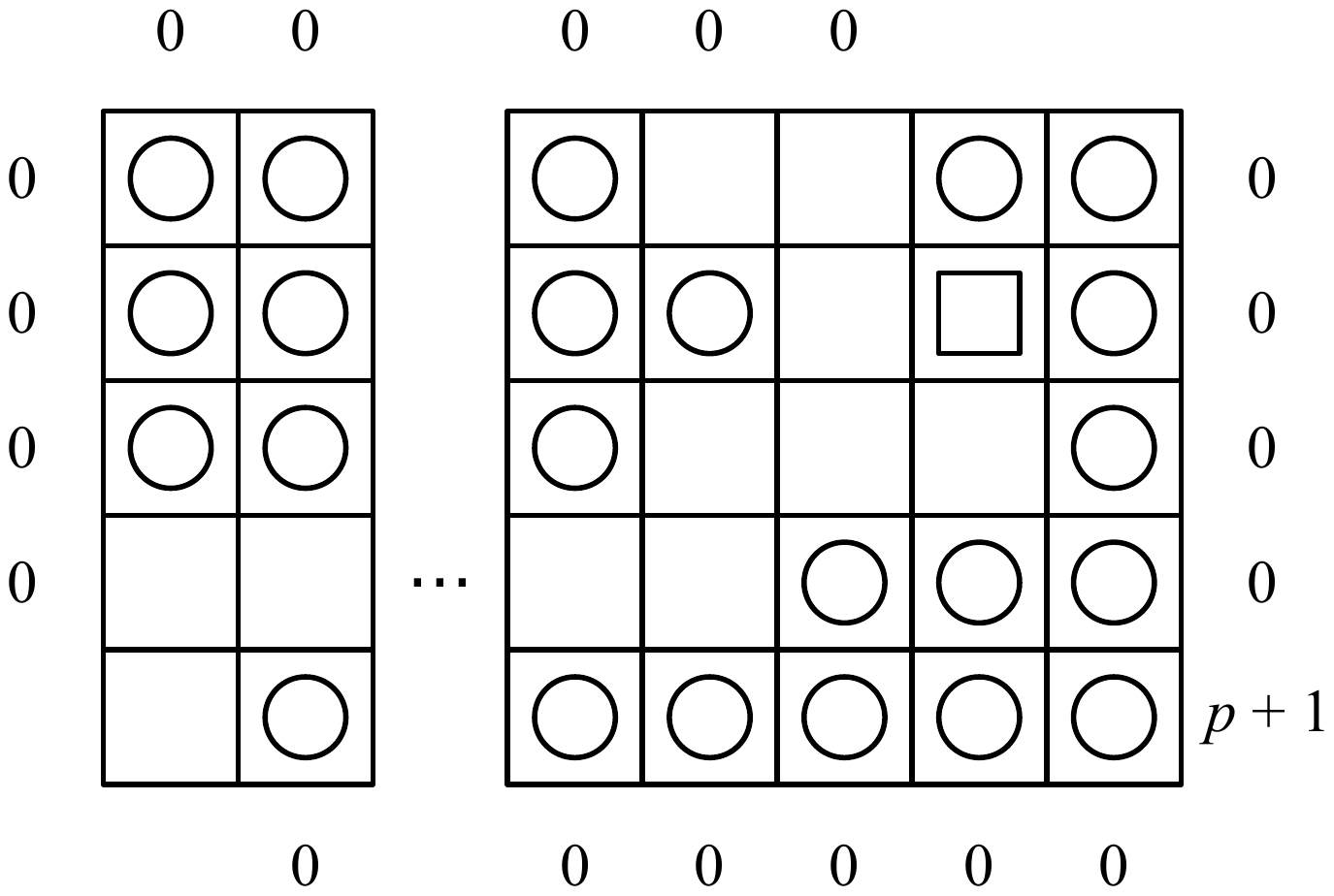}}
    ~~\raisebox{-.5\height}{\scalebox{2}{$\to$}}~~
    \raisebox{-.5\height}{\includegraphics[scale=.25]{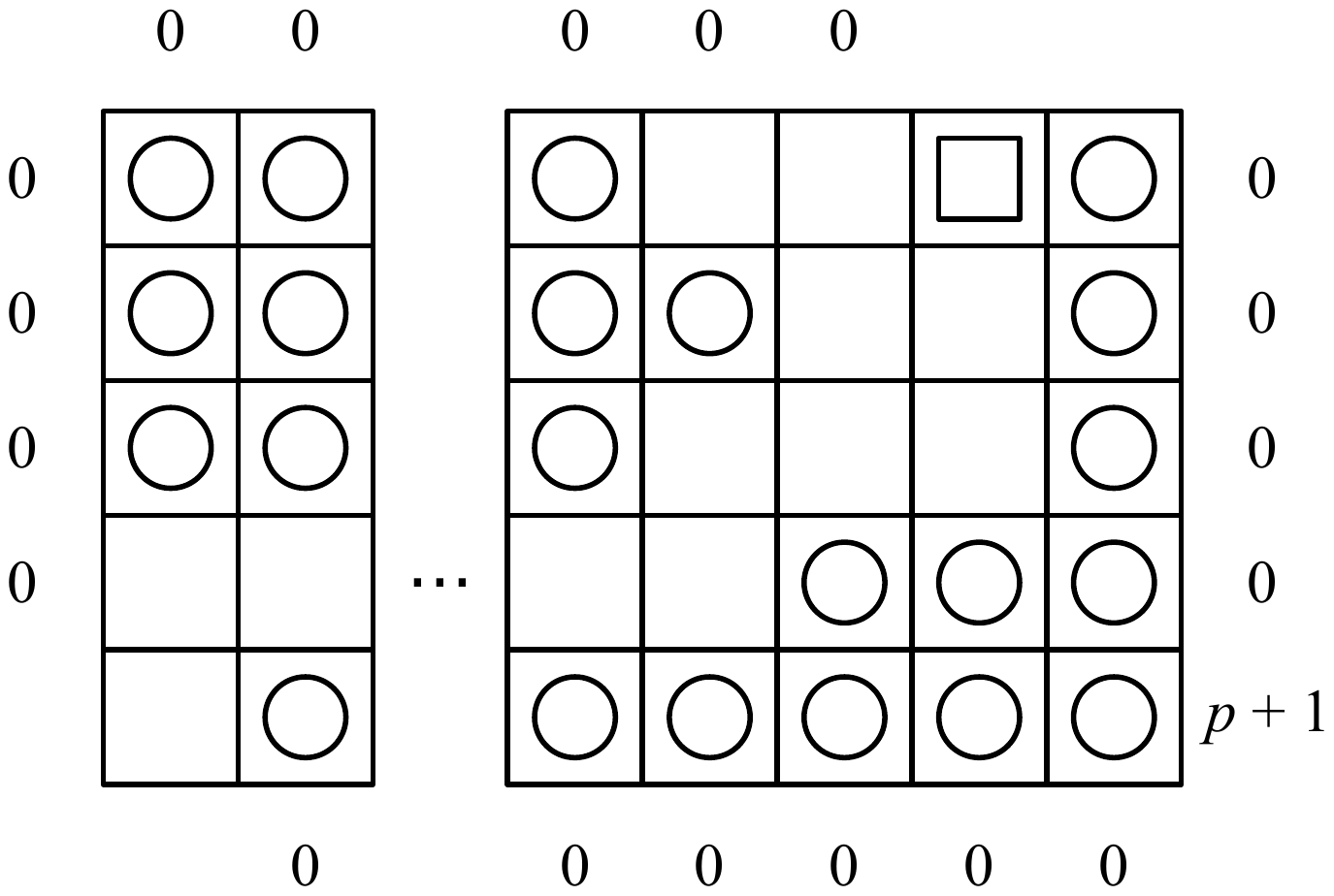}}
    ~~\raisebox{-.5\height}{\scalebox{2}{$\to$}}~~
    \raisebox{-.5\height}{\includegraphics[scale=.25]{images/variable_gadget_O2}}
    \caption{The result of pushing right from the last position in Figure~\ref{figure:variable_gadget_1b}, reaching the position in Figure~\ref{figure:variable_gadget_O2}.}
    \label{figure:variable_gadget_2bb}
\end{figure}

This completes the case analysis.
\end{proof}

\paragraph{Existential variable gadget:} The existential variable gadget, shown in Figure~\ref{figure:existential_variable_gadget}, is nearly the same as the core gadget, differing only in the bottom of the leftmost column. When instantiated in the reduction, the white king enters the gadget by pushing a white pawn up into the leftmost column, becoming exactly the core gadget. From the position immediately after the white king enters the gadget, the white king cannot push left (because there are no empty spaces in the row to the left) nor down (because it just pushed up, leaving an empty space in its former position), satisfying the assumption in Lemma~\ref{thm:core}. Thus by Lemma~\ref{thm:core-chaining}, the white king leaves the existential variable gadget in the second-rightmost column with a white pawn above it, and by either Lemma~\ref{thm:core-true} or \ref{thm:core-false}, all empty squares in one of two rows of the connection block are now filled by pawns pushed out of the existential variable gadget.

\iffull
\begin{figure}
    \centering
    \includegraphics[scale=.5]{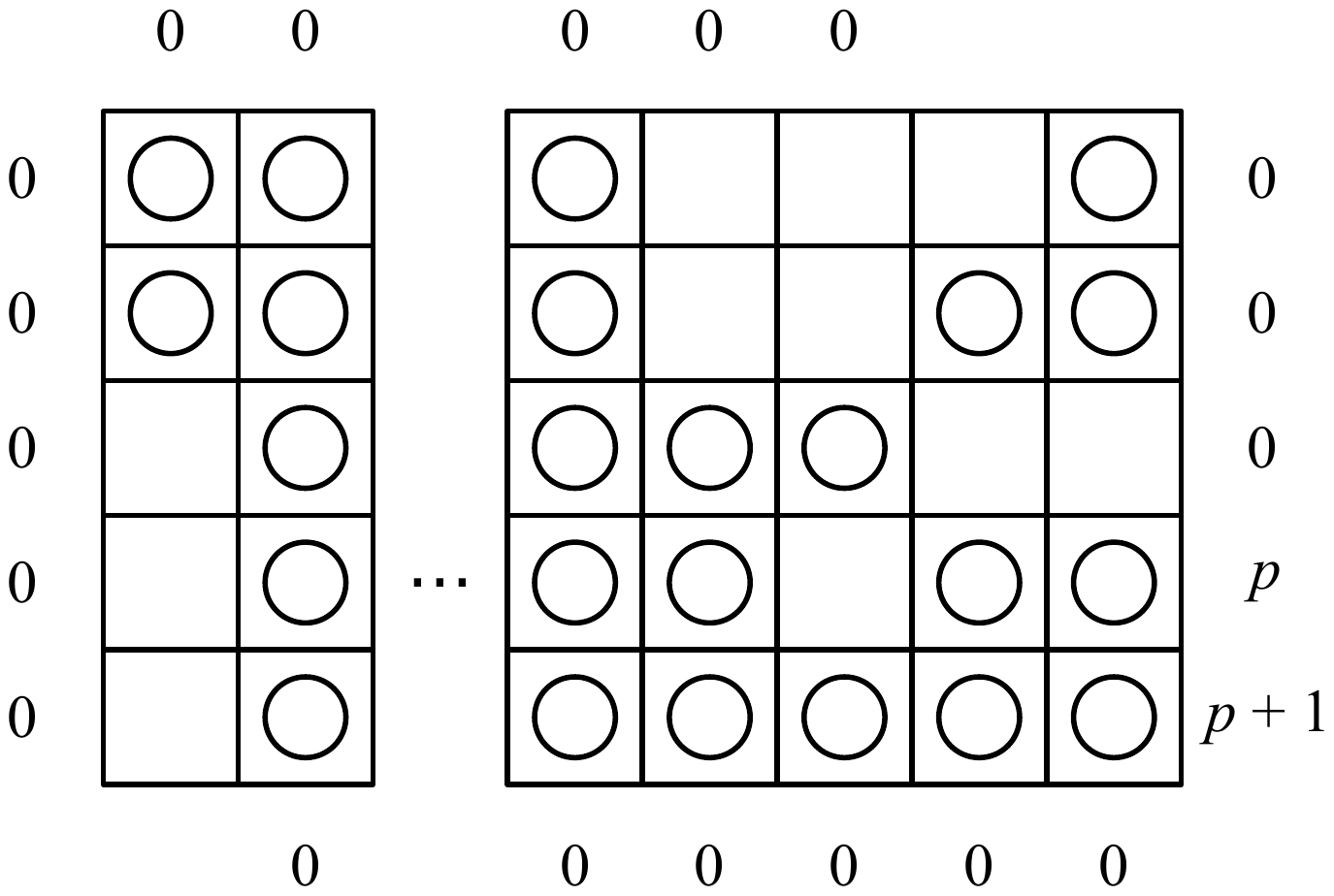}
    \caption{The existential variable gadget.}
    \label{figure:existential_variable_gadget}
\end{figure}
\fi

\paragraph{Universal variable gadget:} The universal variable gadget consists of two disconnected regions. The left subregion of the gadget occupies a $(p+6) \times 5$ rectangle in the \emph{variable gadget I} block. As the white king proceeds through the left region of the gadget, a subregion of the gadget reaches the initial state of the core gadget. The right region of the gadget occupies a $4 \times 4$ rectangle in the \emph{variable gadget II} block and contains a black pawn to allow Black to control the value of the variable. The bottom of the right region is one row lower than the bottom of the left region. The area between the two regions of the gadget (in the three rows shared by both) is entirely filled by white pawns. Figure~\ref{figure:universal_variable_gadget} shows the universal variable gadget, including the pawn-filled area between the regions.

\iffull
\begin{figure}
    \centering
    \includegraphics[scale=.5]{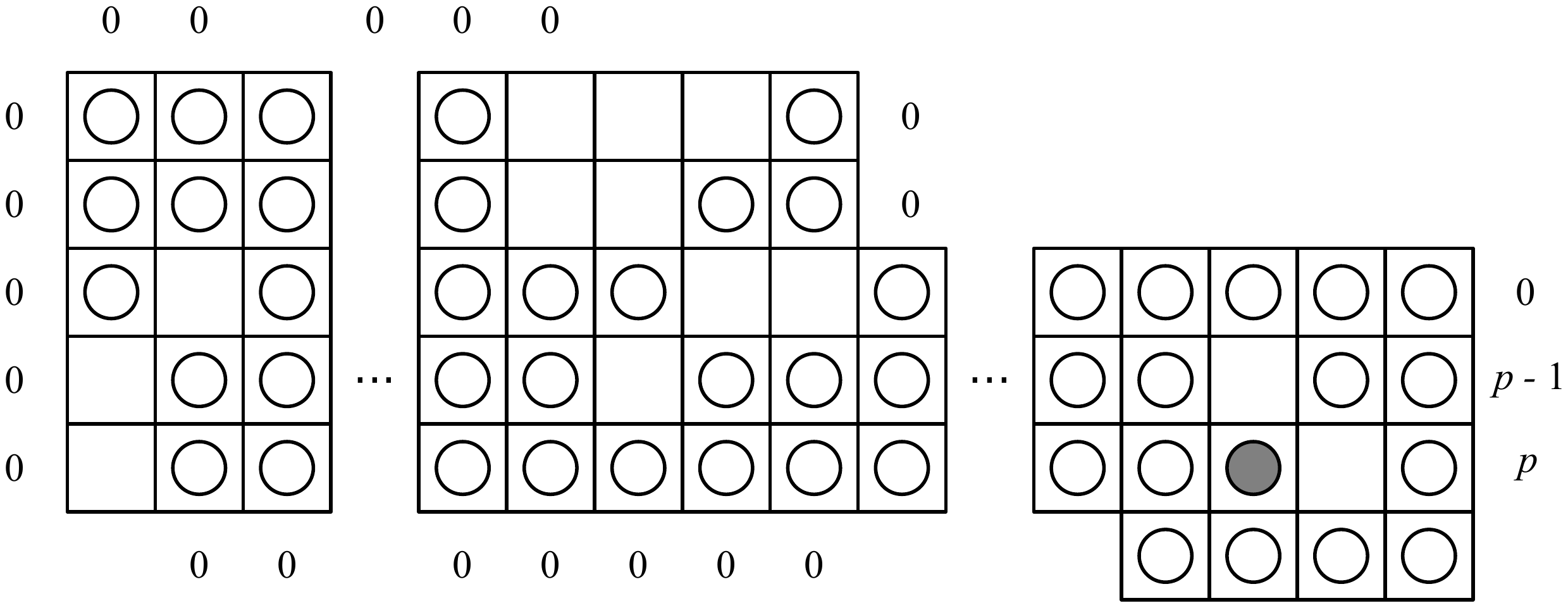}
    \caption{The universal variable gadget.}
    \label{figure:universal_variable_gadget}
\end{figure}
\fi

\begin{figure}
    \centering
    \includegraphics[scale=.5]{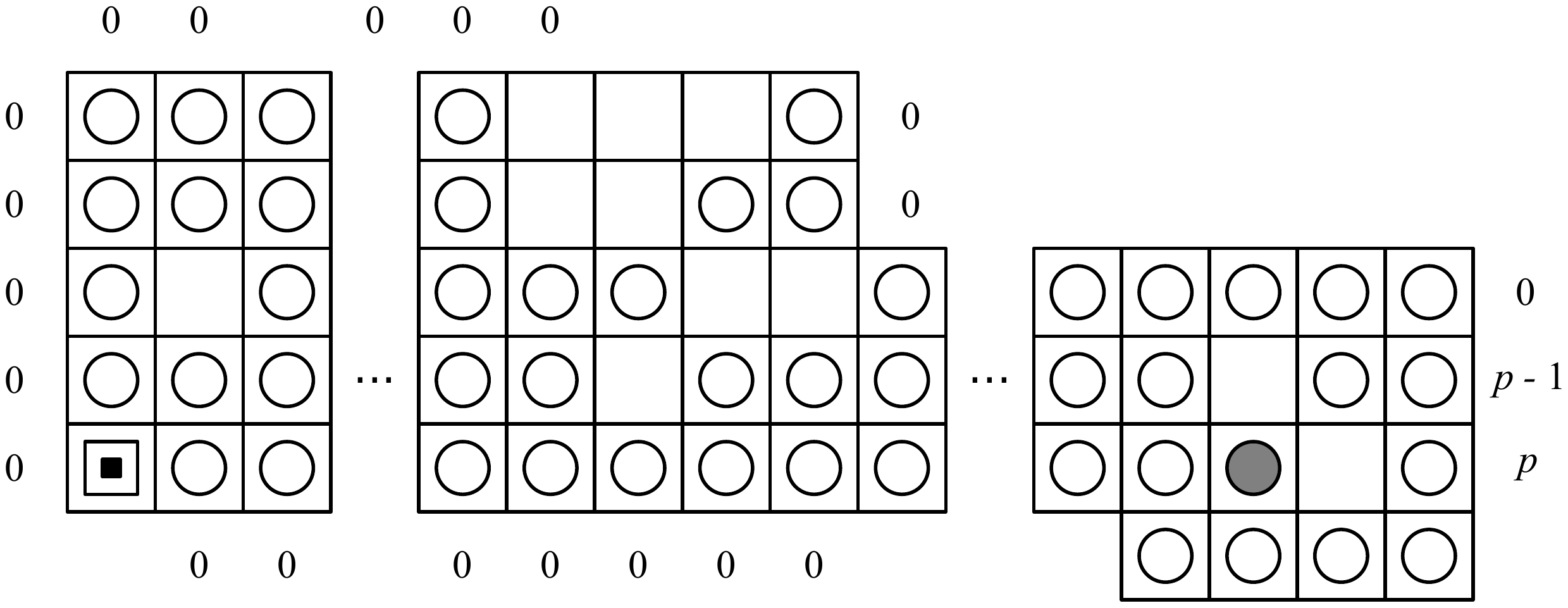}
    \caption{The universal variable gadget after the white king enters.}
    \label{figure:universal_variable_gadget_1}
\end{figure}

As with the existential variable gadget, when instantiated in the reduction, the white king enters the universal variable gadget by pushing a white pawn up into the leftmost column. Figure~\ref{figure:universal_variable_gadget_1} shows the resulting position. Regardless of Black's move, White's only legal push is to the right. By moving the black pawn, Black can choose between the two positions in Figure~\ref{figure:universal_variable_gadget_2}, depending on which of the two rows the black pawn is in when White pushes.

\begin{figure}
    \centering
    \subcaptionbox{\label{figure:universal_variable_gadget_2a}}
      {\includegraphics[scale=.25]{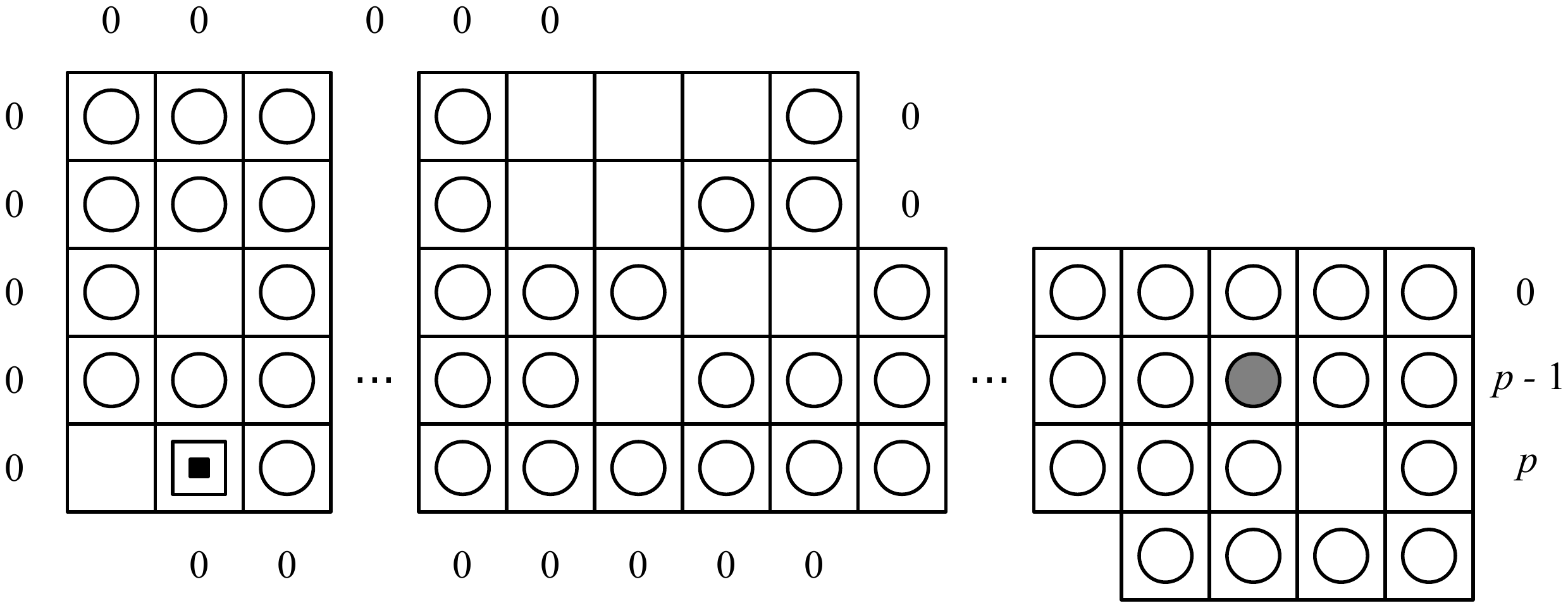}}
    ~~~~
    \subcaptionbox{\label{figure:universal_variable_gadget_2b}}
      {\includegraphics[scale=.25]{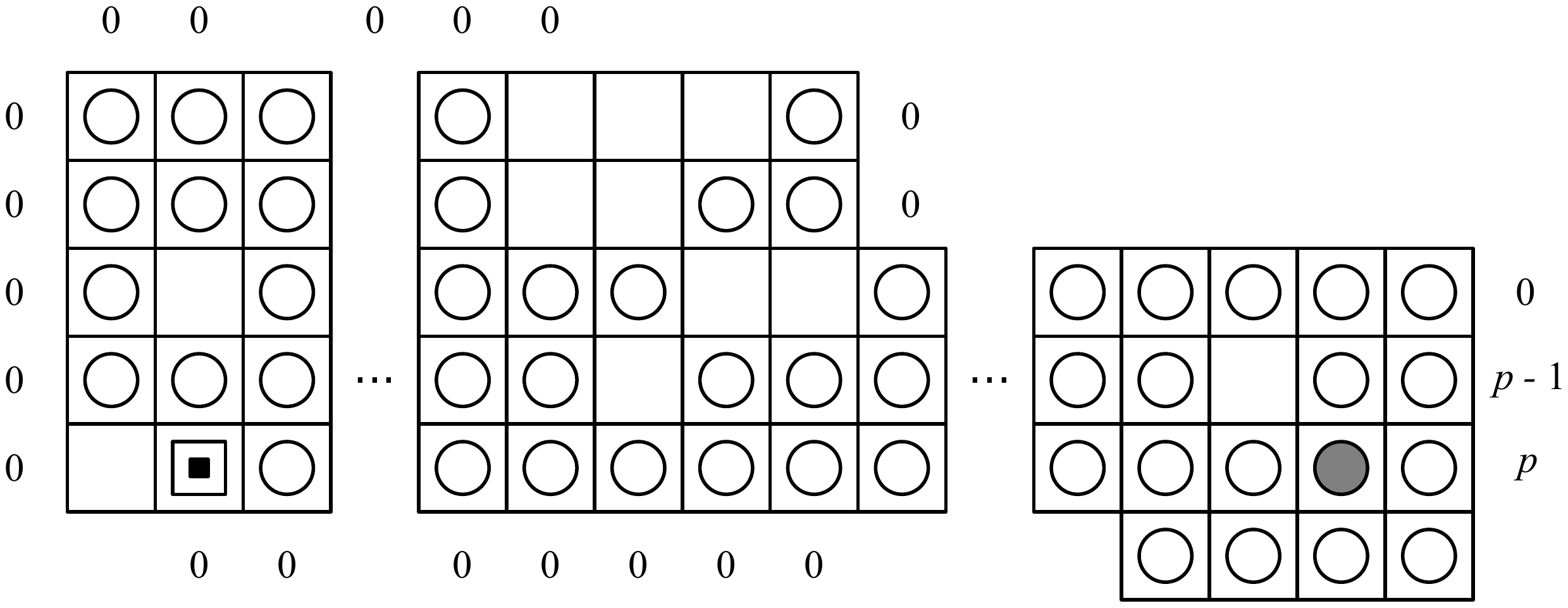}}
    \caption{The two possible configurations of the universal variable gadget one white turn after the configuration from Figure~\ref{figure:universal_variable_gadget_1}.}
    \label{figure:universal_variable_gadget_2}
\end{figure} 

In both of the resulting positions, the black pawn is surrounded, so Black can no longer influence events in this gadget. The left region of the gadget, without the leftmost column, is identical to the initial position of the core gadget. In both positions, the white king cannot push left (empty space) or down (no empty spaces down in the column), satisfying the assumption in Lemma~\ref{thm:core}. Thus either Lemma~\ref{thm:core-true} or Lemma~\ref{thm:core-false} holds. Because of the edge constraints, in Figure~\ref{figure:universal_variable_gadget_2a}, only Lemma~\ref{thm:core-true} is possible, resulting in Figure~\ref{figure:universal_variable_gadget_3a}. Similarly, in Figure~\ref{figure:universal_variable_gadget_2b}, only Lemma~\ref{thm:core-false} is possible, resulting in Figure~\ref{figure:universal_variable_gadget_3b}. By moving the black pawn to select one of these two cases, Black sets the value of the corresponding variable. Then by Lemma~\ref{thm:core-chaining}, the white king leaves in the second-rightmost column of the left region (in the \emph{variable gadget I} block) of the gadget. In both cases, the black pawn remains surrounded by white pawns in the right region of the gadget.

\begin{figure}
    \centering
    \subcaptionbox{\label{figure:universal_variable_gadget_3a}}
      {\includegraphics[scale=.25]{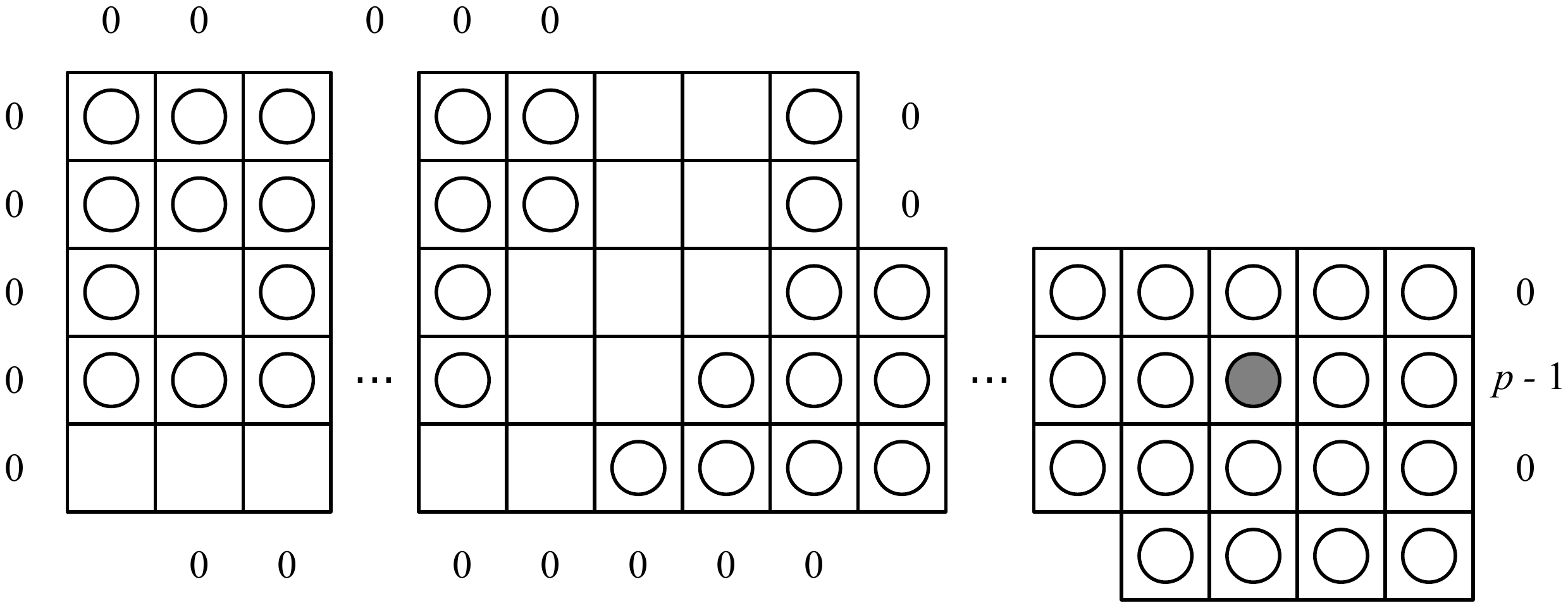}}
    ~~~~
    \subcaptionbox{\label{figure:universal_variable_gadget_3b}}
      {\includegraphics[scale=.25]{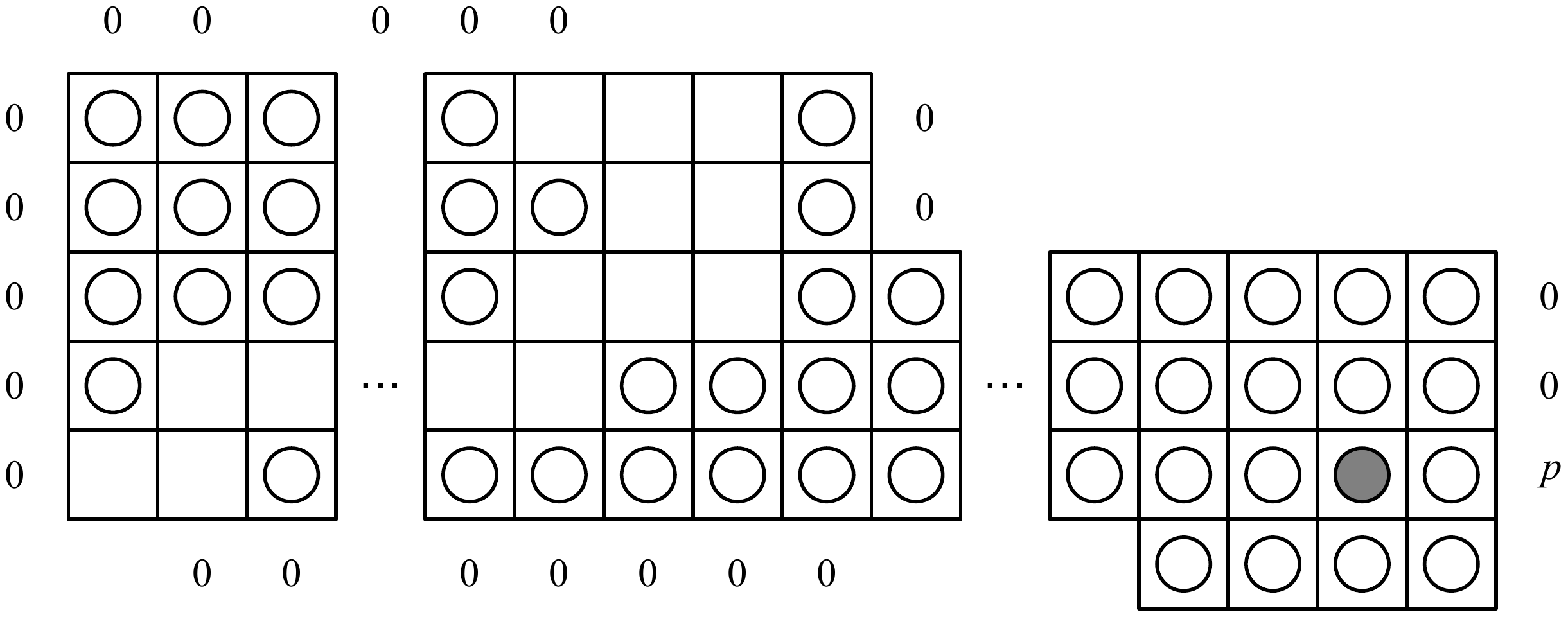}}
    \caption{The two possible final positions of the universal variable gadget after the white king exits.}
    \label{figure:universal_variable_gadget_3}
\end{figure} 
}

\subsection{Bridge gadget}
\abstractlater{\subsection{Bridge gadget}}

The bridge gadget, shown in Figure~\ref{figure:bridge_gadget}, brings the white king from the exit of the last variable gadget to the entrance of the first clause gadget. When instantiated in the reduction, the white king enters the bridge gadget from the bottom of the leftmost column, preceded by a white pawn. The white king's traversal of the bridge gadget is entirely forced, as shown in Figure~\ref{figure:bridge_gadget_1}. The white king leaves the gadget by pushing a white pawn out to the right in the second-to-top row.

\begin{figure}[t]
  \centering
  \begin{minipage}[b]{0.3\textwidth}
    \centering
    \includegraphics[scale=.5]{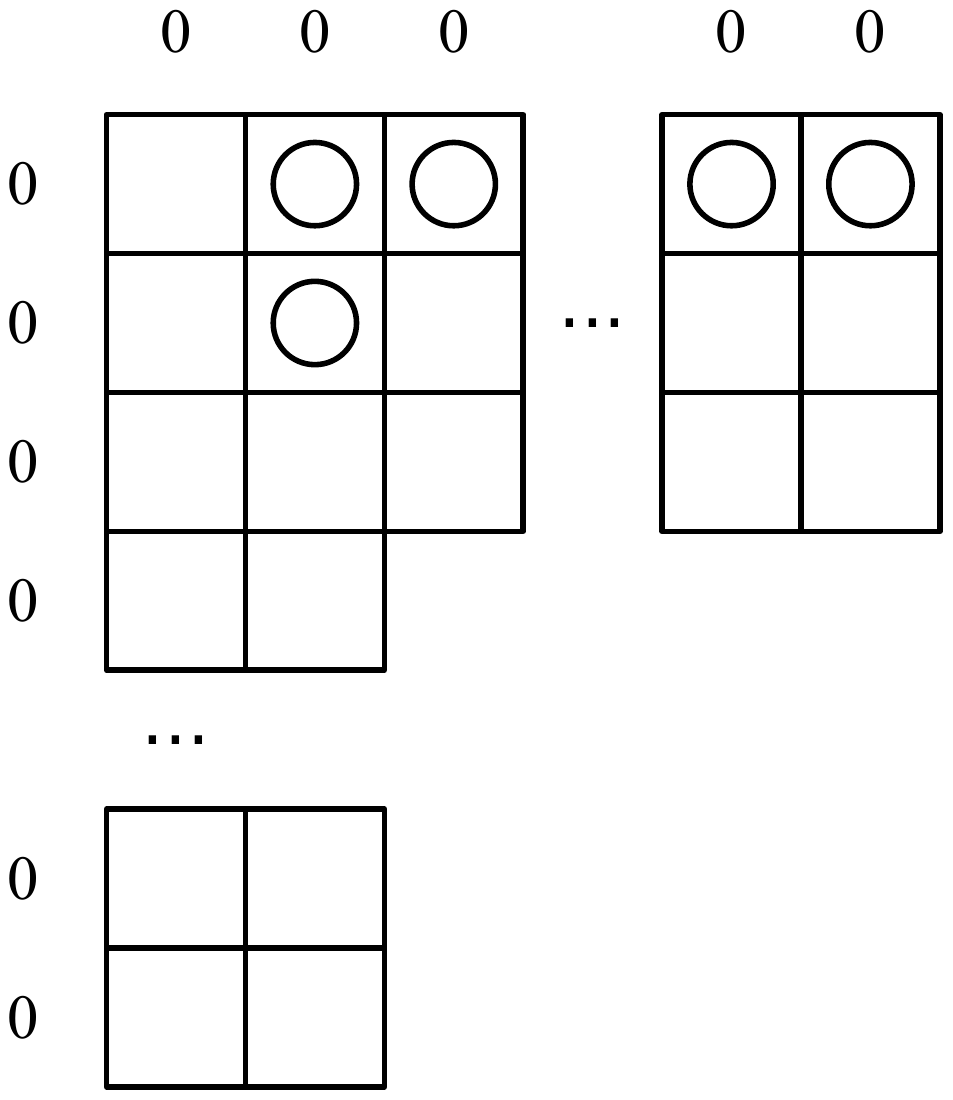}
    \caption{The bridge gadget.}
    \label{figure:bridge_gadget}
  \end{minipage}\hfill
  \begin{minipage}[b]{0.3\textwidth}
    \centering
    \includegraphics[scale=.5]{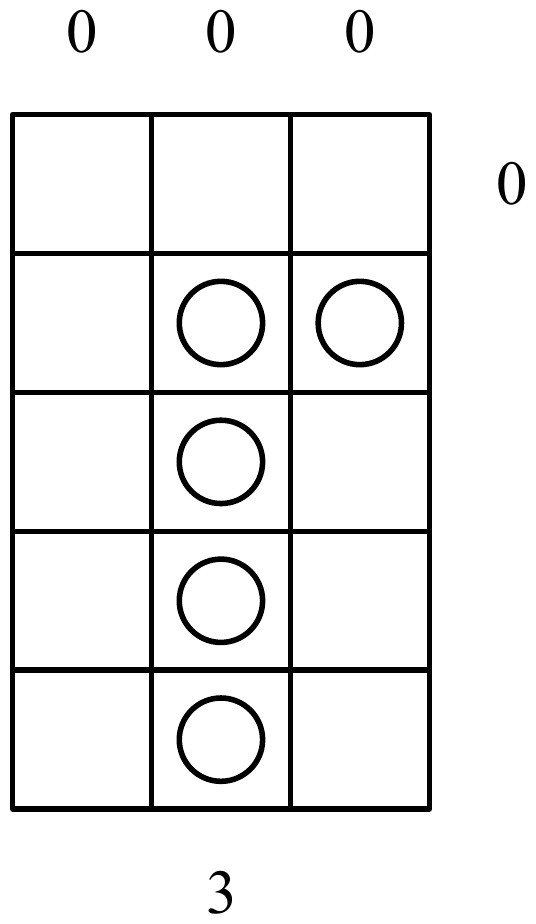}
    \caption{The clause gadget.}
    \label{figure:clause_gadget}
  \end{minipage}\hfill
  \begin{minipage}[b]{0.3\textwidth}
    \centering
    \includegraphics[scale=.5]{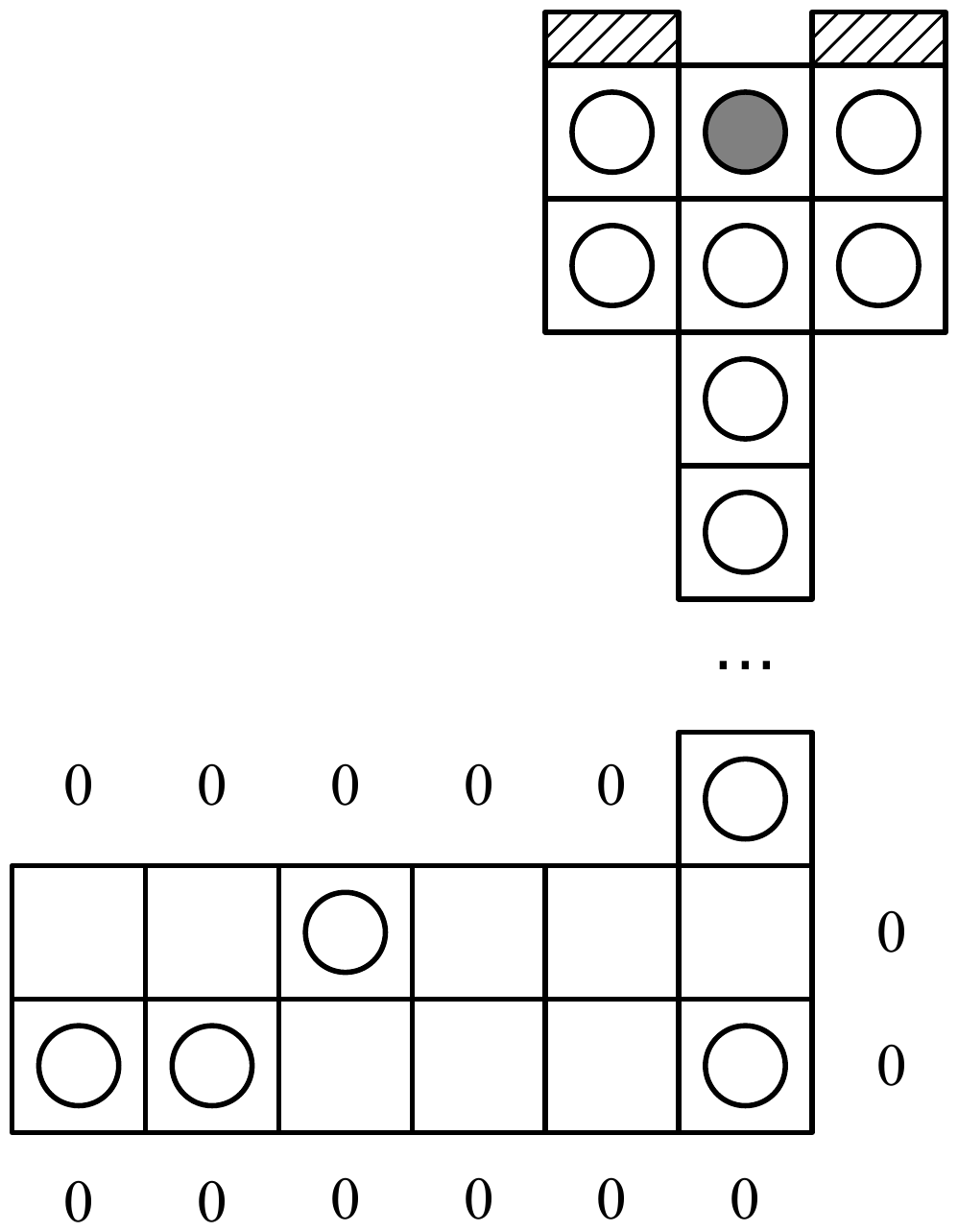}
    \caption{The reward gadget.}
    \label{figure:reward_gadget}
  \end{minipage}
\end{figure} 


\later{
\begin{figure}
    \centering
    ~~~~~~~~
    \raisebox{-.5\height}{\includegraphics[scale=.25]{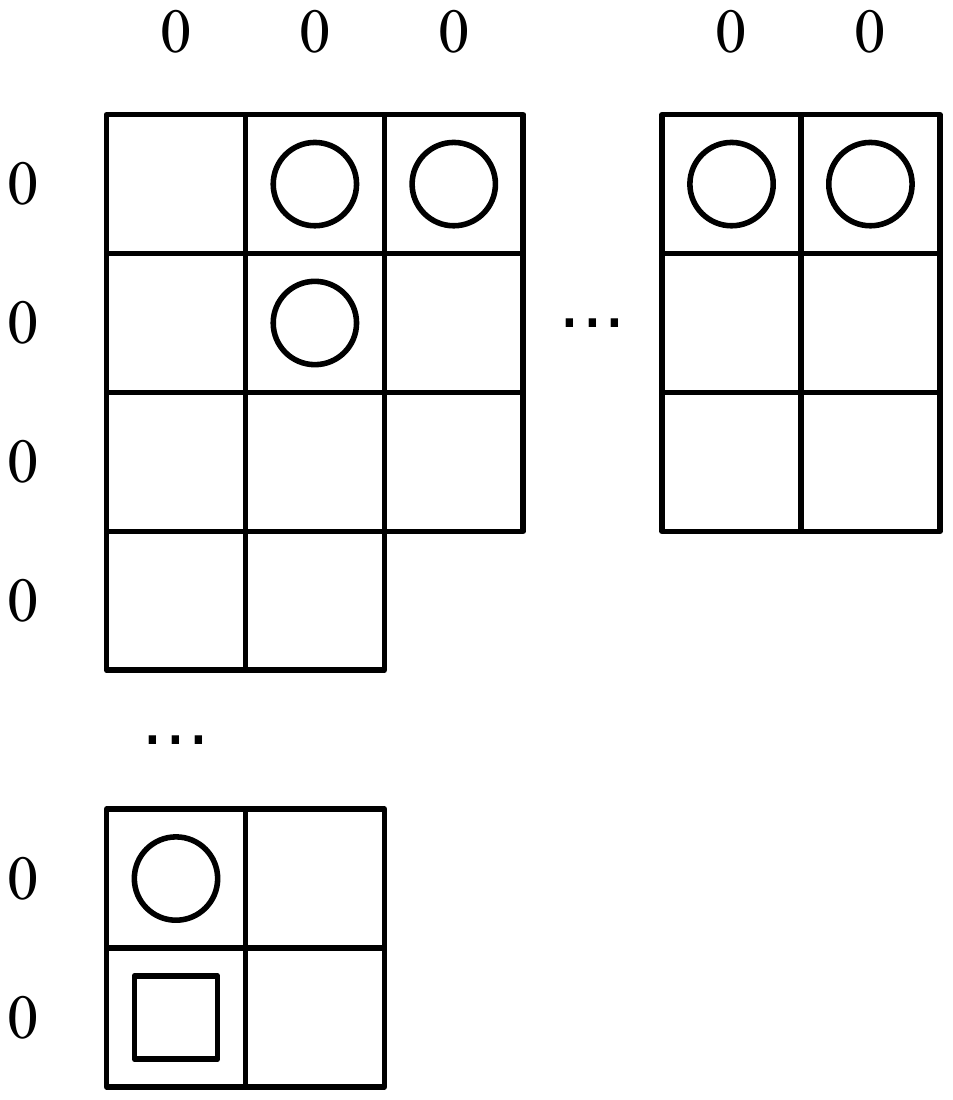}}
    ~~\raisebox{-.5\height}{\scalebox{2}{$\to$}}~~
    \raisebox{-.5\height}{\includegraphics[scale=.25]{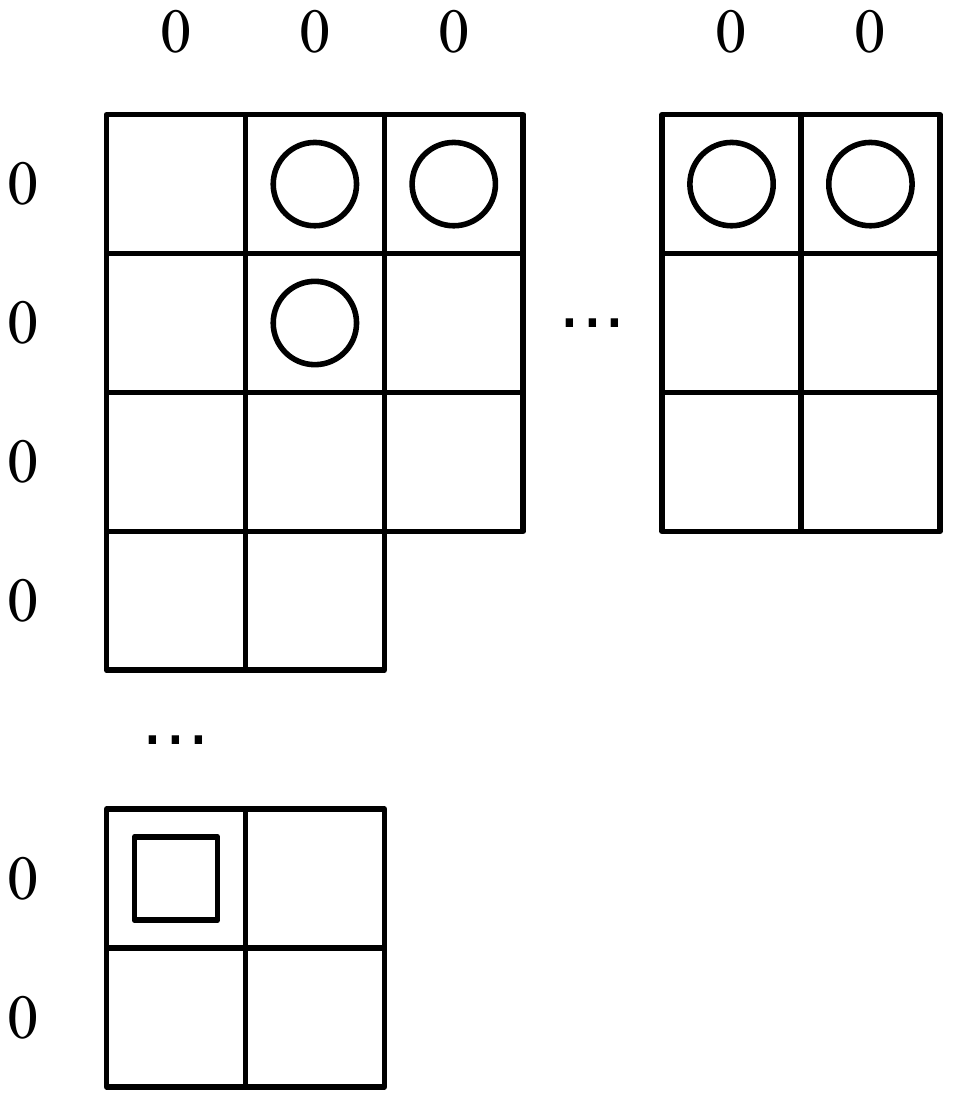}}
    ~~\raisebox{-.5\height}{\scalebox{2}{$\to$}}
    \raisebox{-.5\height}{$\cdots$}
    \raisebox{-.5\height}{\scalebox{2}{$\to$}}~~
    \raisebox{-.5\height}{\includegraphics[scale=.25]{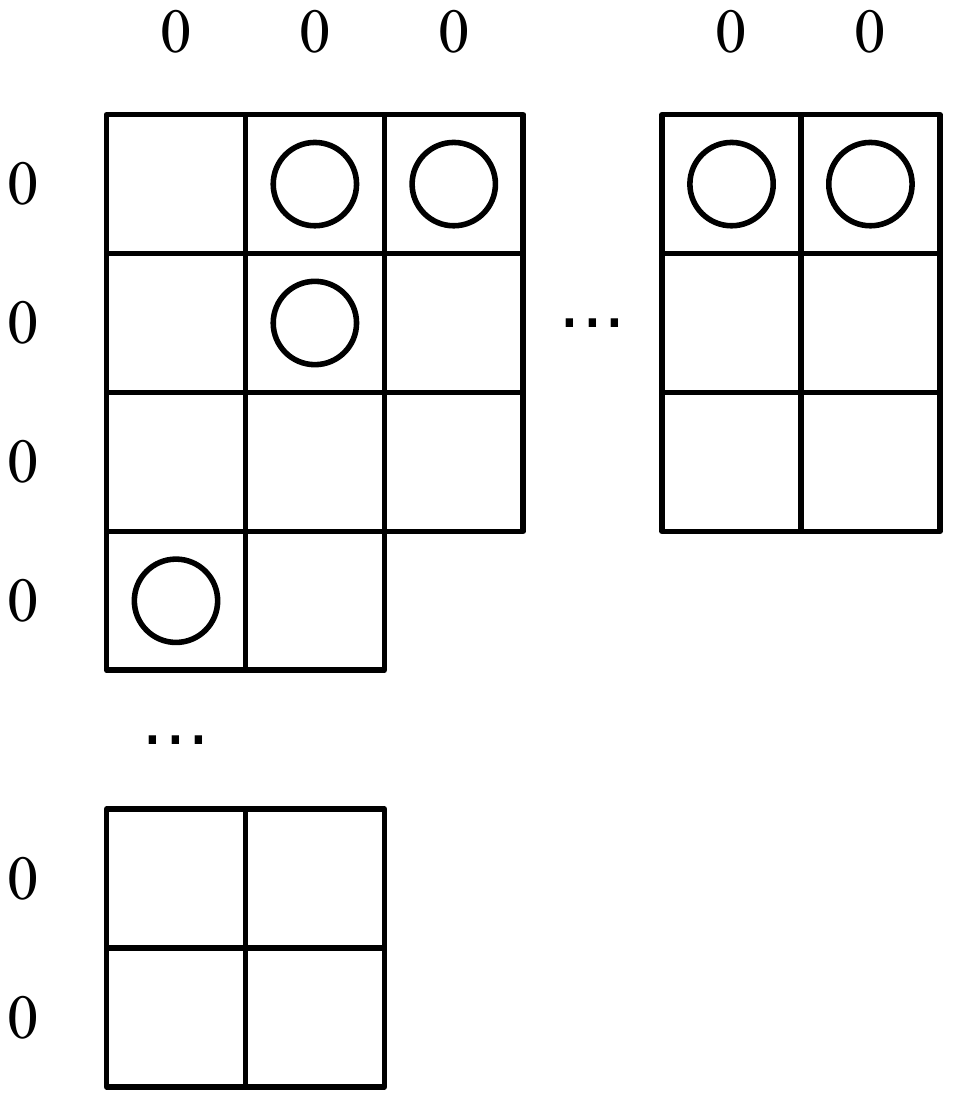}}
    ~~\raisebox{-.5\height}{\scalebox{2}{$\to$}}~~
    \raisebox{-.5\height}{\includegraphics[scale=.25]{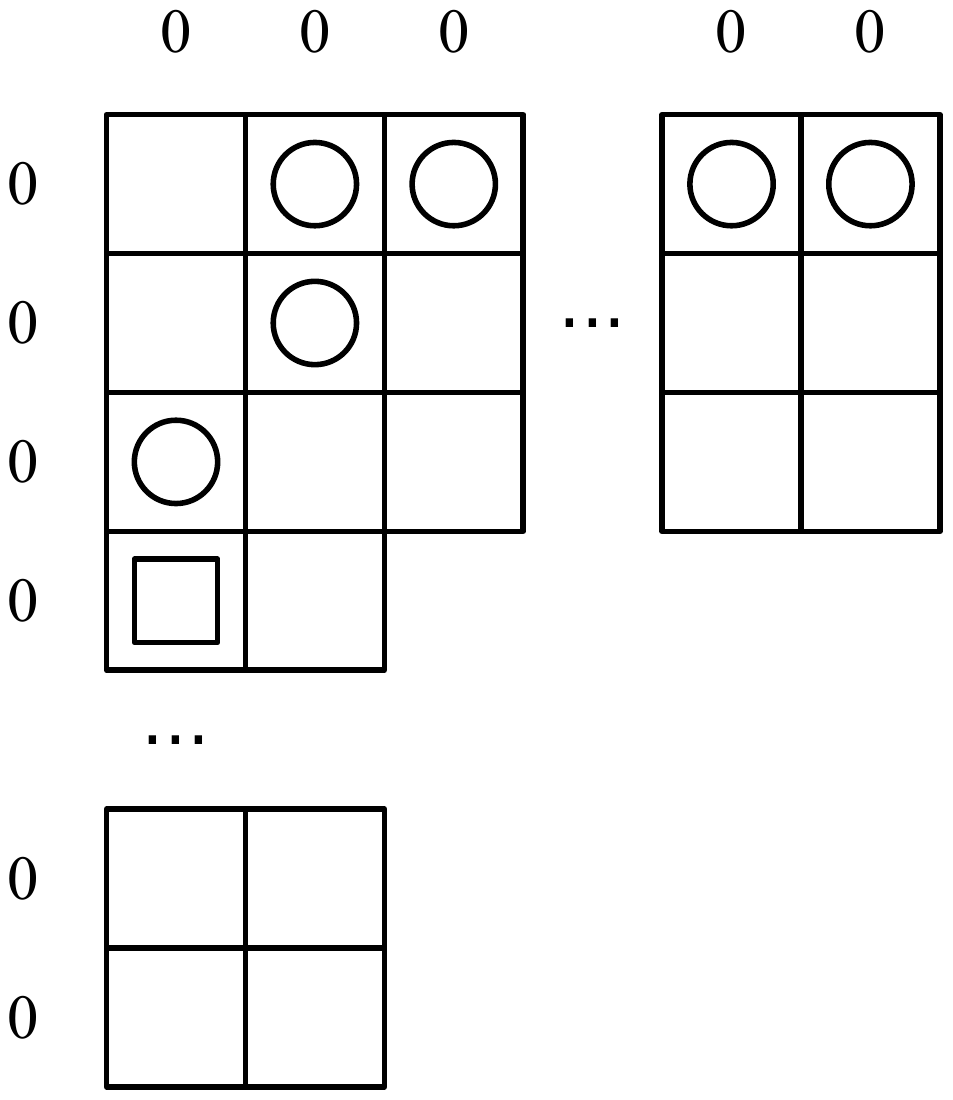}}
    ~~\raisebox{-.5\height}{\scalebox{2}{$\to$}}~~
    \raisebox{-.5\height}{\includegraphics[scale=.25]{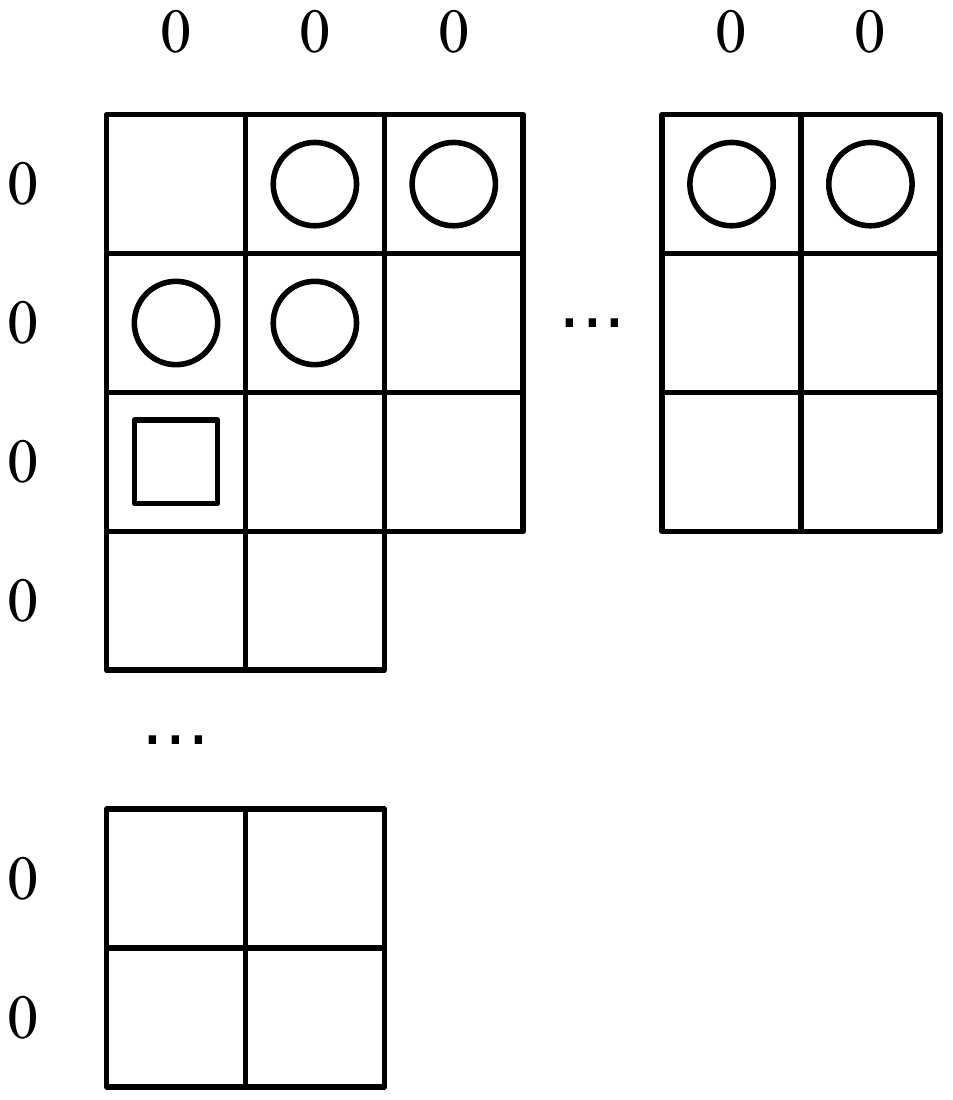}}
    ~~\raisebox{-.5\height}{\scalebox{2}{$\to$}}~~
    \raisebox{-.5\height}{\includegraphics[scale=.25]{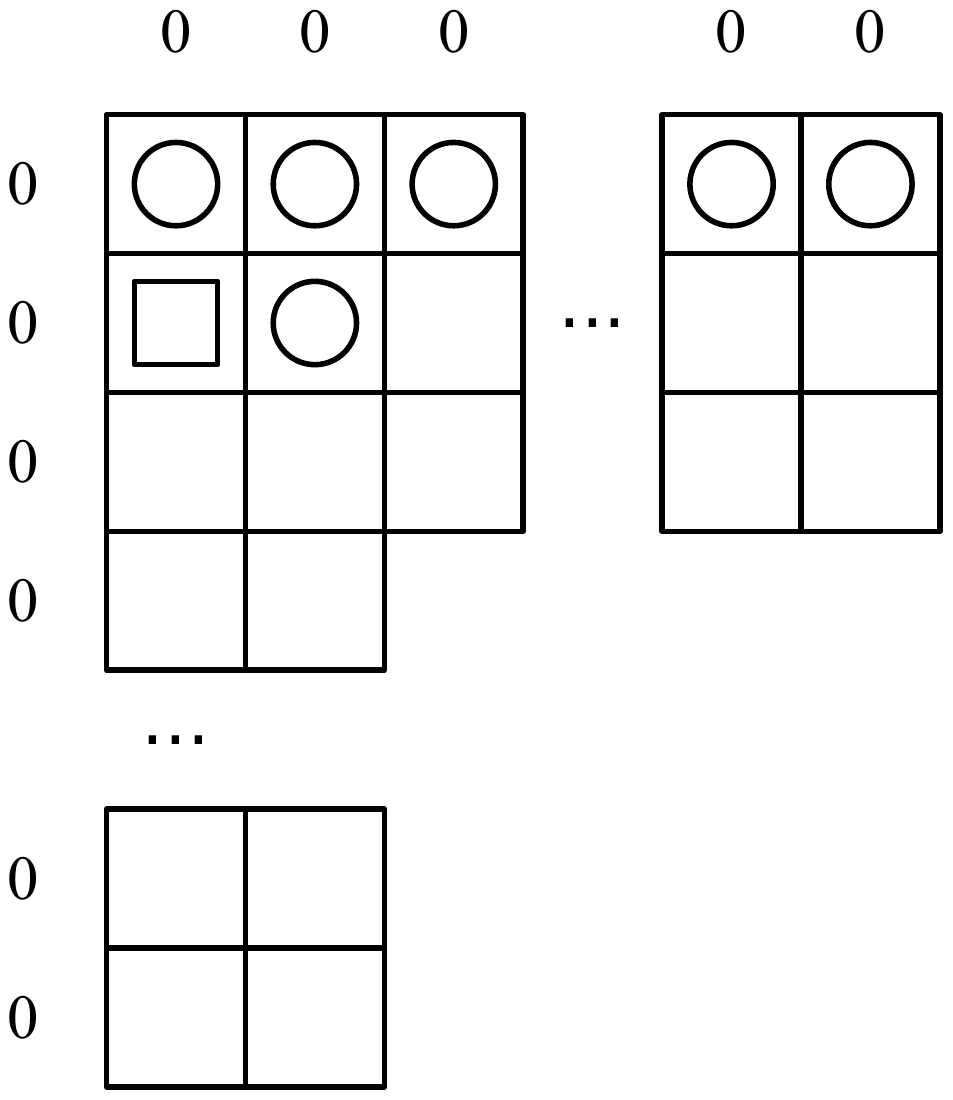}}
    ~~\raisebox{-.5\height}{\scalebox{2}{$\to$}}~~
    \raisebox{-.5\height}{\includegraphics[scale=.25]{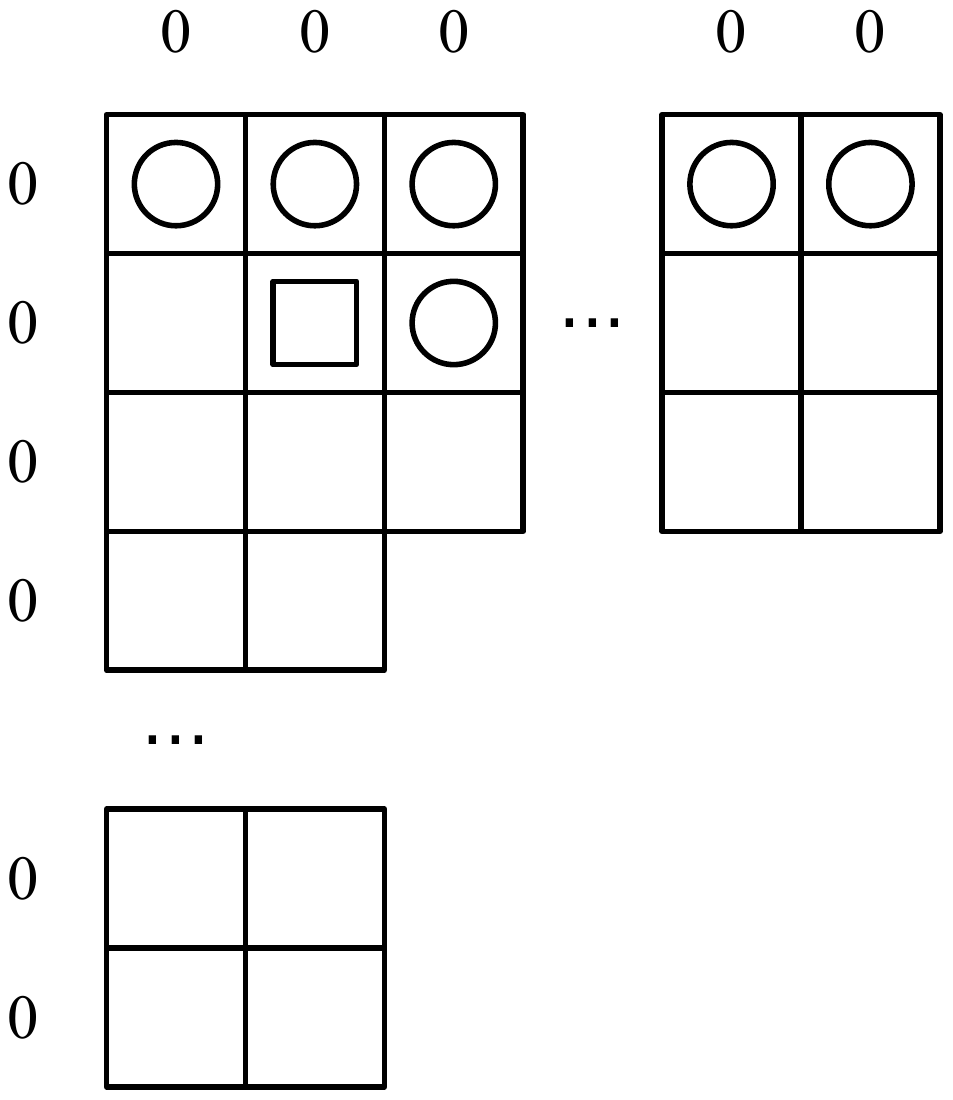}}
    ~~\raisebox{-.5\height}{\scalebox{2}{$\to$}}~~
    \raisebox{-.5\height}{\includegraphics[scale=.25]{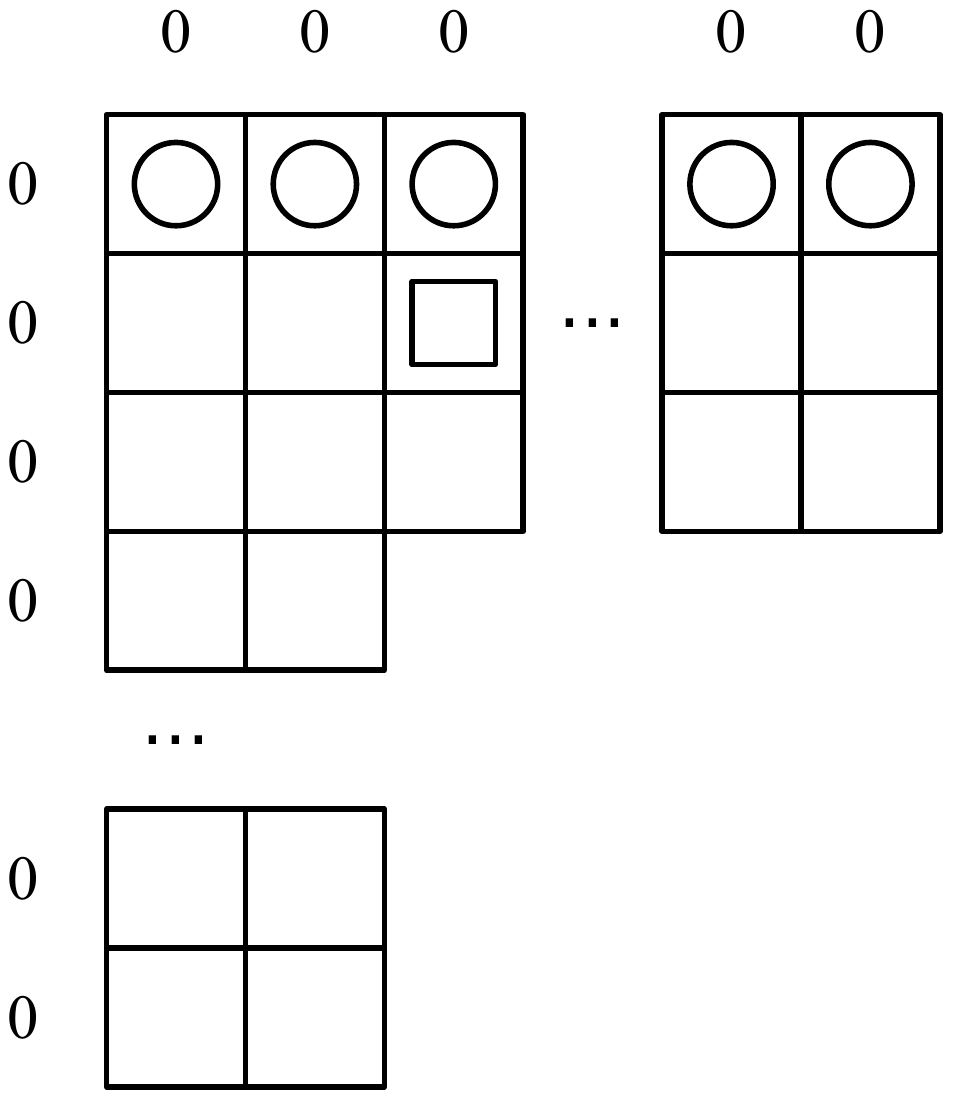}}
    ~~\raisebox{-.5\height}{\scalebox{2}{$\to$}}
    \raisebox{-.5\height}{$\cdots$}
    \raisebox{-.5\height}{\scalebox{2}{$\to$}}~~
    \raisebox{-.5\height}{\includegraphics[scale=.25]{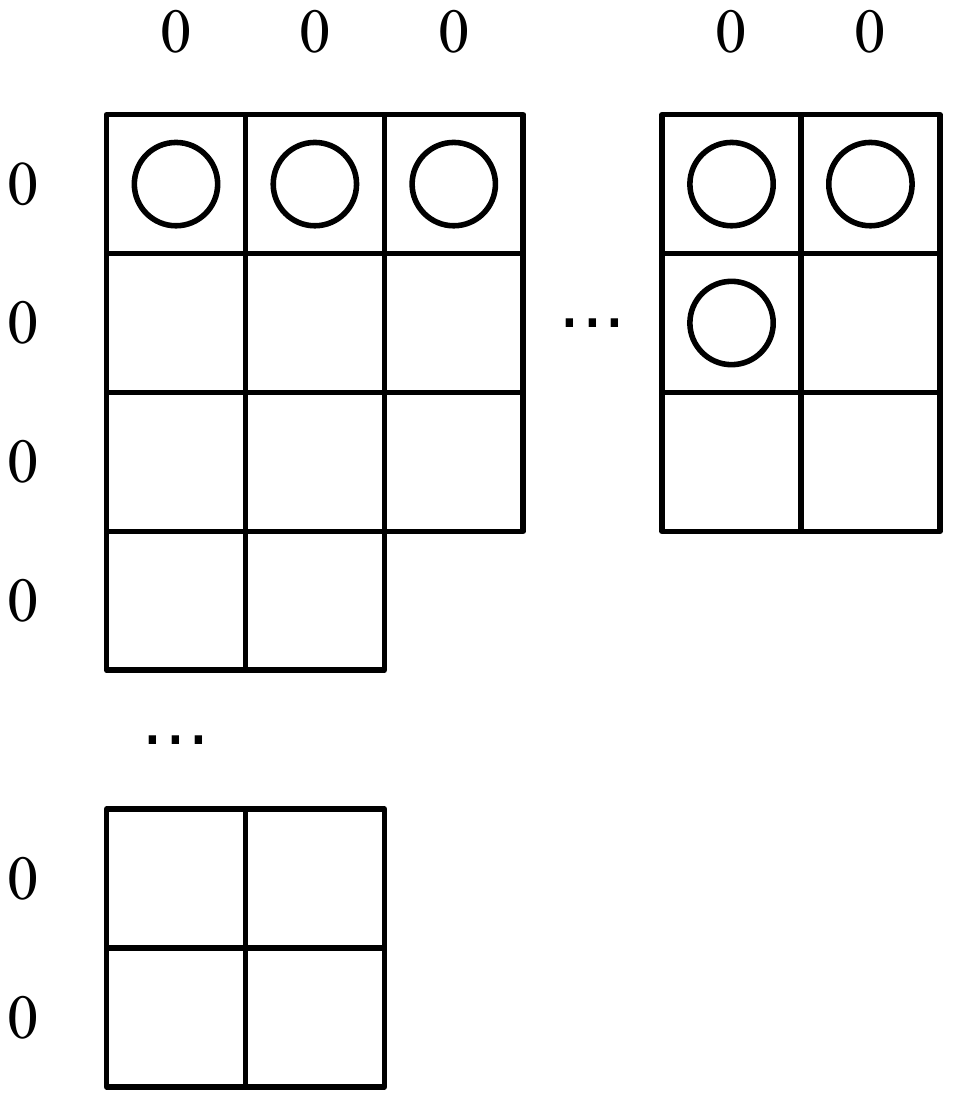}}
    ~~\raisebox{-.5\height}{\scalebox{2}{$\to$}}~~
    \raisebox{-.5\height}{\includegraphics[scale=.25]{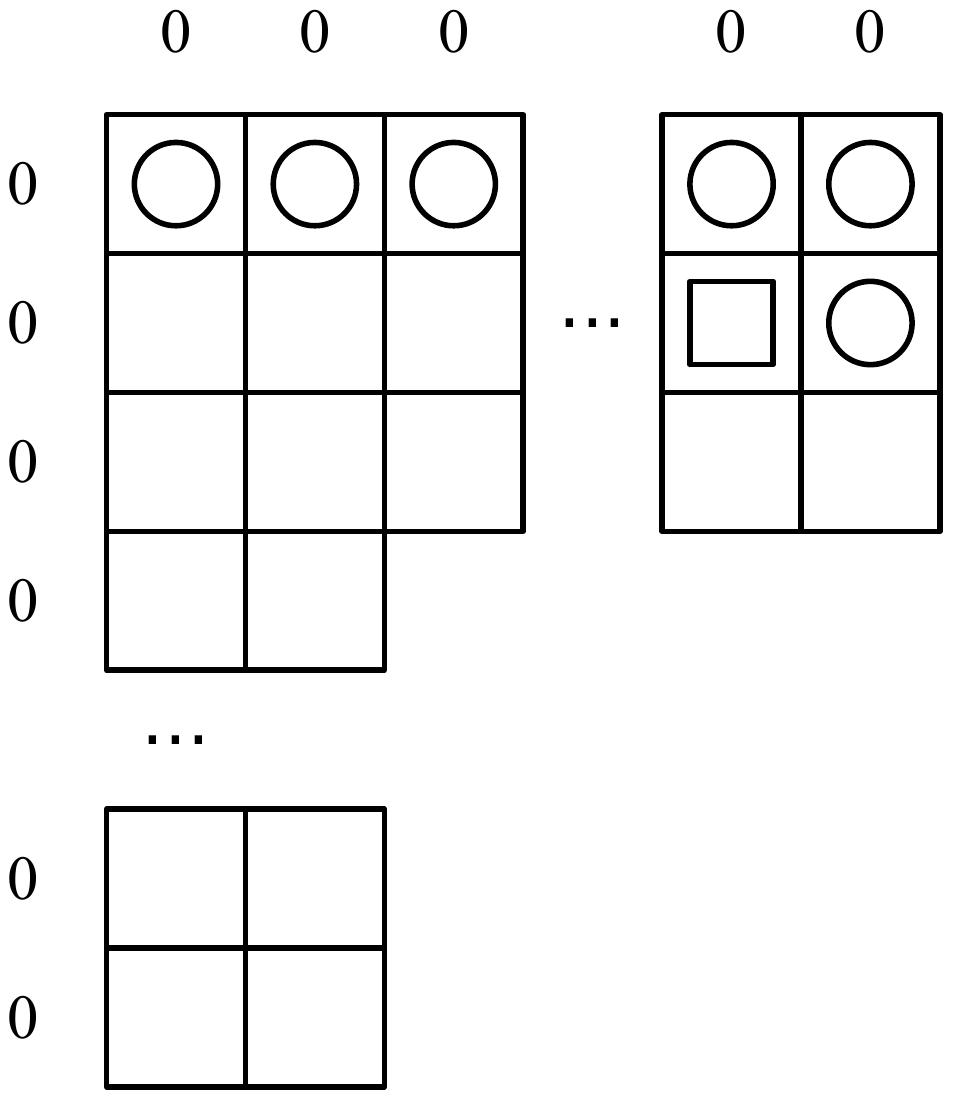}}
    ~~\raisebox{-.5\height}{\scalebox{2}{$\to$}}~~
    \raisebox{-.5\height}{\includegraphics[scale=.25]{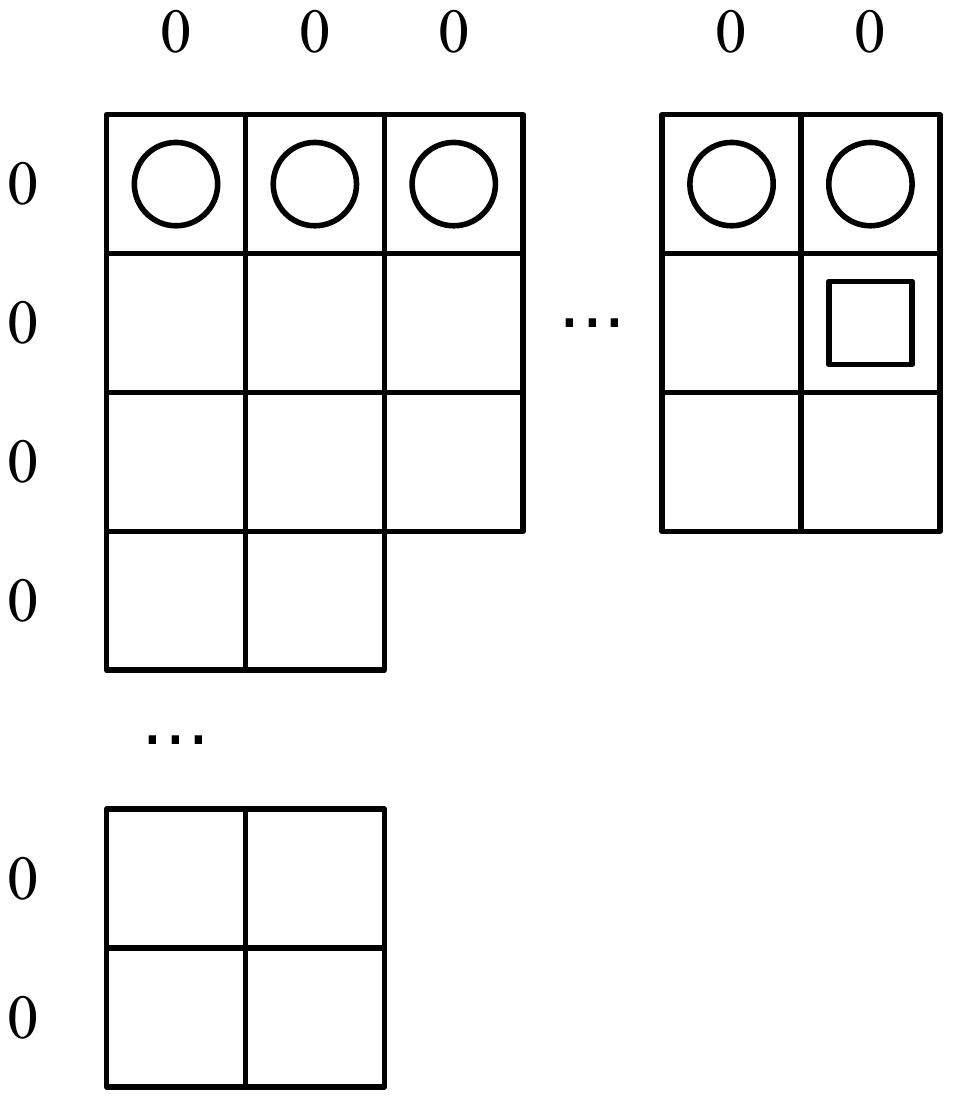}}
    ~~\raisebox{-.5\height}{\scalebox{2}{$\to$}}~~
    \raisebox{-.5\height}{\includegraphics[scale=.25]{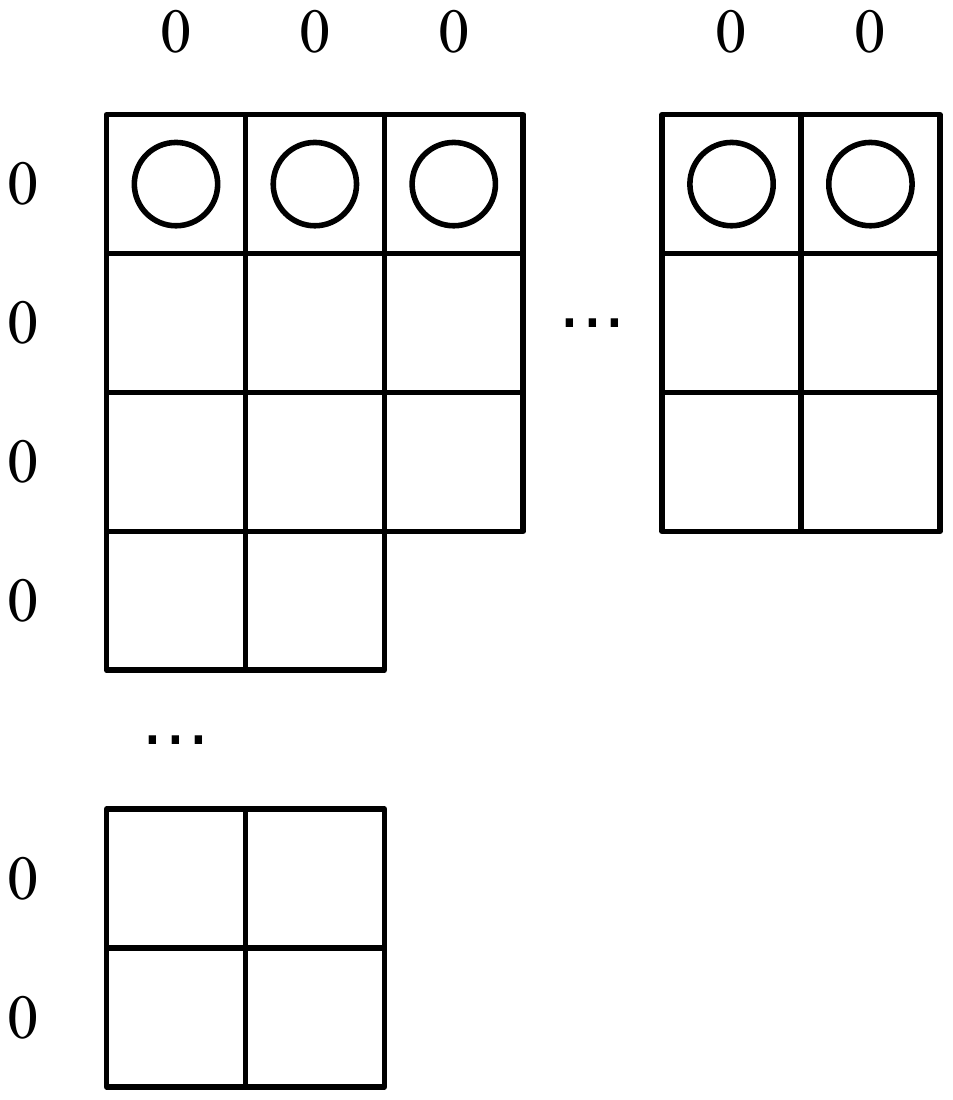}}
    \caption{The sequence of configurations of the bridge gadget as it is traversed by the white king.}
    \label{figure:bridge_gadget_1}
\end{figure}
}

\subsection{Clause gadget}
\abstractlater{\subsection{Clause gadget}}

The clause gadget, shown in Figure~\ref{figure:clause_gadget}, verifies that a column below the gadget contains at least one empty square. When instantiated in the reduction, the white king enters the gadget from the left in the top row, preceded by a white pawn. Figure~\ref{figure:clause_gadget_1} shows the resulting sequence of forced pushes. This push sequence includes a push down in the central column of the gadget; if there are no empty squares below the gadget in that column, the white king has no legal pushes and White loses. If there are more empty squares, White can continue to push down, but (when instantiated in the reduction) there are at most three total empty squares in that column, and once those squares are filled, White cannot push. Thus the white king must push right instead and leave the gadget by pushing a white pawn out to the right in the second-to-top row, as shown in Figure~\ref{figure:clause_gadget_2}.


\later{
\begin{figure}
    \centering
    \raisebox{-.5\height}{\includegraphics[scale=.25]{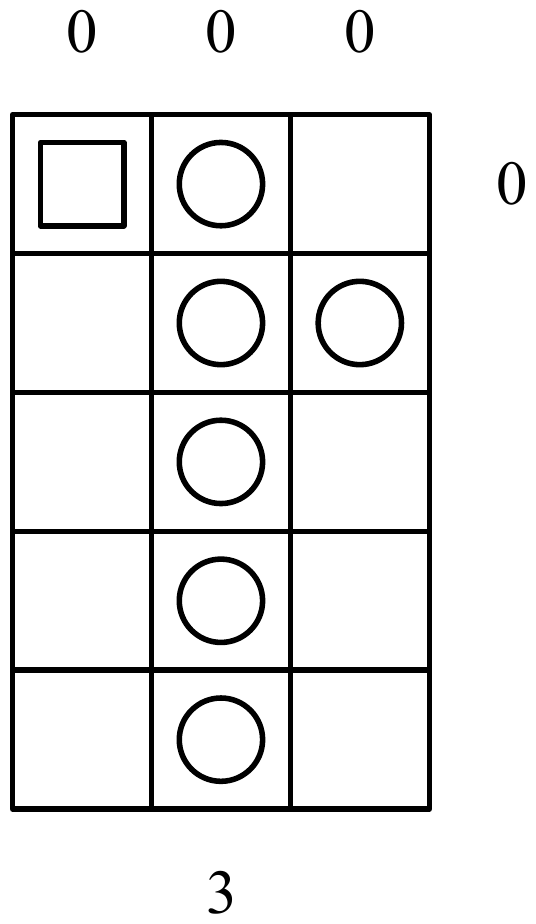}}
    ~~\raisebox{-.5\height}{\scalebox{2}{$\to$}}~~
    \raisebox{-.5\height}{\includegraphics[scale=.25]{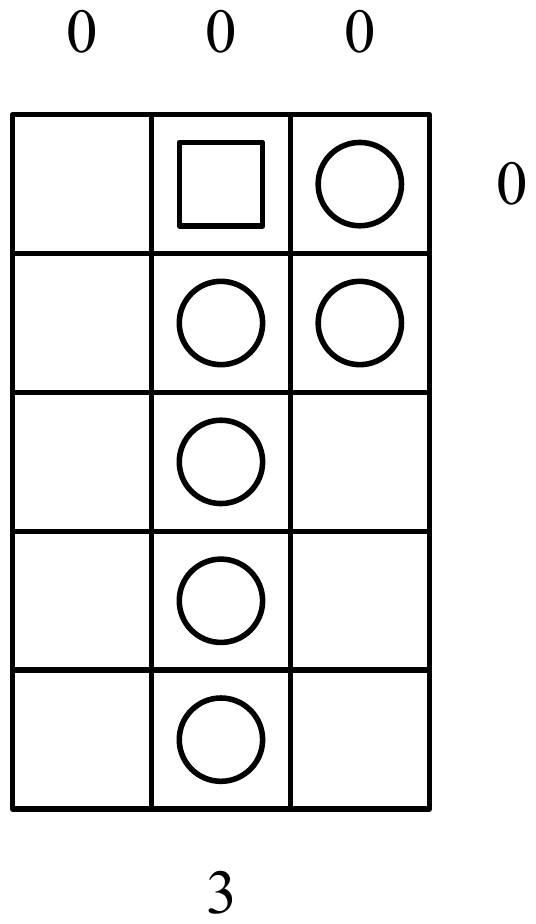}}
    ~~\raisebox{-.5\height}{\scalebox{2}{$\to$}}~~
    \raisebox{-.5\height}{\includegraphics[scale=.25]{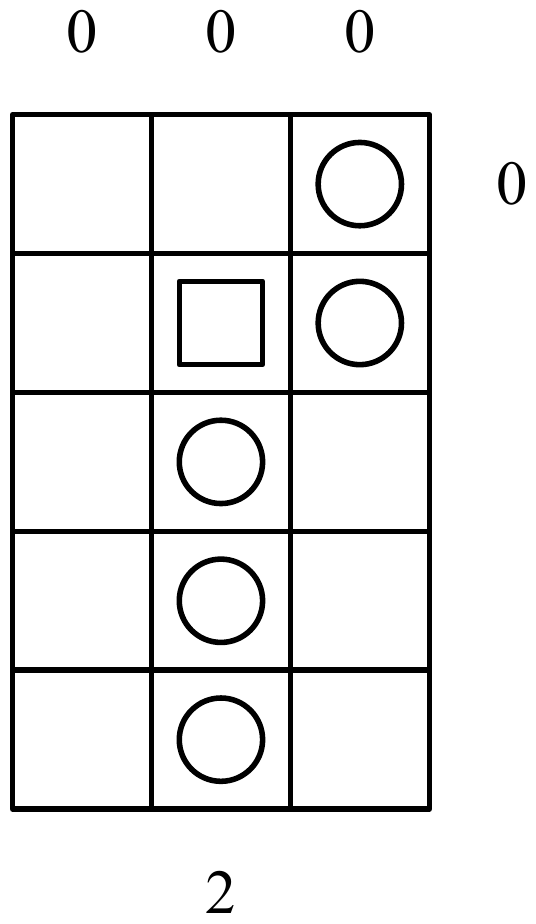}}
    \caption{The sequence of forced configurations of the clause gadget as it is traversed by the white king.}
    \label{figure:clause_gadget_1}
\end{figure}

\begin{figure}
    \centering
    \raisebox{-.5\height}{\includegraphics[scale=.25]{images/clause_3}}
    ~~\raisebox{-.5\height}{\scalebox{2}{$\to$}}~~
    \raisebox{-.5\height}{\includegraphics[scale=.25]{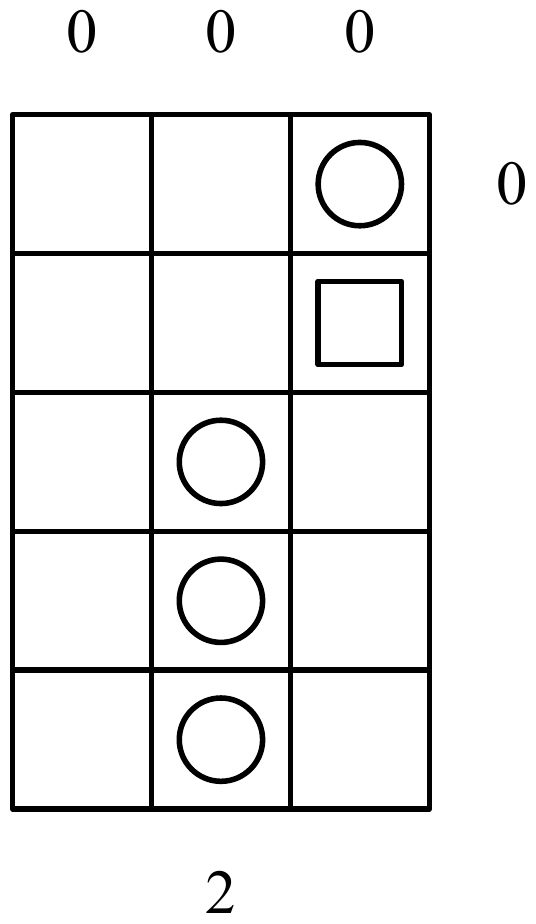}}
    ~~\raisebox{-.5\height}{\scalebox{2}{$\to$}}~~
    \raisebox{-.5\height}{\includegraphics[scale=.25]{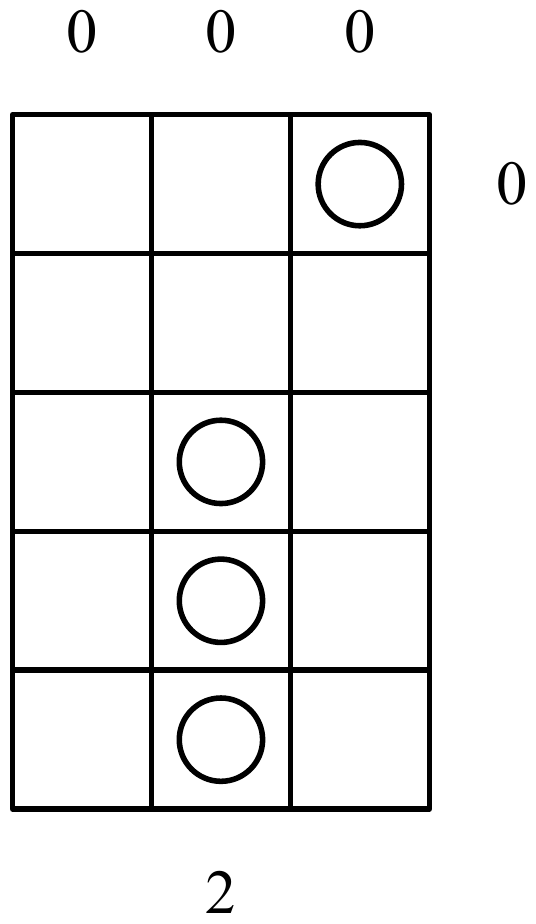}}
    \caption{Starting from the last position in Figure~\ref{figure:clause_gadget_1}, the push sequence by which the white king exits the clause gadget.}
    \label{figure:clause_gadget_2}
\end{figure}
}

\subsection{Reward gadget}
\label{sec:reward-gadget}

The reward gadget, shown in Figure~\ref{figure:reward_gadget}, allows White to win if the white king reaches the gadget. The black pawn in this gadget cannot move because it is surrounded. When instantiated in the reduction, the white king enters the gadget from the left in the top row, preceded by a white pawn. The resulting position and the following sequence of forced moves are shown in Figure~\ref{figure:reward_gadget_1}. From the final position in Figure~\ref{figure:reward_gadget_1}, White can win by moving a white pawn and the white king, then pushing upwards to push the black pawn off the board. (Recall that the move-wasting gadget no longer binds White once White can win in one turn; Black loses before Black can win using the move-wasting gadget.)


\later{
\begin{figure}
    \centering
    \raisebox{-.5\height}{\includegraphics[scale=.25]{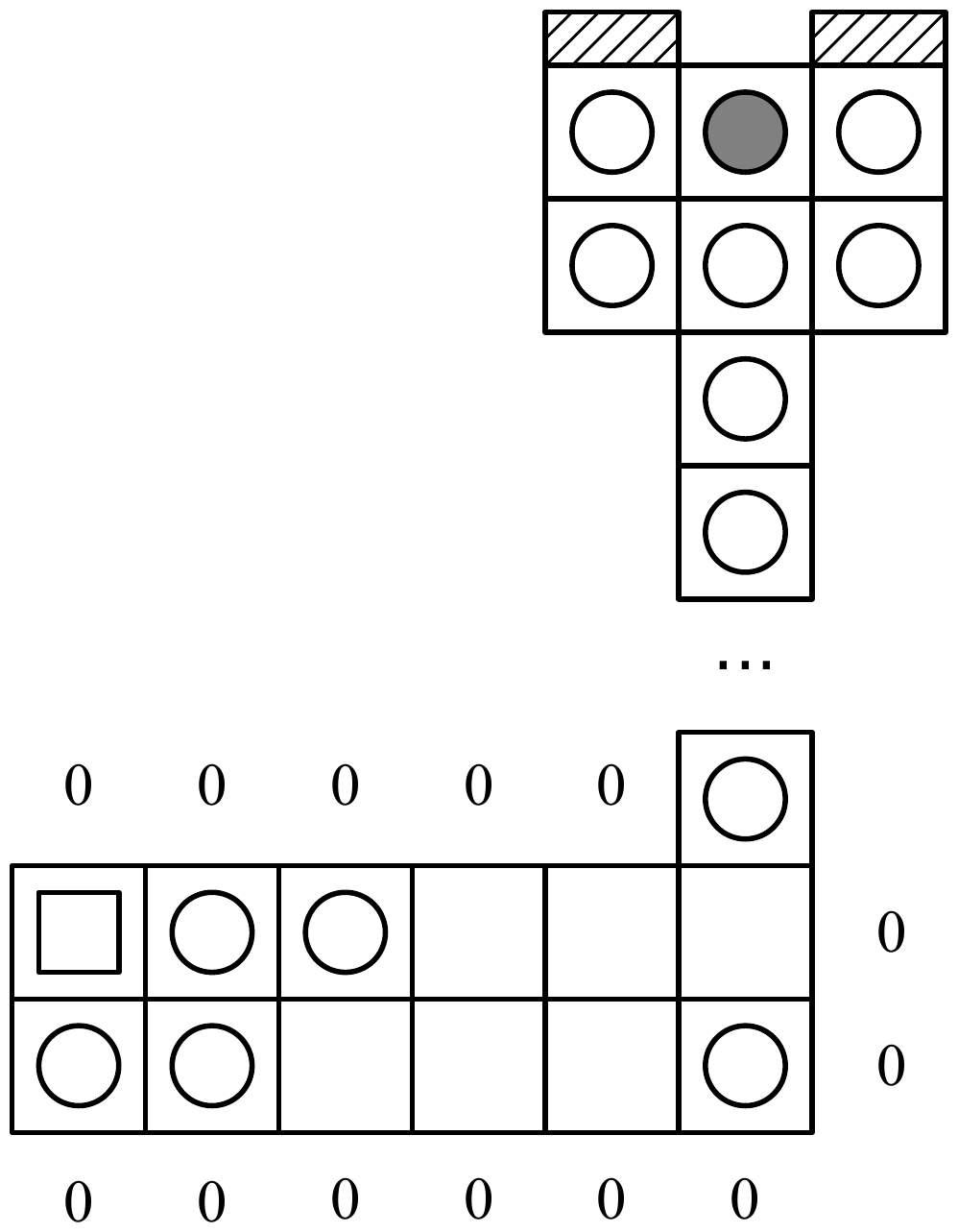}}
    ~~\raisebox{-.5\height}{\scalebox{2}{$\to$}}~~
    \raisebox{-.5\height}{\includegraphics[scale=.25]{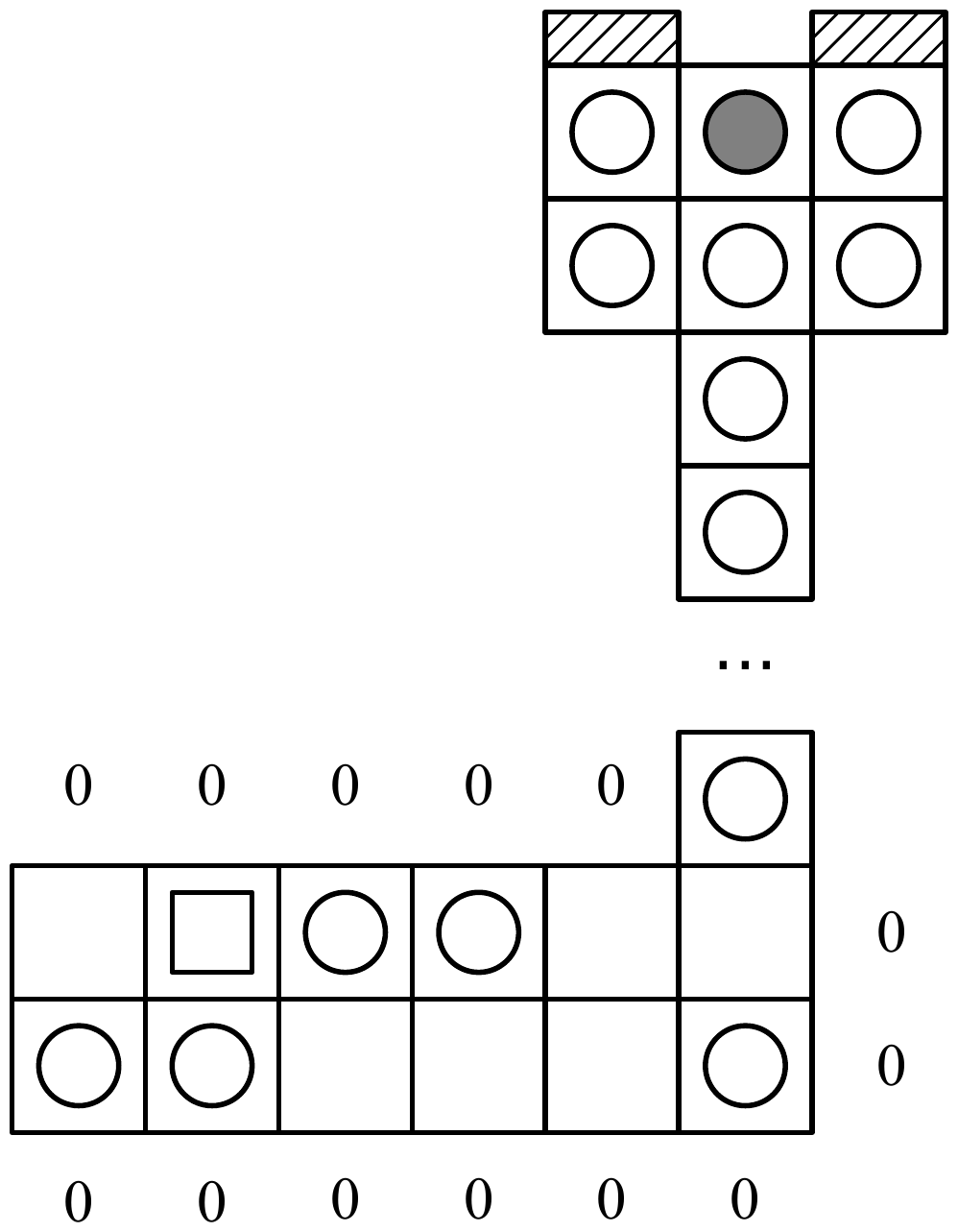}}
    ~~\raisebox{-.5\height}{\scalebox{2}{$\to$}}~~
    \raisebox{-.5\height}{\includegraphics[scale=.25]{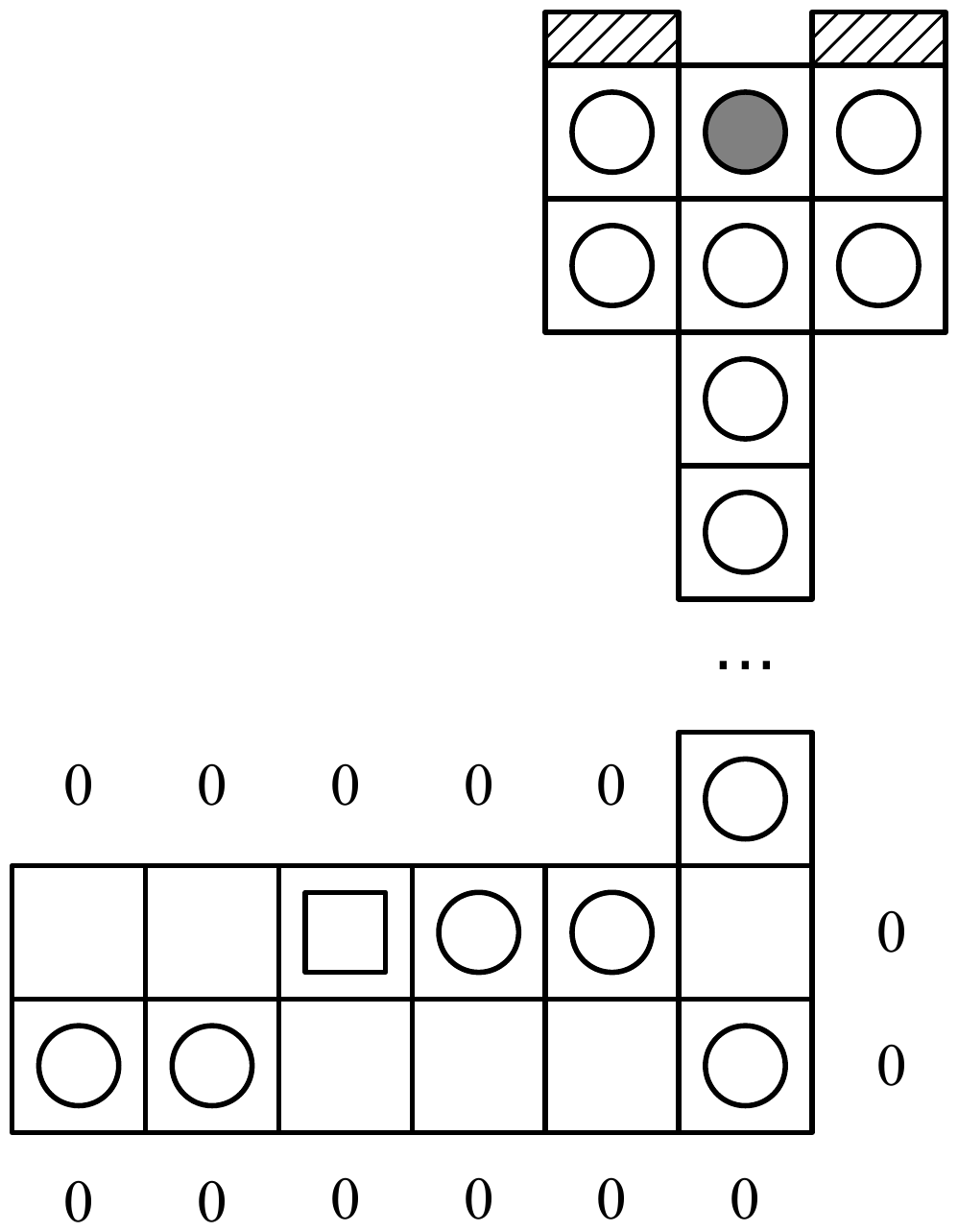}}
    \caption{The sequence of forced configurations of the reward gadget as it is traversed by the white king.}
    \label{figure:reward_gadget_1}
\end{figure}

\begin{figure}
    \centering
    \raisebox{-.5\height}{\includegraphics[scale=.25]{images/reward_3}}
    ~~\raisebox{-.5\height}{\scalebox{2}{$\to$}}~~
    \raisebox{-.5\height}{\includegraphics[scale=.25]{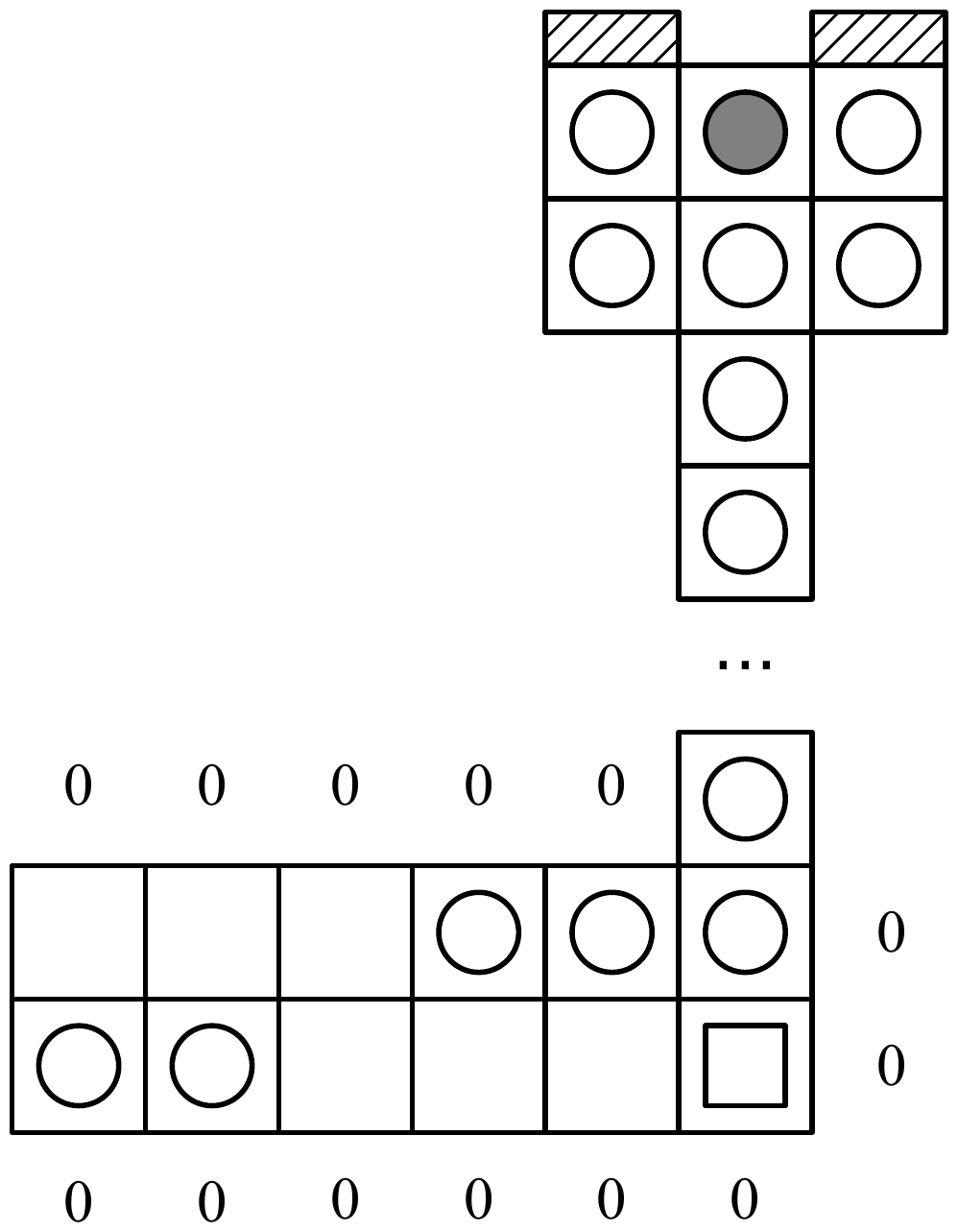}}
    ~~\raisebox{-.5\height}{\scalebox{2}{$\to$}}~~
    \raisebox{-.5\height}{\includegraphics[scale=.25]{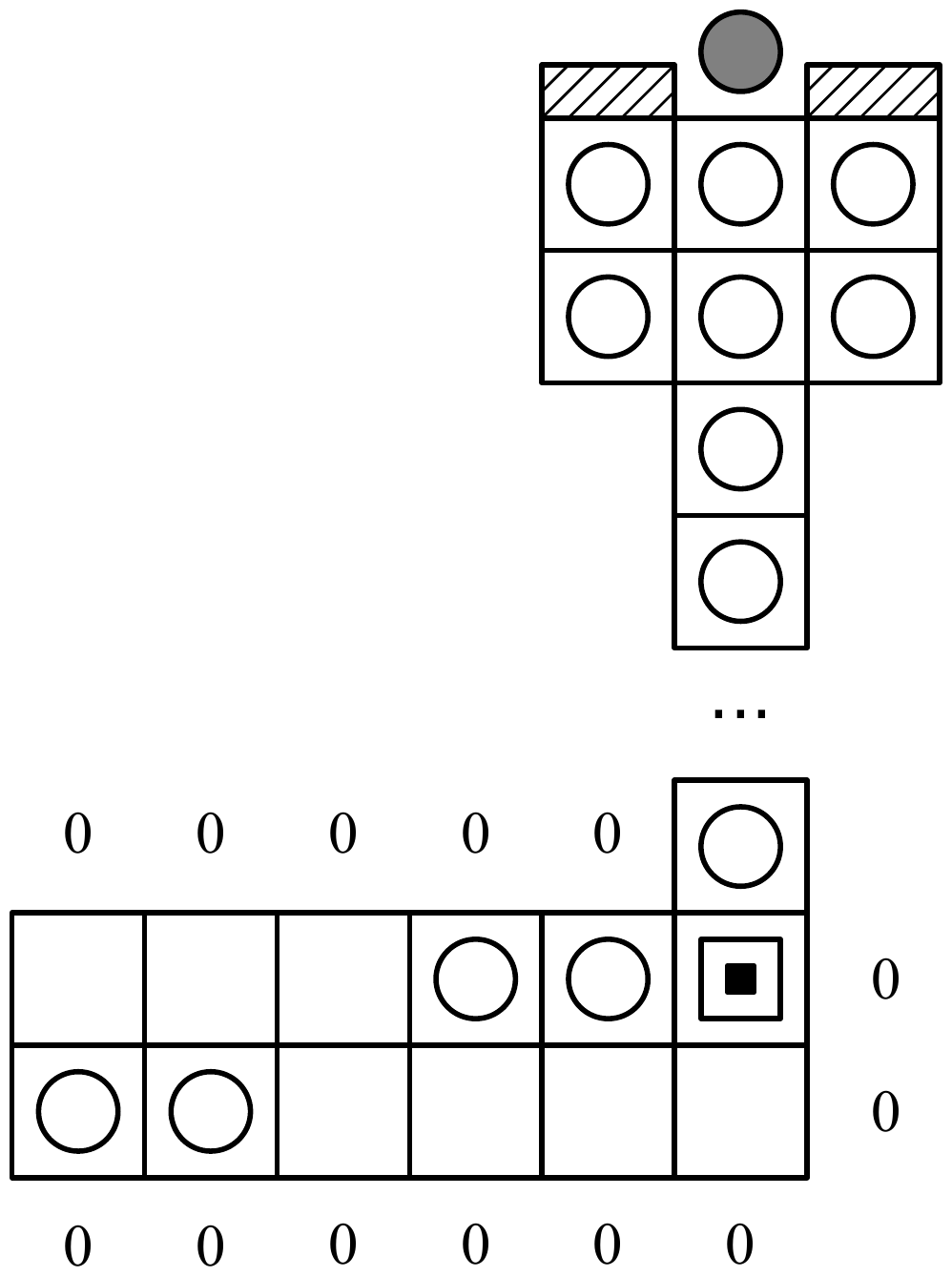}}
    \caption{White can win in a single turn from the final configuration of Figure~\ref{figure:reward_gadget_1}.}
    \label{figure:reward_gadget_2}
\end{figure}

\subsection{Layout}

Having described the gadgets, it remains to show how to instantiate them in a \pushfight{} game state for a given quantified 3-CNF formula. We first place gadgets with respect to each other, remembering which squares should be left empty, then define the board as the bounding box of the gadgets and fill any squares not recorded as empty with white pawns. The resulting board is mostly rectangular with side rails on all boundary edges, with two exceptions: one edge along the top of the rectangle lacks a side rail as part of the reward gadget, and the board is extended in the bottom-right to accomodate the move-wasting gadget along the bottom of the board.

\begin{figure}
    \centering
    \includegraphics[scale=.5]{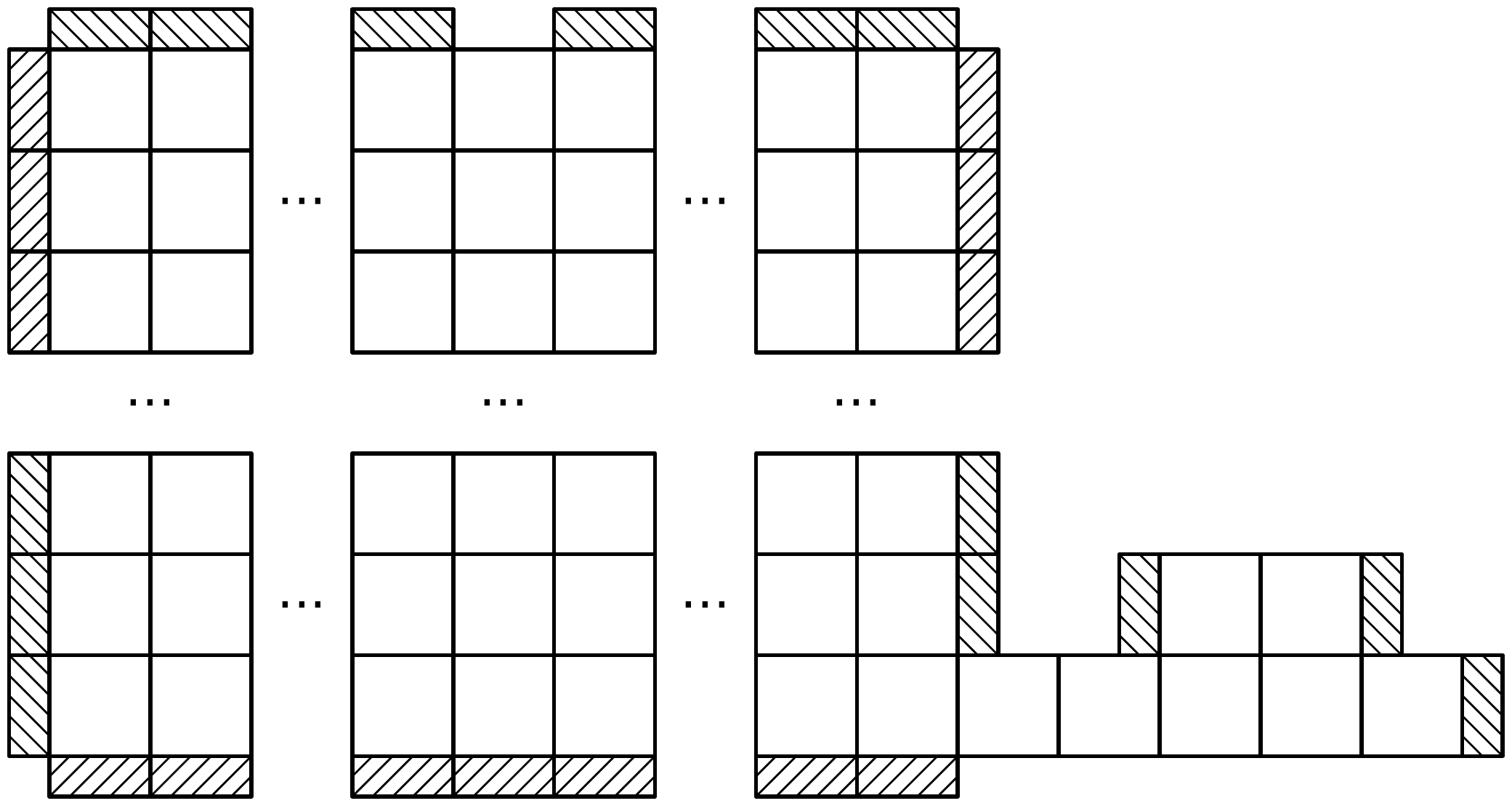}
    \caption{The shape of the \pushfight{} board produced by the reduction.}
    \label{figure:reduction_shape}
\end{figure}

We begin by building the \emph{variable gadget I} block containing the existential variable gadgets and the left portion of the universal variable gadgets. Gadgets are stacked from bottom to top in the order of the quantifiers in the input formula (using the gadget corresponding to the quantifier), with the leftmost column of each gadget aligned with the second-to-right column of the previous gadget. (Recall that the width of the variable gadgets is defined based on $p$, one more than the maximum number of occurrences of a literal in the input formula.)  This alignment allows (and requires) the white king to traverse the gadgets in sequence as specified by Lemma~\ref{thm:core}. Figure~\ref{figure:variables_plan} shows the relative layout of these variable gadgets. \xxx{These two layout figures should be side-by-side/subfigures and roughly comparable in scale.}

\begin{figure}
    \centering
    \includegraphics[scale=.5]{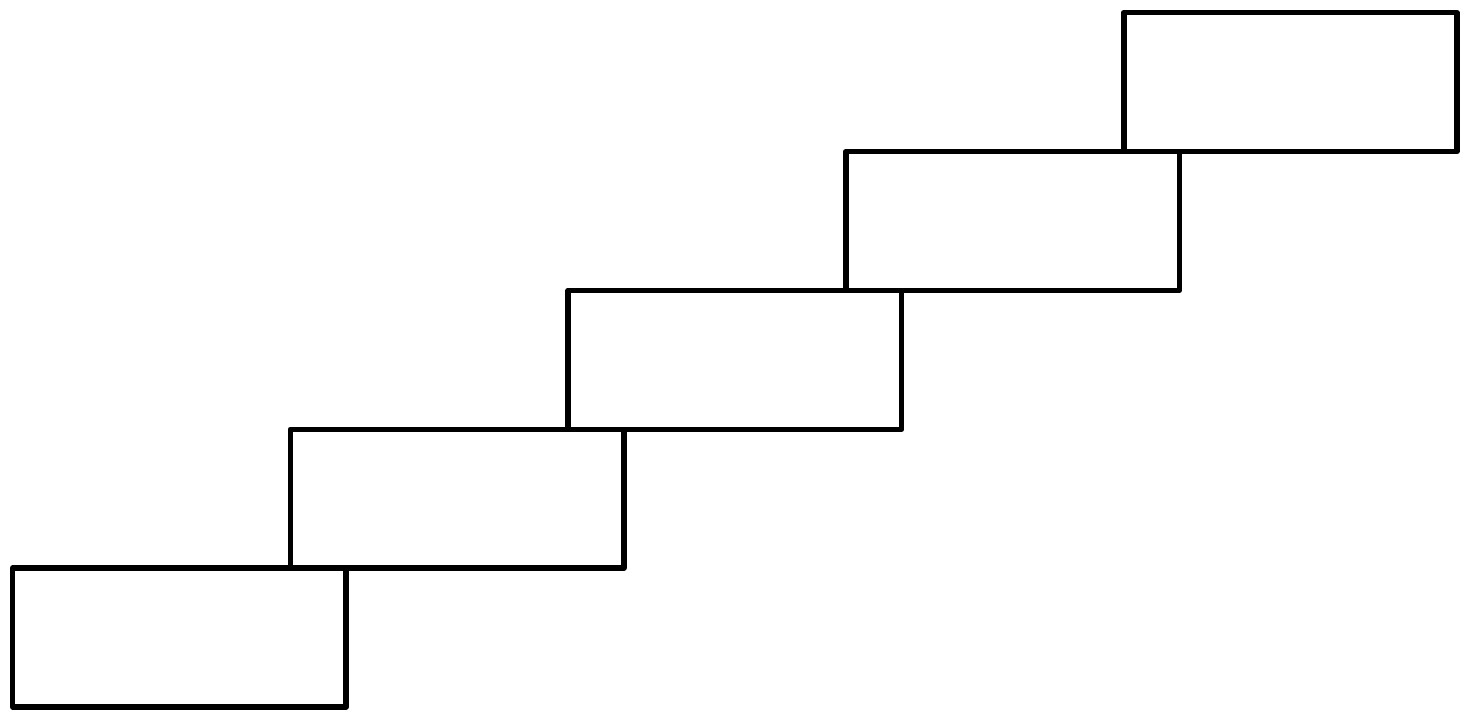}
    \caption{The layout of variable gadgets in the \emph{variable gadget I} block.}
    \label{figure:variables_plan}
\end{figure}

We place the white king one square below the first variable gadget aligned with its leftmost column, and place a white pawn one square above the white king. The white king will push upwards into the first gadget on White's first turn. (If the king was instead placed directly in the variable gadget, if the first variable is universally quantified, Black would not have a move with which to choose the value of the variable before White commits it.)

We then build the \emph{variable gadget II} block by placing the right regions of the universal variable gadgets to the right of the corresponding left regions in a single column (further right than any part of the variable gadget I section).

Next we place one clause gadget for each clause in the input formula. Each clause gadget is directly to the right of and one square lower than the previous clause gadget. The entire clause gadget block is further right of and above the \emph{variable gadget II} block. Figure~\ref{figure:clauses_plan} shows the relative layout of the clause gadgets. Then we place a bridge gadget such that the entrance of the bridge gadget aligns with the exit of the last variable gadget and the exit of the bridge gadget aligns with the entrance of the first clause. \xxx{or just ``Then we place a bridge gadget connecting the last variable gadget to the first clause gadget.''}

\begin{figure}
    \centering
    \includegraphics[scale=.5]{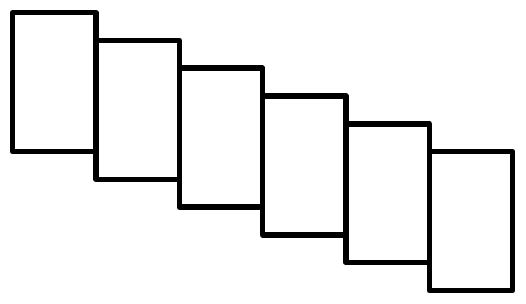}
    \caption{The layout of clause gadgets in the clause gadget block.}
    \label{figure:clauses_plan}
\end{figure}

We place the reward gadget so that its entrance aligns with the exit of the last clause gadget.

We leave empty squares in the connection block to encode the literals in each clause in the input formula. When traversing each variable gadget, the white king pushes pawns to the right in one of two rows. The lower (upper) row corresponds to setting the variable to true (false), or equivalently, preventing negative (positive) literals from satisfying clauses. Associate each row with the literal it prevents from satisfying clauses. Each clause gadget enforces that at least one empty square remains below its middle column, corresponding to at least one of its literals not having been ruled out by the truth assignment. To realize this relation, for each literal in a clause, we leave an empty square at the intersection of the column checked by the clause gadget and the row associated with that literal. All other squares in the connection block are filled with white pawns (as are all squares in the board whose contents are not otherwise specified).

The variable gadgets require each row associated with a literal to contain exactly $p-1$, $p$ or $p+1$ empty squares (depending on the type of gadget and whether the row is the upper or lower row). This is at least the number of occurrences of that literal (by the definition of $p$), but it may be greater. We place any remaining empty squares in each row in columns further right than the reward gadget, forming the overflow block.

The boundary of the board is the bounding box of all the gadgets placed thus far with a move-wasting gadget appended to the bottom of the board. The left column of the move-wasting gadget is aligned with the leftmost column of the first (leftmost) variable gadget and the sixth-from-right column (the rightmost column having height 3) is aligned with the rightmost column of the overflow block. We then fill all squares not part of a gadget nor recorded as empty with white pawns and place side rails on all boundary edges except as described in the move-wasting and reward gadgets. The anchor is on the black king as part of the initial state of the move-wasting gadget.

\subsection{Analysis}

Our analysis of gadget behavior in the preceding sections constrains the white king's pushes under the assumption that there are a specific number of empty spaces (often 0) in a particular row or column on a side of the gadget. We have already discharged the assumptions regarding the rows associated with literals by our layout of the connection and overflow blocks. For every other gadget except the variable gadgets, none of the constrained rows or columns intersects with another gadget, so the constraints on the edges are implied by the dense sea of white pawns outside the gadgets. For the variable gadgets, we assumed that pushing down in the second-to-left column of a variable gadget is not possible, but that column contains the previous variable gadget's rightmost column. We discharge this assumption by noting that in the final state of each variable gadget (after the white king has left the gadget), the rightmost column of that gadget is filled with white pawns, so pushing down in that column is indeed not possible.

Thus the white king must traverse the variable gadgets, setting the value of each variable, then traverse through the bridge gadget to the clause gadgets, where at least one empty space must remain in each checked column for the king to reach the reward gadget. If the choices made while traversing the variable gadgets results in filling all of the empty spaces in a checked column (i.e., the clause is false under the corresponding truth assignment), then White can only push by using a move outside the move-wasting gadget and Black wins on the next turn. If the white king successfully traverses every clause gadget (i.e., every clause is true under the truth assignment), then White wins when the white king pushes the black pawn off the board in the reward gadget. Thus White has a winning strategy for this \pushfight{} game state if and only if the input quantified 3-CNF formula is true.

}

%

\section*{Acknowledgments}

This work grew out of an open problem session originally started during an
MIT class on Algorithmic Lower Bounds: Fun with Hardness Proofs (6.890)
in Fall 2014.

\bibliography{pushfight}
\bibliographystyle{plain}

\clearpage
\appendix
\magicappendix

\end{document}